\documentclass[11pt]{article}
\usepackage{fullpage}

\newcommand\bigone[1]{}
\newcommand\smallone[1]{#1}
\usepackage{amsfonts,amssymb,amsmath,amsthm,latexsym}
\usepackage{array,xspace}

\newcommand{\ignore}[1]{}

\newcommand{\ii}{\mathsf{i}}
\newcommand{\ee}{\mathrm{e}}
\newcommand{\eps}{\varepsilon}

\newcommand{\polylog}{\mathop{\mathrm{polylog}}}

\newcommand{\etal}{{\em et al.}\xspace}

\def\makeletter#1{%
\expandafter \newcommand \csname b#1\endcsname {\mathbb{#1}}%
\expandafter \newcommand \csname c#1\endcsname {\mathcal{#1}}%
\expandafter \newcommand \csname t#1\endcsname {\widetilde{#1}}%
\expandafter \newcommand \csname ct#1\endcsname {\widetilde{\mathcal{#1}}}%
}
\def\makeletters(#1#2){\makeletter#1\ifx#2.\else\makeletters(#2)\fi}
\makeletters(QWERTYUIOPASDFGHJKLZXCVBNM.)

\def\makeSkob#1#2#3{%
\def\LLL{\left} \def\RRR{\right}
\expandafter \edef \csname #1\endcsname #2##1#3{\SkobInner}
\def\LLL{\bigl} \def\RRR{\bigr}
\expandafter \edef \csname #1A\endcsname #2##1#3{\SkobInner}
\def\LLL{\Bigl} \def\RRR{\Bigr}
\expandafter \edef \csname #1B\endcsname #2##1#3{\SkobInner}
\def\LLL{\biggl} \def\RRR{\biggr}
\expandafter \edef \csname #1C\endcsname #2##1#3{\SkobInner}
\def\LLL{\Biggl} \def\RRR{\Biggr}
\expandafter \edef \csname #1D\endcsname #2##1#3{\SkobInner}
\def\LLL{} \def\RRR{}
\expandafter \edef \csname #1O\endcsname #2##1#3{\SkobInner}
}

\def\SkobInner{\LLL(##1\RRR)} \makeSkob{s}[]
\def\SkobInner{\LLL[##1\RRR]} \makeSkob{sk}[]
\def\SkobInner{\LLL\lbrace##1\RRR\rbrace} \makeSkob{sfig}{}{}
\def\SkobInner{\LLL\lfloor##1\RRR\rfloor} \makeSkob{floor}[]
\def\SkobInner{\LLL\lceil##1\RRR\rceil} \makeSkob{ceil}[]
\def\SkobInner{\LLL\langle##1\RRR\rangle} \makeSkob{ip}<>
\def\SkobInner{\LLL|##1\RRR\rangle} \makeSkob{ket}|>
\def\SkobInner{\LLL|##1\RRR|} \makeSkob{abs}||
\def\SkobInner{\LLL\|##1\RRR\|} \makeSkob{norm}||
\def\SkobInner{\LLL\|##1\RRR\|_{\noexpand\mathrm F}} \makeSkob{normFrob}||
\def\SkobInner{\LLL\|##1\RRR\|_{\noexpand\mathrm{tr}}} \makeSkob{normtr}||

\newcommand{\midA}{\mathbin{\bigl|}}
\newcommand{\midB}{\mathbin{\Bigl|}}

\def \elem[#1]{[\![#1]\!]}
\def \bigfrac#1/{\left.#1\right/}
\def \bigfracR/#1.{\left/#1\right.}

\newcommand{\pfstart}{\begin{proof}} 
\newcommand{\pfsketch}{\begin{proof}[Proof sketch]}
\newcommand{\pfend}{\end{proof}} 
\newcommand{\itemstart}{\begin{itemize}\itemsep0pt}
\newcommand{\itemend}{\end{itemize}}
\newcommand{\descrstart}{\begin{description}\itemsep0pt}
\newcommand{\descrend}{\end{description}}
\newcommand{\enumstart}{\begin{enumerate}\itemsep0pt}
\newcommand{\enumend}{\end{enumerate}}

\newcommand{\CrossRef}[2]{#1~\ref{#2}}

\bigone{\newtheorem{thm}{Theorem}[chapter]}
\smallone{\newtheorem{thm}{Theorem}[section]}

\newcommand{\maketheorem}[2]{
\newtheorem{#1}[thm]{#2}
\expandafter\def \csname ref#1\endcsname ##1{\CrossRef{#2}{#1:##1}}
}

\maketheorem{lem}{Lemma}
\maketheorem{prp}{Proposition}
\maketheorem{cor}{Corollary}
\maketheorem{clm}{Claim}
\maketheorem{fact}{Fact}

\theoremstyle{definition}
\maketheorem{rem}{Remark}
\maketheorem{obs}{Observation}
\maketheorem{defn}{Definition}
\maketheorem{exm}{Example}
\maketheorem{assump}{Assumption}

\def \rf(#1:#2){\csname ref#1\endcsname{#2}}

\def\mycommand#1#2{
\expandafter\newcommand \csname#1\endcsname {#2}%
}
\def\remycommand#1#2{
\expandafter\renewcommand \csname#1\endcsname {#2}%
}

\newif\ifnotes\notesfalse

\ifnotes
\usepackage{xcolor}
\definecolor{mygrey}{gray}{0.50}
\newcommand{\notename}[2]{{\textcolor{mygrey}{\footnotesize{\bf (#1:} {#2}{\bf ) }}}}
\newcommand{\noteswarning}{{\begin{center} {\Large WARNING: NOTES ON}\end{center}}}

\else

\newcommand{\notename}[2]{{}}
\newcommand{\noteswarning}{{}}

\fi

\newcommand{\onote}[1]{{\notename{Oded}{#1}}}
\newcommand{\rnote}[1]{{\notename{Ronald}{#1}}}
\newcommand{\bnote}[1]{{\notename{Alexander}{#1}}}

\usepackage{algorithm}

\newcommand\draft[1]{}
\newcommand\release[1]{#1}
\draft{
\usepackage{easy-todo}

}
\release{
\usepackage{todonotes}

}

\draft{
\usepackage[notref,notcite]{showkeys}
}

\draft{

}
\release{
\usepackage{color}
\usepackage[leftbars, color]{changebar}
\setlength{\changebarsep}{2mm}
\setlength{\changebarwidth}{1pt}
\cbcolor{green}

}

\draft{
\usepackage{silence}
\WarningFilter{latex}{Marginpar on page}
\WarningFilter{pdftex}{}
\def\mycommand#1#2{%
\marginpar{\fbox{\sf{#1:} $#2$}}%
\expandafter\newcommand \csname#1\endcsname {#2}%
}
\def\remycommand#1#2{%
\marginpar{\fbox{\sf{#1:} $#2$}}%
\expandafter\renewcommand \csname#1\endcsname {#2}%
}
}

\draft{\newtheorem{res}{Research Problem}}
\release{}

\draft{}

\draft{
\newcommand{\texorpdfstring}[2]{#1}

\newcommand{\url}[1]{#1}
}

\release{\usepackage[breaklinks]{hyperref}}

\def\01{\{0,1\}}
\def\bool{\{0,1\}}
\def\cube{\bool^n}

\newcommand{\Inf}{\mathop{\mathrm{Inf}}\nolimits}
\newcommand{\Exp}{\mathop{\bE}}
\newcommand{\Var}{\mathop{\mathrm{Var}}}

\newcommand{\Adv}{\mathop{\mathrm{ADV}^{\pm}}}

\newcommand{\pr}{\mathop{\mathrm{Pr}}}
\newcommand{\dd}{\mathrm{d}}

\newcommand{\reg}[1]{{\mathsf{#1}}}
\def \ket#1|#2>%
{\ifx&#1&%#1 is empty
|#2\rangle
\else
|#2\rangle_{\mathsf{#1}}
\fi}

\newcommand{\algcaption}[1]{
\refstepcounter{algorithm}
\begin{tabular}{>{\bf }rp{.75\textwidth}}
Algorithm \thealgorithm& #1\\\hline
\end{tabular}
}

\newcounter{subroutine}[algorithm]
\renewcommand{\thesubroutine}{\thealgorithm.\arabic{subroutine}}
\newcommand{\sbrtncaption}[1]{
\refstepcounter{subroutine}
\begin{tabular}{>{\bf }rp{.7\textwidth}}
\hline
Subroutine \thesubroutine& #1\\\hline
\end{tabular}
}

\newcommand{\negmedskip}{\vspace{-\medskipamount}}

\begin{document}

%\title{Better Quantum Junta Testers}
%\title{Quantum Property Testers: from Group Testing to Juntas}
\title{Efficient Quantum Algorithms\\ 
for (Gapped) Group Testing and Junta Testing}
\author{Andris Ambainis\thanks{Faculty of Computing, University of Latvia.
Supported by the European Commission FET-Proactive project QALGO, ERC Advanced Grant MQC and Latvian State Research programme NexIT project No.1.}
\and Aleksandrs Belovs\thanks{
Faculty of Computing, University of Latvia.  Supported by FP7 FET Proactive project QALGO.
Part of this work was done while at CSAIL, Massachusetts Institute of Technology, USA, supported by Scott Aaronson's Alan T. Waterman Award from the National Science Foundation.}
\and Oded Regev\thanks{Courant Institute of Mathematical Sciences, New York University. 
Supported by the Simons Collaboration on Algorithms and Geometry and by the National Science Foundation (NSF) under Grant No.~CCF-1320188. Any opinions, findings, and conclusions or recommendations expressed in this material are those of the authors and do not necessarily reflect the views of the NSF.}
\and Ronald de Wolf\thanks{CWI and University of Amsterdam, the Netherlands.
Partially supported by ERC Consolidator Grant QPROGRESS and by the European Commission FET-Proactive project Quantum Algorithms (QALGO) 600700.}
}
\maketitle

%%%% DON'T REMOVE %%%%%
\noteswarning
%%%% DON'T REMOVE %%%%%

\begin{abstract}
In the $k$-junta testing problem, a tester has to efficiently decide whether a given function $f:\01^n\rightarrow\01$ is a $k$-junta (i.e., depends on at most $k$ of its input bits) or is $\eps$-far from any $k$-junta. Our main result is a quantum algorithm for this problem with query complexity $\tO(\sqrt{k/\eps})$ and time complexity $\tO(n\sqrt{k/\eps})$. This quadratically improves over the query complexity of the previous best quantum junta tester, due to At\i c\i\ and Servedio.  Our tester is based on a new quantum algorithm for a gapped version of the combinatorial group testing problem, with an up to quartic improvement over the query complexity of the best classical algorithm.
For our upper bound on the \emph{time} complexity we give a near-linear time implementation of a shallow variant of the quantum Fourier transform over the symmetric group, similar to the Schur-Weyl transform.\rnote{6/7: changed the latter}  We also prove a lower bound of $\Omega(k^{1/3})$ queries for junta-testing (for constant~$\eps$).
%\rnote{30/6: commented out the ``other problems'' sentence}
%, and obtain quantum-classical separations for several other testing problems involving the Fourier spectrum of~$f$.
\end{abstract}

\newpage

\section{Introduction}
\mycommand{ggt}{GGT\xspace}
\mycommand{eggt}{EGGT\xspace}
\mycommand{qggt}{QGGT\xspace}

\subsection{Quantum property testing}

Many computational problems are too hard to solve perfectly in any reasonable amount of time (especially if $\mathsf{P} \ne\mathsf{NP}$, as seems likely). 
Accordingly, much of theoretical as well as practical computer science is about trying to efficiently solve those problems in a weaker sense.  Examples are trying to \emph{approximate} the optimal solution, trying to solve the problem fast \emph{on average}, trying to solve it fast \emph{on most instances}, etc.  A structured model for the latter is \emph{property testing}.  
Here our goal is to test whether a given (usually very large) object $f$ has a certain property~$\cal P$.  Typically the hardest instances of the problem are the ones that are on the boundary, ``just outside'' of the property, where one needs to look at a large part of~$f$ to decide if it is in or out of the property.
But, in many cases such instances tend to appear due to noise or other imperfections, and as such should not really be rejected. 
The setting of property testing excludes such instances: it assumes that the given instance $f$ either has the property~$\cal P$, or is at least somewhat ``far'' from~$\cal P$ (i.e., far from all instances that have property~$\cal P$, according to some suitable distance measure). 
This ``promise'' on the inputs makes many hard problems much easier, and many property testers have been found over the last two decades to efficiently test properties of very large objects, see for instance~\cite{goldreich:prop}.
%\bnote{I think the main reason for this promise is that we assume that the input has some statistical noise in it in the first place.  So the question is rather: ``Is it \emph{likely} that the input possesses $\cP$?.''  If the input is far, then it is unlikely, and we reject.  On the other hand, inputs that are close should be accepted.  In most cases, this happens automatically if we accept all positive instances and use few queries.}
%\rnote{2/3: added the following sentence in response to Alexander's note}
Note that a tester even allows us to conclude something about inputs that are outside of the promise: if a tester accepts input~$f$ with high probability, then $f$ must be close to at least one element that has the property~$\cal P$.

In this paper we focus on \emph{quantum} algorithms for property testing.  These are substantially less studied than classical algorithms, but quantum property testing has been receiving increasing attention in the last few years, both for testing properties of classical objects and for testing properties of quantum objects.  See~\cite{montanaro:quantumProperyTest} for a recent survey.

\subsection{Group testing}

We first develop a new quantum algorithm for a version of the (combinatorial) \emph{group testing} problem.\footnote{To avoid confusion: while we use the established term ``group testing'' here, this is \emph{not} a property testing problem in the sense described above, because we have no ``in the property or far from the property'' promise here.} 
Group testing was invented in World War II to efficiently identify ill soldiers~\cite{dorfman:grouptesting}.  Suppose $n$ soldiers have each given samples of their blood, and up to $k$ of them are ill.  One way to identify the ill ones is to separately test each of the $n$ blood samples.  However, blood tests are expensive, and if $k\ll n$ then something much more efficient can be done.  By combining parts of the blood samples of a subset $S$ of all soldiers and testing the combined sample, we can determine whether \emph{at least one} of the soldiers in $S$ has the disease, at the expense of only one blood test.  Using binary search we can then identify one ill soldier using $O(\log n)$ tests, and all $k$ ill soldiers using $O(k\log n)$ tests.\footnote{$O(k\log n)$ tests is essentially optimal for simple information-theoretic reasons: each test gives at most one bit of information, but by identifying the $k$ ill soldiers we learn $\log{\binom{n}{k}}=\Omega(k\log(n/k))$ bits of information. There is a large literature that optimizes the constant factor and other aspects of group testing algorithms,\rnote{30/6: added:} see for instance~\cite{du:combinatorialGroupTesting}.}

Here we consider a ``gapped'' decision version of the group testing problem, which in its simplest form is the following:
\begin{quote}
{\bf Gapped group testing (GGT).}
For some set $A\subseteq[n]$, define $f_A\colon 2^{[n]}\to\01$ by setting $f_A(S)=1$ iff $S$ intersects $A$.  Given the ability to query an $f_A$ where either $|A|\leq k$ or $|A|\geq k+d$, decide which is the case.
\end{quote}
Note that the function $f_A$ is like a blood test, where $A$ is the set of all ill soldiers and the input~$S$ is the set of the soldiers whose blood we include in the tested sample. The function outputs~1 exactly when at least one of the soldiers in $S$ is ill.

\rnote{2/7: changed}
%Garc\'{i}a-Soriano~\cite[Section~5.3]{garciasoriano:phd} already studied the classical query complexity of this problem for $d=1$, showing that it lies between $\Omega(k)$ and $O(k\log k)$.
Recently, Belovs~\cite{belovs:learningSymmetricJuntas} showed that if $|A|\le k$, then one can \emph{identify} $A$ with $O(\sqrt{k})$ quantum queries to~$f_A$. Clearly this algorithm can be used to solve the \ggt problem for $d=1$.
The randomized query complexity of this problem for $d=1$ is $\widetilde\Theta(k)$ (see \rf(sec:ggtrandcomplexity) for references and proofs about the randomized case), so we have a quadratic quantum improvement over classical.

Things get more interesting as $d$ grows.  In \rf(sec:groupTesting) we show a tight bound on the quantum query complexity of \ggt of $\Theta(\sqrt{k/d})$ for all $d\le k$, while the randomized complexity of this problem is $\widetilde\Theta(k)$ for $d\le\sqrt{k}$ and $\widetilde\Theta((k/d)^2)$ for $d\geq\sqrt{k}$.  We subsequently use this quantum algorithm as a subroutine for our junta testing algorithm, but we feel the \ggt problem is quite interesting by itself as well, and may find applications elsewhere. 
%Below we will see for instance that a modification of our algorithm solves the problem of testing the support size of a distribution.\rnote{2/7: remove this sentence?} 
We now mention several other reasons why our \ggt algorithm is interesting.

\paragraph{Fourth-power improvement.}
By considering our bounds on the complexity of the \ggt problem, we see that there is a \emph{quartic} (fourth-power) quantum improvement in query complexity for the regime $\sqrt{k}\le d\le k$.
Most speed-ups obtained by quantum algorithms are either exponential (mostly, for computational problems from algebra or number theory) or quadratic or less (for algorithms based on Grover's algorithm or its generalizations). 
In contrast, our algorithm provides a fourth-power speedup, which is quite surprising given that it is based on the OR function, for which the best-possible quantum speedup is only quadratic.

Only a few other examples of such speed-ups are known: a cubic speed-up for exponential congruences by van Dam and Shparlinski~\cite{vanDam:exponentialCongruences}, a quartic speed-up for finding counterfeit coins by Iwama~\etal~\cite{iwama:quantumCounterfeit}, 
a cubic improvement for monotonicity testing on the hypercube by Belovs and Blais~\cite{belovs&blais:monotonicity} (though that work was done after the work presented here),
and a quartic improvement for learning an ``exact-half junta'' by Belovs~\cite{belovs:learningSymmetricJuntas}.  
Among these problems, a matching lower bound 
is known only for the last one.

\rnote{30/6: added} Very recently, several new separations were found for \emph{total} Boolean functions: a quartic speed-up of bounded-error quantum algorithms over \emph{deterministic} classical algorithms~\cite{abblss:separations}, and a 2.5th-power speed-up of bounded-error quantum algorithms over bounded-error classical algorithms~\cite{bendavid:supergrover}. \bnote{05.07: Removed the end of the paragraph as irrelevant}\rnote{30/6: I put it back and added a few words}  
%Both are surprising because they are super-quadratic, but neither is a quartic speed-up of bounded-error quantum algorithms over bounded-error classical algorithms (unlike our result here).

\bnote{03.07: I tend to agree to one of the reviewers.  We say that this is ``unusual'', and then mention a half dozen of examples where it has been attained.  Removed ``unusual'' altogether.}\rnote{I guess it was unusual when we did it 1.5 years ago; serves us right for taking so long to finish this :) }

\paragraph{Direct use of adversary bound.}
We construct the algorithm from a feasible solution to the (dual) semidefinite program for the \emph{adversary bound}.
Reichardt \etal~\cite{hoyer:advNegative,reichardt:spanPrograms,lee:stateConversion} have shown that the adversary bound
characterizes quantum query complexity up to a constant factor.
This means that, in principle, any quantum query algorithm can be derived as a solution to this semidefinite program. 
However, the number of new quantum algorithms that have actually been obtained in this way (without any intermediate framework between the semidefinite program and the algorithm) remains fairly small.
Moreover, the majority of these algorithms are developed in the \emph{learning graph} methodology~\cite{belovs:learning}, while our algorithm is based on different ideas, borrowing from~\cite{belovs:learningSymmetricJuntas}.

\paragraph{Robustness.}
%\bnote{Do we change this caption?}\rnote{1/7: I'm OK with the current caption}
Our algorithm still works if the action of the input function $f_A(S)$ is not defined for some values of $S$.
For instance, if $|A|=k$ and $S$ intersects~$A$, the value of $f_A(S)$ can be anything: 0, 1, or even undefined.  The same is true if $|A|=k+d$ and $S\cap A = \emptyset$.  We say that such $S$ are \emph{irrelevant variables}.
This property turns out to be very useful in applications of the \ggt algorithm, in particular when we use it as a subroutine in our junta tester.

As far as we know, such a property of a quantum algorithm was not previously studied or even explicitly defined.  
%\bnote{Removed mentioning of amplitude amplification.}
%We can only mention that the quantum amplitude amplification algorithm was implicitly used to evaluate the OR function in this regime.\rnote{15/5: where?}\onote{5/18 I second Ronald that a reference is needed here. Also, I think what we want to say is slightly different. Suggested formulation: ``The only example we are aware of is the quantum amplitude amplification algorithm, which can be seen as providing a robust evaluation of the OR function.'' But I'm not sure if that's even accurate}
However, it turns out that this property is automatically satisfied for a quantum algorithm derived from the adversary bound if the latter satisfies some additional conditions.  We obtain a similar behavior for the tight composition property of the adversary bound in the presence of irrelevant variables.

\paragraph{Time-efficient implementation.}
Our algorithm is one of the few quantum algorithms derived from the adversary bound with a \emph{time-efficient} implementation, i.e., one that is efficient in total number of gates as well as in total number of queries
%\rnote{30/6: added in response to a referee} 
(in general, the time complexity of the adversary-derived algorithm can be exponentially large in the number of input bits to the problem).
Other examples are the formula-evaluation algorithm of Reichardt and \v{S}palek~\cite{reichardt:formulae} and the algorithm for $st$-connectivity of Belovs and Reichardt~\cite{belovs:learningClaws}. 

The time complexity of our algorithm is $\tO(n\sqrt{k/d})$, roughly $n$ times its query complexity.  This is probably the best one can hope for: the oracle takes an $n$-qubit input register, so it takes $\Omega(n)$ gates just to touch all those qubits. 
Thus, any algorithm trying to beat our running time can only afford to change a small fraction of the string given to the input oracle.  Also, realistic oracles will typically take time $\Omega(n)$ to answer the query.

%\paragraph{Fast QFT over symmetric group.}
%\paragraph{Low-weight QFT over the symmetric group.}
%The above algorithms are obtained from feasible solutions to the semidefinite program for the adversary bound, which allows us to upper bound their query complexity.  However, in general algorithms that are efficient in terms of query complexity need not be efficient at all in terms of \emph{time} complexity (i.e., the number of elementary quantum gates plus oracle queries needed to implement them).  A famous example is the ``hidden subgroup problem'' on a general non-Abelian group~$G$; Ettinger et al.~\cite{ettinger:hspQuery} showed this can be solved using a number of quantum queries that is polynomial in~$\log|G|$, but the time complexity of these algorithms is exponentially worse (and we do not know any better algorithms).  Fortunately, the feasible solutions for the adversary bound that we give here are very symmetric, enabling us to implement our algorithms \emph{time-efficiently}.
%, something not done in~\cite{belovs:learningSymmetricJuntas}.

\bnote{05.07: decapitated this section and re-wrote the ending}
The key to our time-efficient algorithm is an efficient, $\tO(n)$-time, implementation of the quantum Fourier transform (QFT) on the linear space which we denote by~$M^n$.
It is of dimension $2^n$ and has an orthonormal basis indexed by the set of all subsets of $[n]$.
The symmetric group $\bS_n$ acts naturally on this space by permuting its basis elements, hence $M^n$ can be considered as an $\bS_n$-module (a representation of $\bS_n$).

Most of the previous work in this direction focused on the \emph{regular} representation of $\bS_n$, called the QFT over the symmetric group.  The most efficient implementation of this kind is due to Kawano and Sekigawa~\cite{KawanoS14} and can be implemented in depth $\tO(n^3)$, improving over~\cite{Beals97,mrr:genericqft,KawanoS13}.
Our implementation, \onote{was: on contrary}on the other hand, is close to the efficient quantum Schur-Weyl transform of Bacon, Chuang and Harrow~\cite{bch:prl04,bch:soda07}, though their algorithm is defined for another group, namely a product of a general linear group and a symmetric group.  
\onote{was: Up to}To the best of our knowledge, this is the first ``algorithmic'' application of this transformation (\cite{bch:soda07} lists a number of applications of this transformation for quantum protocols).
\rnote{6/7: added}In order to ensure that the transform has all properties we need and to make our paper self-contained, we describe our construction in full detail instead of just referring to (slight modifications of) the construction of~\cite{bch:prl04,bch:soda07}.

\subsection{Junta testing}

\mycommand{atici}{At{\i}c{\i}\xspace}

Our main result is about \emph{junta testing}. 
Let $f:\01^n\rightarrow\01$ be a Boolean function, and $J\subseteq[n]$ be the set of (indices of) variables on which the function depends. We say that $f$ is a \emph{$k$-junta} if $|J|\leq k$. Such functions are often studied, for instance in learning theory if most of the features are irrelevant for the concept that needs to be learned (e.g., in biology if only a few genes determine some biological property). We say that $f$ is \emph{$\eps$-far from any $k$-junta} if the normalized Hamming distance between $f$ and $g$ is at least~$\eps$ for every $k$-junta $g$ (i.e., $f$ and $g$ differ on at least $\eps 2^n$ inputs). The \emph{$k$-junta testing problem} is:
\begin{quote}
{\bf $k$-junta testing.}
Given the ability to query an $f\colon\01^n\rightarrow\01$ that is either a $k$-junta or $\eps$-far from any $k$-junta, decide which is the case.
\end{quote}
We would like to test this efficiently.  The primary measure of efficiency is the number of ``queries,'' evaluations of~$f$, which are usually the most expensive part of an algorithm.  However, we will also consider time complexity later.

Junta testing has been well-studied in the last decade, see~\cite{blais:testingJuntasSurvey} for a recent survey. Classically, the best known tester is by Blais~\cite{blais:testingJuntas} and uses $O(k\log k + k/\eps)$ queries to~$f$, quadratically improving upon an earlier tester of~\cite{fischer:testingJuntas}. The best known classical \emph{lower} bound is $\Omega(k/\eps)$~\cite{chockler:testingJuntasLower}.\footnote{For the special case of \emph{non-adaptive} junta testers (i.e., ones that choose all queries in advance, so the next query will not depend on outcomes of earlier queries), Servedio \etal~\cite{stw:adaptivityhelps} very recently showed a slightly stronger lower bound, which is bigger than the upper bound of Blais's tester for appropriate values of~$\eps$. This shows that adaptivity helps (slightly) for classical junta testers.} 

The best \emph{quantum} tester, due to \atici and Servedio~\cite{atici:testingJuntas}, uses $O(k/\eps)$ queries. It is based on Fourier sampling. 
This quantum tester is better than Blais's classical tester by a $\log k$-factor (for constant~$\eps$), but does not beat the best known classical lower bound, leaving open the possibility of an equally efficient classical tester.
%\footnote{Montanaro and de Wolf~\cite[Section~2.1.3]{montanaro:quantumProperyTest} observed that the upper bound of~\cite{atici:testingJuntas} can be improved to $O(k/\sqrt{\eps})$ using amplitude amplification (see~\cite{chakraborty:quantumTestLinearity} for a similar improvement). This beats the classical lower bound for $\eps=o(1)$, but not for constant~$\eps$. } 

Our main result in this paper is a quantum tester with query complexity $\tO(\sqrt{k/\eps})$, which (up to logarithmic factors) quadratically improves over the previous best quantum junta tester and actually beats the known classical lower bound \rnote{30/6: added}for the first time. We also give a time-efficient implementation.
\medskip

\noindent
{\bf Main theorem (informal).}
There is a quantum $k$-junta tester that uses $\tO(\sqrt{k/\eps})$ queries and $\tO(n\sqrt{k/\eps})$ time (i.e., elementary quantum gates and query gates).

\medskip

Similarly to the \ggt problem, this time complexity is the best that one could reasonably expect given our query complexity, because each query to~$f$ involves an $n$-qubit input register.

Our junta tester is described in~\rf(sec:qalgs). The idea is the following (suppressing the dependence on~$\eps$ for simplicity). If $f$ is far from any $k$-junta, then it depends on some $K>k$ variables, and together those $K-k$ ``extra'' variables will have at least $\eps$ ``influence''\rnote{30/6: added:} (this will be quantified using the Fourier coefficients of~$f$). We divide those extra variables in sets depending on their influence, and show that one of the following cases holds: 
\begin{enumerate}
\item \label{point:odin}
For some integer $j\in\{0,\ldots,O(\log k)\}$ there are $d\geq k/2^j$ extra variables, each of influence at least $2^j/k$. We thus need to distinguish the case where there are at most $k$ influential variables from the case where there are at least $k+d$ variables, each of influence roughly $2^j/k$. Our quantum group tester, combined with a procedure to detect variables of influence roughly $2^j/k$, can distinguish these two cases using $O(\sqrt{k})$ queries.\footnote{In the setting of classical testers, Garc\'{i}a-Soriano~\cite[p.~111]{garciasoriano:phd} also noted ``a striking resemblance between group testing and junta testing.''}
\item There are many ($\gg k$) variables of very low influence ($<1/k$). In this case a random subset~$V$ of $1/k$ of all variables has influence $\Omega(1/k)$ with probability close to~1.  In contrast, if $f$ is a $k$-junta then with significant probability $V$ will not contain any relevant variable, and hence have 0 influence. We can distinguish those two cases using $O(\sqrt{k})$ queries. 
\end{enumerate}
We then put together the testers for these special cases in order to get an $\tO(\sqrt{k/\eps})$-query tester that covers all cases. 

Let us now briefly mention the lower bounds we obtain.
As already noted by \atici and Servedio~\cite{atici:testingJuntas} and explained in \rf(sec:lower), the classical lower bound approach for junta testing fails for quantum algorithms, because the corresponding instances can be easily solved quantumly in $O(\log k)$ queries.
Instead of this, in \rf(sec:lower) we describe a different approach using reduction from the problem of testing image size of a function.\rnote{30/6: should we keep it like this if we remove the `other results' stuff?}
\bnote{03.07: You are right.  Changed to ``image size of a function''}
This already gives a lower bound of $\Omega(k^{1/3})$ by the Aaronson-Shi lower bound for the collision problem~\cite{shi:collisionLower}.
We believe that the actual complexity of testing support size of a distribution is around $\Omega(\sqrt{k})$, but proving this seems to require techniques beyond the state of the art in quantum lower bounds.

\subsection{Some remarks on organization}
The main results of the paper---quantum algorithms for the \ggt problem and junta testing---are given in Sections~\ref{sec:groupTesting}--\ref{sec:qggtalgorithm}.  For both problems, we give two versions of the algorithm: first, a query-efficient algorithm, and then its time-efficient implementation.
Our organization of the paper is such that a reader only interested in the query-efficient algorithms can read Sections~\ref{sec:groupTesting} and~\ref{sec:qalgs}, skipping the more technical time-efficient implementations of~\rf(sec:qggtalgorithm).

For the \ggt problem, a quantum query-efficient algorithm is given in \rf(sec:ggtquery), and a time-efficient implementation of the same algorithm is given in Sections~\ref{sec:proofOfQGGT}--\ref{sec:QFT}.

For the junta testing problem, we actually give two quantum algorithms: first, a query algorithm with complexity $O(\sqrt{k/\eps}\log k)$ in \rf(thm:juntaMain), and then an algorithm with slightly worse query complexity $\tO(\sqrt{k/\eps})$ and time complexity $\tO(n\sqrt{k/\eps})$ in \rf(thm:juntaTimeEfficient).

%\bnote{Removed the end of this section.}

%The junta testing algorithm is described in \rf(sec:qalgs), and all technical difficulties of its time-efficient implementation are contained in the \ggt algorithm.

%The difference between the two algorithms for junta testing lies in the composition of the \ggt procedure and the procedure for detecting influential variables, mentioned in Point~\ref{point:odin} on page~\pageref{point:odin}.  
%For the query-efficient algorithm, we first compose the dual adversary bounds for both procedures, and then apply a general result that transforms the resulting dual adversary into a quantum query algorithm.  Here we use our new result for robust composition with irrelevant variables, \rf(prp:irrelevant)(b), proven in \rf(app:irrelevant).
%
%For the time-efficient algorithm, we first transform the adversary bound for the \ggt problem into a quantum algorithm, and then compose it (as a quantum algorithm) with the procedure for detecting influential variables.  We have to prove not only that our \ggt algorithm has low time complexity, but also that it works if some input variables are ill-defined.
%To show the latter, we have to go through the implementation of the algorithm in \rf(sec:proofOfQGGT), and verify that it still works in these robust settings.
%
%\rf(sec:other) contains the complementary results mentioned in \rf(sec:otherresults).  \rf(sec:lower), which contains the proof of a lower bound for quantum junta testing, can be read independently.

\section{Preliminaries}
\label{sec:prelim}

We use $[n]$ to denote the set $\{1,2,\ldots,n\}$, and $2^A$ to denote the set of subsets of $A$.  A $k$-subset is a subset of size $k$.
All matrices in this paper have real entries.
If $A$ is a matrix, $A\elem[i,j]$ denotes the element at row $i$ and column~$j$.
A \emph{projector} always stands for an orthogonal projector.
We use $\Pi_S$ to denote the projector onto a subspace $S$.
We use $\log$ and $\ln$ to denote logarithms in base~2 and~$\ee$, respectively.
The notation $\bF_q$ stands for a finite field with $q$ elements. 

We assume familiarity with basic probability theory.  
Let $\cB(k,p)$ denote the binomial distribution: $\Pr[\cB(k,p)=i] = \binom ki p^i(1-p)^{k-i}$.
We use $\cH_n(k,m)$ to denote the hypergeometric distribution, i.e., the distribution of $\absA|A\cap [k]|$ when $A$ is sampled from all $m$-subsets of $[n]$ uniformly at random.
By $X\sim \cB$, we denote that $X$ is sampled from probability distribution $\cB$.

%\bnote{todo} Our main algorithmic tools will be Fourier analysis, amplitude amplification, and the adversary method. 

\subsection{Quantum algorithms}
\label{sec:quantumAlgorithms}

Let us define quantum query algorithms.
For a more complete treatment see~\cite{buhrman:querySurvey}.
A quantum query algorithm is defined as a sequence of unitary transformations alternating with oracle calls:
\begin{equation}
\label{eqn:queryAlgDef}
U_0\to O_x\to U_1\to O_x \to \cdots \to U_{T-1} \to O_x\to U_T.
\end{equation}
Here the $U_i$s are arbitrary unitary transformations that are independent of the input.  The input oracle $O_x$ is the same throughout the algorithm, and is the only way the algorithm accesses the input string $x=(x_j)$.
The input oracle decomposes in the following way:
\begin{equation}
\label{eqn:OxDef}
O_x = \bigoplus\nolimits_{j\in[n]} O_{x,j},
\end{equation}
where $O_{x,j}$ is some unitary transformation that only depends on the symbol $x_j$.
In this paper $x$ will be a Boolean string, and we adopt the following convention: $O_{x,j} = I$ if $x_j=0$, and $O_{x,j}=-I$ if $x_j=1$, where $I$ is the identity operator.

The computation starts in a predefined state $\ket|0>$.  After all the operations in~\rf(eqn:queryAlgDef) are performed, some predefined output register is measured.  We say that the algorithm \emph{computes} a function $F$ (with bounded error) if, for any $x$ in the domain, the result of the measurement is $F(x)$ with probability at least~$2/3$.  The number $T$ is the \emph{query complexity} of the algorithm.  The smallest value of $T$ among all algorithms computing $f$ is the quantum query complexity of $F$, and is denoted by $Q(F)$.

We will also be interested in \emph{time complexity} (also known as \emph{gate complexity}) of the algorithm.  It is defined as the total number of elementary quantum gates (from some fixed universal set of gates) required to implement all the unitary transformations $U_0,\ldots,U_T$.

One of our main algorithmic tools is amplitude amplification.  This is encapsulated in the following result of Brassard~\etal~\cite[Section~2]{brassard:amplification}, which generalizes Grover's quantum search algorithm~\cite{grover:search}.

\begin{lem}[Amplitude amplification]
\label{lem:amplampl}
Let $\cal A$ be some quantum procedure and $S$ some set of basis states on the algorithm's output space. Suppose that the probability that measuring the state ${\cal A}\ket|0>$ gives a basis state in~$S$ is at least $p$. Then there exists another procedure ${\cal B}$, which invokes ${\cal A}$ and ${\cal A}^{-1}$ $O(1/\sqrt{p})$ many times (we sometimes call such an invocation a ``round of amplitude amplification''), such that the probability that measuring the state ${\cal B}\ket|0>$ gives a basis state in~$S$ is at least~$9/10$. 
\rnote{30/6: added in response to referee}
If, in contrast, the probability of obtaining a basis state in~$S$ when measuring ${\cal A}\ket|0>$ was~0, this probability will still be~0 when measuring ${\cal B}\ket|0>$.
\end{lem}

For time-efficient implementation of our algorithm, we need the following two results.
\begin{thm}[Phase Estimation~\cite{kitaev:phaseEstimation, cleve:phaseEstimation}]
\label{thm:estimation}
Assume a unitary $U$ is given as a black box.  There exists a quantum algorithm that, given an eigenvector $\ket|psi>$ of $U$ with eigenvalue $\ee^{\ii\phi}$, outputs a real number $w$ such that $|w-\phi|\le\delta$ with probability at least $9/10$.  Moreover, the algorithm uses $O(1/\delta)$ controlled applications of $U$ and $U^{-1}$ and $\frac{1}{\delta}\polylog(1/\delta)$ other elementary operations.
\end{thm}

\begin{lem}[Effective Spectral Gap Lemma~\cite{lee:stateConversion}]\label{lem:effective}
Let $\Pi_1$ and $\Pi_2$ be two orthogonal projectors in the same vector space (not necessarily \emph{pairwise} orthogonal), and $R_1 = 2\Pi_1-I$ and $R_2 = 2\Pi_2-I$ be the reflections about their images.
For $\delta \ge 0$, let $P_\delta$ be the projector on the span of all eigenvectors of $R_2R_1$ that have eigenvalues $\ee^{\ii\theta}$ with $|\theta|\le \delta$.  Then, for any vector $w$ in the kernel of $\Pi_1$, we have
\[ \|P_\delta \Pi_2 w \|\le \frac{\delta}{2}\|w\|. \]
\end{lem}

\subsection{Adversary bound}
Here we describe the dual adversary bound, the main tool for the construction of our algorithms.

Let $F\colon\cD\to\bool$, with $\cD\subseteq\cube$, be a partial Boolean function.
The (dual) adversary bound, $\Adv(F)$, is defined as the optimal value of the following semi-definite optimization problem:
\begin{subequations}
\label{eqn:advOrig}
\begin{alignat}{3}
&\mbox{\rm minimize} &\quad& \max_{z\in \cD}\sum\nolimits_{j\in [n]} X_j\elem[z,z]  \label{eqn:advOrigObjective} \\
& \mbox{\rm subject to}&& \sum\nolimits_{j: x_j\ne y_j} X_j\elem[x, y] = 1 &\quad& \text{\rm for all $x,y\in\cD$ with $F(x)\ne F(y)$;} \label{eqn:advOrigCondition} \\
&&& X_j\succeq 0 && \mbox{\rm for all $j\in [n]$,} \label{eqn:advOrigSemidefinite}
\end{alignat}
\end{subequations}
where $X_j$ are $\cD\times\cD$ positive semi-definite matrices.  
Recall that $Q(F)$ denotes the bounded-error quantum query complexity of $F$.
Then, we have the following important result.
\begin{thm}[\cite{hoyer:advNegative,reichardt:advTight,lee:stateConversion}]
\label{thm:adv}
For every $F$, $Q(F) = \Theta(\Adv(F))$.
\end{thm}

Because of \rf(thm:adv), one may come up with a solution to the adversary bound instead of explicitly constructing a quantum algorithm. This is how we construct the algorithm 
in \rf(sec:groupTesting). The following ``unweighted adversary bound'' is a useful special case (and precursor) of the general adversary lower bound:

\mycommand{fin}{F^{-1}(0)}
\mycommand{fip}{F^{-1}(1)}
\begin{thm}[\cite{ambainis:adv}]\label{thm:unweightedadv}  
Suppose there is a non-empty relation $R\subseteq \fip \times\fin$ that satisfies
\itemstart
\item[(i)] for each $x \in \fip$ appearing in $R$, there are at least $m$ distinct $y\in\fin$ such that $(x,y)\in R$;
\item[(ii)] for each $y \in\fin$ appearing in $R$, there are at least $m'$ distinct $x\in\fip$ such that $(x,y)\in R$;
\item[(iii)] for each $x \in \fip$ and each $j\in[n]$, there are at most $\ell$ distinct $y \in\fin$ such that $(x,y) \in R$ and $x_j\ne y_j$;
\item[(iv)] for each $y \in \fin$ and each $j\in[n]$, there are at most $\ell'$ distinct $x \in\fip$ such that $(x,y) \in R$ and $x_j\ne y_j$;
\itemend
Then, the bounded-error quantum query complexity of $F$ is $\Omega\left(\sqrt{\frac{mm'}{\ell\ell'}}\right)$.
\end{thm}

The adversary bound is also useful for function composition.  Assume 
$F\colon\cD\to\bool$, with $\cD\subseteq\cube$, and, for any $j\in [n]$,
let $G_j$ be a partial Boolean function on $m_j$ variables.
The composed Boolean function $F\circ(G_1,\ldots,G_n)$ on $\sum_{j=1}^n m_j$ variables is defined by
\begin{equation}
\label{eqn:composition}
(x_{11},\ldots,x_{1m_1},\ldots,x_{n1},\ldots,x_{nm_n}) \mapsto
F\sA[G_1(x_{11},\ldots,x_{1m_1}),\ldots,G_n(x_{n1},\ldots,x_{nm_n})] ,
\end{equation}
where the composed function is defined on the input $(x_{11},\ldots,x_{nm_n})$ iff the values of all $G_j$ on the right-hand side of~\rf(eqn:composition) are defined, and the corresponding $n$-tuple belongs to $\cD$.

\begin{thm}[\cite{reichardt:spanPrograms}]
\label{thm:composition}
We have
\[
\Adv\sA[ F\circ(G_1,\ldots,G_n) ] \le \Adv(F)\max_{j\in[n]} \Adv(G_j).
\]
\end{thm}

In particular, this theorem together with \rf(thm:adv) implies
\begin{cor}[Tight composition result]
\label{cor:composition}
\(
Q \sA[ F\circ(G_1,\ldots,G_n) ] = O\sA[ Q(F) \max\limits_{j\in[n]} Q(G_j) ]  .
\)
\end{cor}
That is, one can compose functions without the logarithmic overhead in query complexity that arises in the standard method of composition (which would reduce the error probability of the algorithm for the internal functions to $\ll 1/n$ by taking the majority-outcome of $O(\log n)$ independent runs of the algorithm).
Notice, though, that \rf(cor:composition) talks only about the \emph{query} complexity of the resulting function, not its \emph{time} complexity.
Accordingly, even if the functions $F$ and $G_j$ can be evaluated time-efficiently, this does not imply that the algorithm for the composed function from \rf(cor:composition) can be implemented time-efficiently.  In order to get a time-efficient implementation, it is usually better to compose the algorithms for~$F$ and $G_j$ using the standard method.

\subsection{Irrelevant Variables}
\label{sec:irrelevant}
Consider the following motivating example.  
Let $G_1,\ldots,G_n$ be partial Boolean functions. 
We define the ``robust conjunction'' of the functions $G_i$, 
\begin{equation}
\label{eqn:irrelExample}
H(x) = \widetilde\bigwedge_{i\in[n]} G_i(x) \; ,
\end{equation}
as the partial Boolean function $H$ given by the following:
if $G_1(x)=\cdots=G_n(x)=1$ (in particular, $x$ is in the domain of all $G_i$), then $H(x)=1$;
if there exists an $i$ for which $G_i(x)=0$ (in particular, $x$ is in the domain of $G_i$), then $H(x)=0$;
and otherwise $H(x)$ is not defined.
Note that in the second case $x$ may lie outside the domain of $G_j$ for some (or even all) $j\ne i$.

One interpretation of this expression, which we use in \rf(sec:qalgs), is as follows.
The function $H$ is some test for $x$, and the $G_i$ are sub-tests, which check for different possibilities of how $H$ can fail.
Thus, a positive input must satisfy all the sub-tests, whereas a negative $x$ has to fail at least one sub-test $G_i$
but might give an ambiguous answer on other tests $G_j$.  

Classically, the above is a non-issue, since we can always apply the algorithm for $G_i$ on an input~$x$,
even if that input is outside the domain of $G_i$---the algorithm's output must still be either $0$ or~$1$. 
Quantumly, the situation is more delicate: strictly speaking, we cannot apply the textbook Grover search to evaluate $H$ since the oracle in the definition of a quantum query algorithm is supposed to apply either $I$ or $-I$ on each input, but a quantum algorithm for $G_i$ on an input $x$ outside its domain may apply an arbitrary unitary transformation on its entire working space.

%\onote{5/18: paragraph requires work; I can't make sense of it}
In this particular case, $H$ can be evaluated using amplitude amplification, \rf(lem:amplampl), instead of the usual Grover search.  In \rf(prp:irrelevant) below, we extend this result to the case when the conjunction is replaced by an arbitrary partial Boolean function $F$.
Additionally, we show how to generalize the tight composition result, \rf(cor:composition), to this more general setting.  In order to do this, we have to make a number of definitions.

\begin{defn}[Irrelevant variables]
\label{defn:irrelevant}
Let $F\colon\cD\to\bool$ be a partial Boolean function with the domain $\cD\subseteq\cube$.
For each input $x\in\cD$, some input variables $j\in[n]$ may be called \emph{irrelevant}, the remaining variables called \emph{relevant}.
This can be done in an arbitrary way, as long as the following \emph{consistency condition} is satisfied: for any $x,y\in\cD$ such that $F(x)\ne F(y)$, there must exist a variable $j$ relevant to both $x$ and $y$ and such that $x_j\ne y_j$.
\end{defn}

\begin{defn}[Evaluation with irrelevant variables]
\label{defn:irrEvaluate}
Evaluation of the function $F$ with irrelevant variables is defined as in \rf(sec:quantumAlgorithms), with the difference that, for an input $z\in\cD$, the input oracle may malfunction on irrelevant variables, i.e., $O_{z,j}$ in~\rf(eqn:OxDef) may be an arbitrary unitary if $j$ is irrelevant for $z$.
\end{defn}

\begin{defn}[Composition with irrelevant variables]
\label{defn:irrCompose}
The composition $F\circ(G_1,\ldots,G_n)$ with irrelevant variables is defined as in~\rf(eqn:composition) but on a larger domain.  Namely, the right-hand side of~\rf(eqn:composition) is defined iff there exists $z\in\cD$ such that $z_j = G_j(x_{j1},\ldots,x_{jm_j})$ for all relevant~$j$.  In particular, the value of $G_j(x_{j1},\ldots,x_{jm_j})$ need not be defined for irrelevant $j$.  The value of the composed function on this input is then set to $F(z)$, and does not depend on the particular choice of $z$.
\end{defn}

We use the last two definitions as follows.
\rf(defn:irrCompose) is used in \rf(sec:qalgs) to get a \emph{query-efficient} algorithm for testing juntas, using \rf(cor:irrCompose) below.  Thus we save a logarithmic factor, as described after \rf(cor:composition).
\rf(defn:irrEvaluate) is used in \rf(sec:qggtalgorithm) to get a \emph{time-efficient} implementation of the algorithm from \rf(sec:qalgs) (again, see the discussion after \rf(cor:composition)).

The following proposition is a special case of the construction in~\cite{belovs:variations}.

\begin{prp}
\label{prp:irrelevant}
Let $(X_j)$ be a feasible solution to the adversary bound~\rf(eqn:advOrig) with objective value~$T$.  Call an input variable $j$ is \emph{irrelevant} for an input $z\in\cD$ iff $X_j\elem[z,z]=0$.  With this choice of irrelevant variables,
\itemstart
\item[(a)] 
There exists a quantum algorithm that evaluates the function $F$ in the sense of \rf(defn:irrEvaluate), using $O(T)$ queries.
\item[(b)] For arbitrary partial Boolean functions $G_j$, we have 
\[
\Adv\sA[ F\circ(G_1,\ldots,G_n) ] \le T\; \max_{j\in[n]} \Adv(G_j),
\]
where $F\circ(G_1,\ldots,G_n)$ is as in \rf(defn:irrCompose).
\itemend
\end{prp}

It is easy to see that this choice of irrelevant variables satisfies the consistency condition of \rf(defn:irrelevant).  Indeed, if $F(x)\ne F(y)$, then~\rf(eqn:advOrigCondition) implies the existence of $j$ with $X_j\elem[x,y]\ne 0$, and since $X_j\succeq 0$, both $X_j\elem[x,x]$ and $X_j\elem[y,y]$ are non-zero.

The proof of point~(a) is analogous to the proof of the upper bound of~\rf(thm:adv), and we will skip the details here.
However, we will prove in \rf(sec:proofOfQGGT) that our time-efficient implementation of the corresponding solution to the adversary bound has this property.

We give a short proof of point~(b) in \rf(app:irrelevant).
Point~(b) immediately gives the following variant of \rf(cor:composition).

\begin{cor}
\label{cor:irrCompose}
Let $(X_j)$ be a feasible solution to the adversary bound~\rf(eqn:advOrig) with objective value $T$.  We say that an input variable $j$ is irrelevant for an input $z\in\cD$ iff $X_j\elem[z,z]=0$.  For arbitrary partial Boolean functions $G_j$, we have
\[
Q \sA[ F\circ(G_1,\ldots,G_n) ] = O\sA[ T\; \max_{j\in[n]} Q(G_j) ] ,
\]
with the composition as in \rf(defn:irrCompose) with this choice of irrelevant variables.
\end{cor}

\begin{exm}[AND]
\label{exm:and}
Let us return to the example in~\rf(eqn:irrelExample).  Consider the AND function on the domain $\cD = \{z\in\cube \mid |z|\ge n-1\}$, where $|z|$ is the Hamming weight.  A feasible solution to~\rf(eqn:advOrig) for this function is given by $X_j = \psi_j\psi_j^*$, where $\psi_j\in\bR^\cD$ is given by
%\rnote{15/5: added subscript $j$ to $\psi$}\rnote{15/5: the condition in the 2nd case below was $z_j=0$ before, which I think was wrong?} \bnote{Why??  It is equivalent, and we don't have to describe what $e_j$ is.}\onote{5/18: Ronald, I think the way it was written was equivalent to what you wrote. Remember that $z$ is taken from $\cD$. But I much prefer the way you wrote it. Alexander, do you mind using Ronald's version? it's obvious what $e_j$ means, and we used it before} \bnote{Did we?  Where?  It should be $(1^n)^{\oplus j}$, if we want to be consistent.  And yes, it is obvious: $e_j$ is the $j$th element of the standard basis, and $\oplus$ is the direct sum.}
%\onote{I like how you wrote it now! Thanks. (You're right we did not use $e_j$ before; I meant to erase that comment and forgot)}
\[
\psi_j\elem[z] = 
\begin{cases}
%n^{-1/4} ,& \text{if $z = 1^n$;}\\
%n^{1/4} ,& \text{if $z_j = 0$;}\\ % z = 1^n\oplus e_j$;}\\
%0,&\text{otherwise.}
n^{-1/4} ,& \text{if $|z|=n$;}\\
n^{1/4} ,& \text{if $|z|=n-1$ and $z_j = 0$;}\\
0,&\text{if $|z|=n-1$ and $z_j=1$.}
\end{cases}
\]
The objective value of this solution is $\sqrt{n}$.
Note that if $|z|=n-1$, then any variable $j$ with $z_j=1$ is irrelevant for this input.
This coincides with our definition of the ``robust conjunction'' at the beginning of this section.
\end{exm}

\subsection{Fourier analysis}
\mycommand{sbool}{\{\pm1\}}
\mycommand{hf}{\widehat{f}}
We use Fourier analysis for arbitrary real-valued functions $f\colon \cube \to \bR$.
If $f$ is Boolean, it is usually convenient to assume that its range is $\sbool=\{1,-1\}$ rather than \bool.
For a string $s\in\cube$, the corresponding \emph{character} is a Boolean function $\chi_s\colon\cube\to\sbool$ defined by $\chi_s(x) = (-1)^{s\cdot x}$, where $s\cdot x = \sum_{j} s_jx_j$ denotes the inner product of $s$ and $x$.
We will often use the corresponding subset $S\subseteq[n]$ instead of a string $s\in\cube$.

Every function $f\colon\cube\to\bR$, has a \emph{Fourier decomposition} as follows:
\[
f(x) = \sum_{s\in\cube} \hf(s)\chi_s(x),
\]
where $\hf(s) = 2^{-n} \sum_{x} f(x)\chi_s(x)$ is the \emph{Fourier coefficient}.
The set $\sfigA{s\mid \hf(s)\ne 0}$ is called the \emph{(Fourier) spectrum} of $f$.
\emph{Parseval's identity} says that rhis transformation respects the norm:
$\Exp_{x}\sk[f(x)^2] = \sum_s \hf(s)^2$.  In particular, for a Boolean $f\colon\cube\to\sbool$, we have $\sum_{s} \hf(s)^2 = 1$.

For a subset $S\subseteq[n]$, we define the \emph{influence} of $S$ on $f$ by
\begin{equation}
\label{eqn:InfSdef}
\Inf_S(f) = \sum_{T\colon T\cap S\ne\emptyset} \hf(T)^2.
\end{equation}
If $S$ consists of a single element $j\in[n]$ we write $\Inf_j(f)$, which is $\sum_{T\colon j\in T} \hf(T)^2$.  

An alternative (but equivalent) definition of influence for functions with range $\sbool$ is as follows.
Consider the following randomized procedure.  Generate $x\in\cube$ uniformly at random. Obtain $y\in\cube$ from $x$ by replacing, for each $j\in S$, $x_j$ by an independent uniformly random bit.  Then, the influence $\Inf_S(f)$ is precisely twice the probability that $f(x)\ne f(y)$.
Note that the influence $\Inf_j(f)$ equals the probability that $f(x)\ne f(x^{\oplus j})$ when $x$ is sampled from $\cube$ uniformly at random, where $x^{\oplus j}$ denotes $x$ with the $j$th bit flipped.  We repeatedly use the following two obvious properties of influence:
\itemstart
\item Monotonicity.  If $S\subseteq T$, then $\Inf_S(f)\le \Inf_T(f)$.
\item Subadditivity.  $\Inf_{S\cup T} (f) \le \Inf_S(f) + \Inf_T(f)$ for all $S,T\subseteq[n]$.
\itemend
The following lemma from~\cite{atici:testingJuntas} (implicit in the proof of their Theorem III.3)\rnote{30/6: added the previous} explains why influence is important in our junta testing algorithm.

\begin{lem}
\label{lem:weightonextravars}
If $f$ is $\eps$-far from any $k$-junta, then for all $W\subseteq[n]$ of size $|W|\leq k$ we have 
$\Inf_{[n]\setminus W} (f) \ge \eps$.
\end{lem}

\begin{proof}
Define a (not necessarily Boolean) function $g\colon \01^n\rightarrow\mathbb{R}$ by $g(x)=\sum_{S\subseteq W}\widehat{f}(S)\chi_S(x)$. Let $h$ be the Boolean function that is the sign of~$g$. This $h$ only depends on the variables in~$W$, so it is a $k$-junta. Since $f$ is $\eps$-far from any $k$-junta, we have (using Parseval's identity)
%\begin{align*}
%\eps %& \leq d(f,h)\\
%& \leq \frac{1}{2^n}\sum_{x:f(x)\neq h(x)}1\\
%& \leq \frac{1}{2^n}\sum_{x:f(x)\neq h(x)}(f(x)-g(x))^2\\
%& \leq\Exp_x[(f(x)-g(x))^2]\\
%& =\sum_{s}(\widehat{f}(s)-\widehat{g}(s))^2\\
%& =\sum_{s\not\subseteq W}\widehat{f}(s)^2.
%\end{align*}
\[
\eps %& \leq d(f,h)\\
\leq \frac{\abs|\sfigA{x\mid f(x)\ne h(x)}|}{2^n} %\sum_{x:f(x)\neq h(x)} 1
\leq\Exp_x\skA[(f(x)-g(x))^2]
=\sum_{S}(\widehat{f}(S)-\widehat{g}(S))^2
=\sum_{S\not\subseteq W}\widehat{f}(S)^2.\qedhere
\]
\end{proof}

%\begin{rem}
%\label{rem:testInfluence}
%Note that \rf(alg:testInfluence) never returns a subset with $S\cap V=\emptyset$ if $\Inf_V(f)=0$.  This means that it can be used as a tester with 1-sided error that accepts if  $\Inf_V(f)\ge\delta$ and rejects if $\Inf_V(f)=0$.  The tester uses $O(\sqrt{1/\delta})$ queries and $O(n/\sqrt{\delta})$ other elementary operations.
%\end{rem}

\section{Gap version of group testing}
\label{sec:groupTesting}

Our junta testers, which we describe in the next section, work by generating a random subset $S\subseteq[n]$ and testing whether it intersects the set of influential variables of $f$.
In this section we study a more abstract problem, where we assume that we have an oracle that answers this intersection-question with certainty.
We believe this problem is of independent interest.
%Before presenting our improved quantum junta tester, we first construct a dual adversary for the following related problem: the group testing problem. 
%This problem is defined for functions from $2^{[n]}$ to $\{0,1\}$, so we switch from $z\in\{0,1\}^n$ to $S = \{i\in[n]\mid z_i=1\}$.  
%The functions under consideration will be defined by subsets of $[n]$, so we form $\cX$ and $\cY$ out of the corresponding subsets.

\mycommand{Inter}{\mathrm{Intersects}}
For each $A\subseteq [n]$, define the function $\Inter_A\colon 2^{[n]}\to \{0,1\}$ by
\begin{equation}
\label{eqn:fADefn}
\Inter_A(S) = 
\begin{cases}
1,&  \text{if }A\cap S \ne\emptyset; \\
0,& \mbox{otherwise.}
\end{cases}
\end{equation}

In the standard version of group testing~\cite{du:combinatorialGroupTesting} one is given oracle access to the function $\Inter_A$ for some $A$ with $|A|\le k$, and the task is to identify $A$.  
%\onote{11/1: modified text so we can mention previous work:}
A natural variant of this problem is to compute or approximate the cardinality of $A$ (a task already considered in the group testing literature~\cite{DamaschkeM10}), and a decision version of the latter is deciding whether that cardinality is $k$ or $k+d$ for some $d,k\geq 1$. 
We define this formally next.
%\onote{was: In this section, we study the following \onote{11/1: added ``natural''}natural decision versions of this problem.  We give two definitions of the problem, the second one being a relaxation of the first one.}

\begin{defn}[\eggt]
\label{defn:exactGroupTesting}
Let $k$ and $d$ be positive integers, $\cX$ consist of all subsets of $[n]$ having size exactly $k$, and $\cY$ consist of all subsets of $[n]$ having size exactly $k+d$.
In the \emph{exact gap version of the group testing (\eggt) problem} with parameters $k$ and $d$, one is given oracle access to the function $\Inter_A$ with $A\in\cX\cup\cY$, and the task is to decide whether $A\in\cX$ or $A\in\cY$.
\end{defn}

%\onote{11/1: added:}
We also study a relaxation of \eggt, in which we allow ``false negatives'' in the small-set case and ``false positives'' in the large-set case.
This will be convenient for our applications. Luckily, algorithms for solving 
\eggt often turn out to also solve this harder problem. 

\begin{defn}[\ggt]
\label{defn:GroupTesting}
Let $k$ and $d$ be positive integers.
Define two families of functions
\begin{equation}
\label{eqn:cXPrim}
\ctX = \sfig{f\colon 2^{[n]}\to\{0,1\} \midA
\exists A\in\cX\; \forall S\subseteq[n] : S\cap A = \emptyset \implies f(S)=0}
\end{equation}
and
\begin{equation}
\label{eqn:cYPrim}
\ctY = \sfig{f\colon 2^{[n]}\to\{0,1\} \midA
\exists B\in\cY \; \forall S\subseteq[n] : S\cap B \ne \emptyset \implies f(S)=1}.
\end{equation}
In the \emph{gap version of the group testing (\ggt) problem} with parameters $k$ and $d$, one is given oracle access to $f\in \ctX\cup\ctY$, and the task is to decide whether $f\in\ctX$ or $f\in\ctY$.
\end{defn}

It is easy to see that \eggt is a special case of the \ggt, where the implications in~\rf(eqn:cXPrim) and~\rf(eqn:cYPrim) are replaced by equivalences.
\bnote{03.07: Added the following:}
The \ggt problem also includes as a special case the problem of distinguishing a function $\Inter_A(f)$ with $|A|\le k$ from a function $\Inter_A(f)$ with $|A|\ge k+d$.
%\onote{11/1: added:}
%Notice also that for any $k$ and $d' \ge d \ge 1$, \ggt with parameters $k$ and $d'$ is a special case of \ggt with parameters $k$ and $d$ (i.e., increasing $d$ does not make the problem harder). 
%\onote{11/1: removing: Despite vast amount of research done on combinatorial group testing, the problems in Definitions~\ref{defn:exactGroupTesting} and~\ref{defn:GroupTesting} were not, up to our knowledge, studied previously.}

\subsection{Randomized complexity}\label{sec:ggtrandcomplexity}

In this section we show that the randomized query complexity of the gap version of group testing is $\widetilde\Theta\sA[\min\{k, 1+(k/d)^2\}]$.

The upper bound already follows from the existing literature on group testing: an $O\sA[1+(k/d)^2]$ bound appears in~\cite{ChengX14} (with a surprisingly elaborate analysis) and an $O\sA[k \log k]$ bound appears in~\cite{Cheng11} and independently in~\cite[Section~5.3]{garciasoriano:phd}. Strictly speaking those results only apply to the \eggt problem. Because of that, and also for completeness, we include in Section~\ref{sec:ggtupperbound} an upper bound for the more general \ggt problem. 

The only \emph{lower} bound for randomized complexity that we are aware of is $\Omega(k)$ for the special case $d=1$ due to Garc\'{i}a-Soriano~\cite[Section~5.3]{garciasoriano:phd}, who calls the problem ``relaxed group testing.'' We give two different lower bound proofs here for the general case of arbitrary~$d$.  The first (Section~\ref{sec:lowerboundeggt}) is a short reduction from the Gap Hamming Distance problem in communication complexity which applies only to \ggt.  The second (Section~\ref{sec:lowerboundggt}) is a substantially longer but self-contained proof which applies to \eggt as well.

\subsubsection{Upper bound}
\label{sec:ggtupperbound}

\begin{thm}
\label{thm:randomizedUpperBound}
For any $k,d \ge 1$, the randomized query complexity of the \ggt problem with parameters $k$ and $d$ is $O\sA[\min\{k \log k, 1+(k/d)^2\}]$.
\end{thm}

\pfstart
We start with the easier upper bound of $O\sA[1+(k/d)^2]$.
Take $S\subseteq[n]$ by including each element independently at random with probability $1/k$.
If we are in the ``small'' case of \ggt (as in~\eqref{eqn:cXPrim}), the probability of
$f(S)=1$ is at most $1-(1-1/k)^k$. If, on the other hand, we are in the ``large'' case of \ggt (as in~\eqref{eqn:cYPrim}) then that probability is at least $1-(1-1/k)^{k+d}$.
As these two probabilities differ by $\Omega(\min\{1,d/k\})$, by a Chernoff bound, we can distinguish the two cases by repeating this procedure $O\sA[1+(k/d)^2]$ times. 

We now prove the upper bound $O\sA[k \log k]$. The algorithm maintains a partition of $[n]$, initially set to the trivial partition $\{[n]\}$. Each set in the partition can be either active or inactive, with the initial set $[n]$ being active. We maintain the invariant that for all sets $S$ in the partition, $f(S)=1$. (We can assume $f([n])=1$ as otherwise we are clearly in the ``small'' case.) At each step of the algorithm, we take an active set $S$ in the partition and repeat the following $10\log k$ times. We partition $S$ into $S_1$ and $S_2$ by 
taking each element of $S$ independently to be in either $S_1$ or $S_2$ with probability $1/2$. We then query $f$ on both $S_1$ and $S_2$. If $f$ returns $1$ on both, then we replace $S$ with $S_1$ and $S_2$, and this ends the loop for $S$. Otherwise, if we did not manage to split $S$ after $10\log k$ attempts, we declare $S$ to be inactive and move on to another set in the partition. If at any point the partition contains at least $k+1$ sets we stop and output ``large.'' Otherwise, if all (at most $k$) sets are inactive, we stop and output ``small.'' 
  
Notice that after each step we either add a set to the partition or declare a set inactive. There can therefore be at most $2k$ steps, and since each step involves at most $10\log k$ queries, the total number of queries is at most $20 k \log k$. The correctness in the ``small'' case 
is immediate from our invariant: the only way for the algorithm to output ``large'' is 
if there are $k+1$ sets in the partition on which $f$ returns $1$ but this cannot
happen in the small case. So consider the ``large'' case as in~\eqref{eqn:cYPrim} with 
some set $B$ of size at least $k+1$. We claim that with high probability, 
all inactive sets intersect $B$ in at most $1$ element, and hence it cannot happen that there are at most $k$ sets in the partition and all are inactive. To see why, notice that if a
set $S$ intersects $B$ in at least two elements, then there is probability $1/2$ 
that when we split $S$ into $S_1$ and $S_2$, these two elements would end up in a 
different set. In this case $f$ must answer $1$ on both $S_1$ and $S_2$ and 
$S$ would be split. Therefore the probability for such an $S$ to become inactive
is at most $2^{-10\log k} = k^{-10}$, which means that with high probability
this bad event will never happen.
\pfend

\subsubsection{Lower bound for \ggt via communication complexity}
\label{sec:lowerboundeggt}

\begin{thm}
\label{thm:randomizedLowerBoundCommunication}
For any $k,d \ge 1$, the randomized query complexity of the \ggt problem with parameters $k$ and $d$ is $\widetilde\Omega\sA[\min\{k, 1+(k/d)^2\}]$.
\end{thm}
\begin{proof}
Following the general approach to classical testing lower bounds of Blais \etal~\cite{blais:testingLowerViaCommunication}, we will show a reduction from the Gap Hamming Distance (GHD) problem in communication complexity. In this problem there are two parties, Alice and Bob. Alice receives a bit string $x \in \{0,1\}^n$ and Bob receives a bit string $y \in \{0,1\}^n$.
Their goal is to decide if the Hamming weight of $x \oplus y$ (i.e., the Hamming distance between $x$ and $y$) is greater than $n/2+g$ or at most $n/2-g$ (and can behave arbitrarily for values in between). They are allowed the use of shared randomness and, for each input $(x,y)$, need to output the correct answer with probability, say, at least~$2/3$. 
It was shown in~\cite{ChakrabartiR10} (see also~\cite{Sherstov12}) that for any $1 \le g < n/2$, any protocol solving this problem
must use at least $\Omega(\min\{n,n^2/g^2\})$ bits of communication (and this is tight). 

Now let $k,d \ge 1$. It clearly suffices to prove the theorem for $d < k$, so assume that this is the case. The result now follows from the observation that any algorithm solving \ggt with parameters $k$ and $d$ making $q \le k$ queries implies a protocol for GHD with parameters $n=2k$ and $g=d$ using $O(q \log k)$ bits of communication. 
In this protocol, Alice and Bob simulate the \ggt algorithm given oracle access to $\Inter_{A}$ where $A \subseteq [n]$ is the support of $x \oplus y$, using their shared randomness as coins to the algorithm. Whenever the algorithm performs a query $S$, Alice and Bob compute $\Inter_{A}(S)$ by running an Equality protocol with error probability less than $1/(10k)$ and communication $O(\log k)$ (see~\cite[Example~3.13]{kushilevitz&nisan:cc}) to check if the restrictions of $x$ and $y$ to $S$ are identical. 
It remains to notice that if the Hamming weight of $x \oplus y$ is at most $n/2-g=k-d$, then $\Inter_{A} \in \ctX$, whereas if it is greater than $n/2+g=k+d$, then $\Inter_{A} \in \ctY$. 
\end{proof}

\subsubsection{Direct lower bound for \eggt}
\label{sec:lowerboundggt}

\begin{thm}
\label{thm:randomizedLowerBound}
For any $k,d \ge 1$, the randomized query complexity of the \eggt problem with parameters $k$ and $d$ is $\widetilde\Omega\sA[\min\{k, 1+(k/d)^2\}]$.
\end{thm}

%\onote{12/8: removing, and replacing with a lemma containing only the second item:
%\begin{lem}
%\label{lem:subsetSamplingOLD}
%Provided that $n$ is large enough,
%\itemstart
%\item[(a)] if $m\ge 16n (k+d)/d^2$ and $t=(k+d/2)\frac mn$, then 
%\[
%\Pr\skA[\cH_n(k,m)\ge t]\le \frac14\qquad\text{and}\qquad \Pr\skA[\cH_n(k+d,m)\le t]\le \frac14\;;
%\]
%\item[(b)] if $m\le n/2$ and $m\le n(k+d)/d^2$, then the total variation distance between $\cH_n(k,m)$ and $\cH_n(k+d,m)$ is $1-\Omega(1)$.
%\itemend
%\end{lem}
%}

\pfstart
\mycommand{sizelimit}{4\frac nk\ln k}
We apply Yao's principle.  Let $\cX$ and $\cY$ be the uniform probability distributions over all subsets of $[n]$ of sizes $k$ and $k+d$, respectively, and $\cD$ be the uniform mixture of $\cX$ and $\cY$.
Denote $m = \min\sfig{\floor[n/4], \floor[n(k+d)/d^2]}$.
Assume that $\cT$ is a deterministic decision tree of depth less than $m/(\sizelimit)$.
We prove that, if $n$ is large enough, the total variation distance between the outputs of $\cT$ on $\cX$ and $\cY$ is at most $1-\Omega(1)$.
As the success probability of a randomized algorithm with error probability $\leq 1/3$ can be amplified to be larger than that same $1-\Omega(1)$ in a constant number of repetitions, the theorem follows.

The main idea of the proof is this. We first show in Lemma~\ref{lem:TreeSubstitution} that queries to subsets of cardinality larger than $\sizelimit$ are going to return $1$ with very high probability. As a result, such queries are useless and we can avoid them altogether, resulting in a decision tree $\cT'$ in which all queries are to subsets of size at most $\sizelimit$. We then complete the proof by showing that since $\cT'$ looks at no more than $m$ elements of the universe $[n]$, it cannot distinguish $\cX$ and $\cY$ with high probability.

\mycommand{vertex}{\cV}
Let us introduce some notation.
Each vertex $\vertex$ of a deterministic decision tree $\cT$ is characterized by its query $S_\vertex \subseteq[n]$.  If $\vertex$ is at depth $t-1$, let $\vertex_1,\dots,\vertex_t$ be the vertices on the path from the root to $\vertex$, with $\vertex_1$ being the root and $\vertex_t = \vertex$.  
For $i\in[t-1]$, let $S_i$ be the corresponding query and $b_i\in\bool$ be the 
output of the oracle that leads from $\vertex_i$ to $\vertex_{i+1}$.
We say that the vertex $\vertex$ is \emph{used} on an input $A\subseteq [n]$ if it is visited when applying the decision tree, i.e., if and only if $\Inter_A(S_i) = b_i$ for all $i\in[t-1]$.  
Finally, the set of \emph{fresh} variables queried by a vertex~$\vertex$ consists of those
variables not queried by its ancestors, namely, $F_\vertex = S_\vertex\setminus \bigcup_{i\in[t-1]} S_i$.

\begin{lem}
\label{lem:TreeSubstitution}
For any decision tree $\cT$ of depth $T \le k/2$, there exists a decision tree $\cT'$ of the same depth such that
\itemstart
\item for any vertex $\vertex$ of $\cT'$, the set of fresh variables, $F_\vertex$, is of size at most $\sizelimit$; and
\item the probability that $\cT$ and $\cT'$ disagree on $A\sim \cD$ is at most $1/k$.
\itemend
\end{lem}

\pfstart
Consider a sequence of decision trees
\[
\cT_0 = \cT,\;\; \cT_1,\dots, \cT_{T-1},\;\; \cT_T = \cT'
\]
defined in the following way.  Assume that $\cT_{t-1}$ is already defined.  For any vertex $\vertex$ of $\cT_{t-1}$ at depth $t-1$ such that $|F_\vertex|\ge\sizelimit$, replace it by a dummy query that always returns 1.  It is clear that $\cT'$ satisfies the first requirement.
From \rf(clm:substitution) below, it follows that the probability that $\cT_{t-1}$ and $\cT_t$ disagree on $A\sim \cD$ is at most $1/k^2$.  As $T \le k/2$, it follows \rnote{30/6: added}by a union bound that $\cT'$ satisfies the second requirement as well.
\pfend

\begin{clm}
\label{clm:substitution}
Assume that $\vertex$ is a vertex of $\cT_{t-1}$ at depth $t-1$.
If $|F_\vertex| \ge \sizelimit $, then $\pr\skA[A\cap S_\vertex = \emptyset \mid \vertex~\text{used on}~A] \le 1/k^2$.
\end{clm}

\pfstart
Let $S_i$ and $b_i$ be defined as above the statement of \rf(lem:TreeSubstitution).
Let $P = \{i\mid b_i=1\}$ be the locations where the path expects an answer~$1$ from the oracle.
For each $A$ on which $\vertex$ is used and each $i\in P$, define $h_i(A) \in [n]$ as the minimal element of $A\cap S_i$.  

Assume that $(c_i)_{i\in P}$ is some possible list of the values of $h_i(A)$ (attained for at least one $A$).
Notice that if $A$ is chosen from $\cX$ and we condition on $\vertex$ being used, and on 
$h_i(A) = c_i$ for all $i \in P$, then the distribution of $A$ is uniform over all 
sets of the form $C \cup B$ where
\[
C = \{ c_i \}_{i \in P} \, ,
\]
and $B$ is a subset of $[n] \setminus D$ of cardinality $k-|C|$ for 
\[
D = \sfig{j \in [n] \midA \exists i: (b_i=0 \wedge j\in S_i)\vee(b_i=1 \wedge j\in S_i\wedge j\le c_i) }\,.
\]
Therefore, under the same conditional distribution, we have
\begin{align*}
\pr\skA[A\cap S_\vertex = \emptyset] &\le \pr\skA[A\cap F_\vertex = \emptyset]
= \bigfrac{\binom{n-|D|-|F_\vertex|}{k-|C|}}/{\binom{n-|D|}{k-|C|}} \\
&\le \s[1-\frac {|F_\vertex|}n]^{k-|C|}
\le \s[1-\frac {4\ln k}{k}]^{k/2} \le \ee^{-2\ln k} = 1/k^2 \,
\end{align*}
where we used $|C| \le |P| \le k/2$. 
The case $A\in\cY$ is similar, and we get the same upper bound.  
As the conditional probability is at most $1/k^2$ in all possible cases, 
the probability in the statement of the lemma is also at most~$1/k^2$.
\pfend

Therefore, in order to complete the proof of the theorem, it suffices to show 
that the total variation distance between the output distributions of $\cT'$ 
on $\cX$ and on $\cY$ is $1-\Omega(1)$. 
Notice that since the depth of $\cT'$ is at most $m/(\sizelimit)$ and since each
query is to a subset of cardinality at most $\sizelimit$, the total number of 
elements of the universe $[n]$ ``addressed'' by $\cT'$ is at most $m$. 
It therefore suffices to prove the following stronger statement. 
For any deterministic decision tree $\cB$ of depth at most $m$ that makes queries of the form
$j \in_? A$, the total variation distance between the output distributions of $\cB$ 
on $\cX$ and on $\cY$ is at most $1-\Omega(1)$.
%\bnote{Do you mean that you split a query $S_\cV$ into $|F_\cV|$ individual queries to the fresh elements?  This is not apparent at first.}
%\onote{Any deterministic procedure that queries at most $m$ elements can clearly be represented by a decision tree of depth at most $m$; it indeed boils down to what you're saying (splitting each vertex) but I don't think we need to say it explicitly. It's really just follows from 
%basic properties of decision trees.}

To show this, we first claim that due to the strong symmetry properties of $\cX$ and $\cY$,
we can assume without loss of generality that for all $j$, all vertices of $\cB$ at depth $j-1$ 
perform a query on element $j$. Indeed, we can obviously assume that the queries along each
path from the root to a leaf are distinct. Therefore, by applying a permutation on the 
universe $[n]$, we can ensure that the query at the root is $1$ and all other vertices
perform queries in $\{2,\ldots,n\}$. Crucially, this permutation leaves both $\cX$ and $\cY$ invariant. Similarly, we can now recur on each subtree of the root, 
and perform a permutation so as to make both vertices at depth $1$ query the element $2$.
Since $\cX$ is invariant under permutations of $\{2,\ldots,n\}$,
both when conditioned on containing $1$ and on not containing $1$, and similarly for $\cY$,
this modification again does not affect the output distribution. Continuing in the same
fashion we obtain the claim. 

As a result, we obtain that $\cB$ depends only on $A \cap \{1,\ldots,m\}$. 
Since the distribution of the latter when $A$ is chosen from $\cX$ 
and from $\cY$ is symmetric, it is enough to show that the total
variation distance between the distribution of $|A \cap \{1,\ldots,m\}|$ 
when $A$ is chosen from $\cX$ and that distribution when $A$ is chosen from $\cY$, 
is $1-\Omega(1)$. 
This is shown in the following technical claim, for which
we observe that the former distribution is precisely 
the hypergeometric distribution $\cH_n(k,m)$ and the latter 
distribution is $\cH_n(k+d,m)$.
\pfend

\begin{clm}
\label{clm:subsetSampling}
For all integers $k,d>1$, there exists $N$ such that, for all $n>N$, the total variation distance between $\cH_n(k,m)$ and $\cH_n(k+d,m)$ is $1-\Omega(1)$, where $m = \min\sfig{\floor[n/4], \floor[n(k+d)/d^2]}$.
\end{clm}

\pfstart
\mycommand{distA}{\cB(k,p)}
\mycommand{distB}{\cB(k+d,p)}
Let $p=\min\sfig{1/4, (k+d)/d^2}$.
As is well known, $\cH_n(k, \floor[pn])$ converges to the binomial distribution $\distA$ as $n\to\infty$.
Thus, it suffices to estimate the total variation distance between $\distA$ and $\distB$.

If $d>k$, then 
\(
\Pr\skA[\cB(d,p)=0] \ge \s[1-2/d]^d \ge \ee^{-2}/2
\)
for large enough $d$.  Conditioned on the event that the last $d$ out of $k+d$ coin flips in $\distB$ are~0 the two distributions $\distA$ and $\distB$ become the same, hence the total variation distance between the (unconditioned) distributions $\distA$ and $\distB$ is at most $1-\Omega(1)$.
Thus we may further assume that $d\le k$.

Next,
\[
\frac{\Pr\skA[\distA=\ell]}{\Pr\skA[\distB=\ell]} = (1-p)^{-d} \frac{k(k-1)\cdots(k-\ell+1)}{(k+d)(k+d-1)\cdots(k+d-\ell+1)}
\]
is a non-increasing function in $\ell$, hence, the total variation distance between $\distA$ and $\distB$ is equal to the Kolmogorov distance between the two:
\begin{equation}
\label{eqn:kolmogorov}
\max_\ell \absB| \Pr\skA[\distB\ge\ell] - \Pr\skA[\distA\ge\ell] |\;.
\end{equation}

Let $\ell$ be the value for which the maximum in~\rf(eqn:kolmogorov) in attained. There are two cases.  \bnote{05.07: simplified} First, if $\ell>\ceil[(k+d)p]$, then 
\[
\Pr\skA[\distA\ge\ell]\le\Pr\skA[\distB\ge\ell]\le \frac12
\]
using the expression for the median of the binomial distribution~\cite{kaas:MedianBinomial}, hence~\rf(eqn:kolmogorov) is at most~$1/2$.

The second case is where $\ell\le \ceil[(k+d)p]$.
Defining $\ell' = \min \sfigA{k, kp + 3\sqrt{kp(1-p)} }$
and noting that $\ell \le \ell'$, we have that 
\bnote{03.07:  I guess that $\Omega(1)$ is $\int_3^{+\infty} \cN(p)\dd p$, but I am not sure.  I have left Slud's paper on my other computer.  Should we add this expression explicitly?}
\[
\Pr\skA[\distB\ge\ell] \ge
\Pr\skA[\distA\ge\ell] \ge
\Pr\skA[\distA\ge\ell'] = 
%\Pr\skB[\distA \ge \min \sfigA{k, kp + 2\sqrt{kp(1-p)} } ] = 
\Omega(1) \; ,
\]
where the last equality follows from Slud's inequality~\cite{slud:inequality}
(which applies here due to our choice of $\ell'$ and since $p \le 1/4$).
Again, it follows that~\rf(eqn:kolmogorov) is at most $1-\Omega(1)$.
\pfend

%%%%%%%%%%%%%%%%%%%%%%%%%%%%%%%%%%%%%%%%%%%%%%%%%%%%%%%%%%%%%%%%%%%%%%%%%%%%
\subsection{Quantum complexity}
\label{sec:ggtquery}
%%%%%%%%%%%%%%%%%%%%%%%%%%%%%%%%%%%%%%%%%%%%%%%%%%%%%%%%%%%%%%%%%%%%%%%%%%%%

The aim of this section is to show that the quantum query complexity of the \eggt \bnote{03.07: was \ggt}problem is $\Theta\sA[1 + \sqrt{k/d}]$.
Thus, when $\sqrt{k}\le d\le k$, the quantum algorithm provides a quartic improvement over the randomized one. We start with a lower bound for \eggt, which implies the same lower bound for \ggt.

\begin{prp}
The quantum query complexity of the \eggt problem with parameters $k$ and $d$ is $\Omega(1+\sqrt{k/d})$.
\end{prp}

\pfstart
Take $n=k+d$.  In this case $\cY$ contains only one element, namely $[n]$, and the corresponding function takes value~1 on every $S$ except $S=\emptyset$. Intuitively, one detects that $A\in\cX$ by finding an $i\not\in A$, so the problem becomes the unstructured search problem of size~$n$, where the $d$ elements of $[n]\setminus A$ are marked.  Unstructured search requires $\Omega\sA[\sqrt{n/d}]$ queries~\cite{boyer:groverTight}.

This intuitive argument can be made rigorous via the unweighted adversary lower bound (\rf(thm:unweightedadv)). We put the one input in $\cY$ in relation with all $\binom{n}{k}$ inputs in $\cX$, so $m=1$ and $m'=\binom{n}{k}$. For each fixed nonempty $S\subseteq[n]$, there are $\binom{n-|S|}{k}$ different $A\in \binom{[n]}{k}$ such that $A\cap S=\emptyset$ (these are the $A\in\cX$ where a query to $S$ returns value~0, showing that $A\neq[n]$). This number is maximized for $|S|=1$, so $\ell'\leq \binom{n-1}{k}$ (and $\ell=1$). The lower bound from \rf(thm:unweightedadv) is 
\[
\Omega\left(\sqrt{\frac{mm'}{\ell\ell'}}\right)=\Omega\left(\sqrt{\bigfrac\binom{n}{k}/\binom{n-1}{k}}\right)=\Omega\sA[\sqrt{n/d}]=\Omega(1+\sqrt{k/d}).\qedhere
\]
\pfend

The remaining part of this section is devoted to showing a matching upper bound on the quantum query complexity. In Section~\ref{sec:qggtalgorithm} we will show how to implement this algorithm time-efficiently.
\begin{thm}
\label{thm:combGroupTestbasic}
There exists a quantum algorithm that solves the \eggt problem with parameters $k$ and $d$ 
using $O \sA[\sqrt{1 + k/d}]$ queries.
\end{thm}

\bnote{03.07: Removed, we don't: We will call this algorithm the Quantum Gap Group Testing (\qggt) algorithm.}

\pfstart
We prove the theorem by constructing a feasible solution to the semidefinite program~\rf(eqn:advOrig)
and then applying \rf(thm:adv).
For the \eggt problem, the adversary bound reads as follows:
\begin{subequations}
\label{eqn:advNew}
\begin{alignat}{3}
&\mbox{\rm minimize} &\quad& \max_{A\in \cX\cup\cY}\sum\nolimits_{S\subseteq[n]} X_S\elem[A,A]  \label{eqn:advNewObjective} \\
& \mbox{\rm subject to}&& \sum_{S\colon A\cap S=\emptyset \;\mathrm{xor}\; B\cap S=\emptyset} X_S\elem[A, B] = 1 &\quad& \text{\rm for all $A\in\cX$ and $B\in\cY$;} \label{eqn:advNewCondition} \\
&&& X_S\succeq 0 && \mbox{\rm for all $S\subseteq[n]$.} \label{eqn:advNewSemidefinite}
\end{alignat}
\end{subequations}

Our feasible solution to~\rf(eqn:advNew) is an adaptation of the solution given in~\cite{belovs:learningSymmetricJuntas} for the task of \emph{finding} the subset~$A$.
It is possible to give a solution to~\rf(eqn:advNew) in the style of~\cite{belovs:learningSymmetricJuntas}. 
%\onote{12/21: for my benefit, what is the ``style of [7]''? also, are the matrices there not of rank 1?}
However, below we give a more direct construction resulting in matrices $X_S$ of rank~1.
%\bnote{We had it already, I can send you the old version containing this solution.  The rank was not estimated there, but I think it was 2.}

%We will assume that $n$ is large enough compared to $k$ and $d$.
Clearly, we may assume that $n\ge k+d$.
Let $S\subseteq[n]$, and $s=|S|$.  If $s=0$ or $s>n-k-d+1$, we define $X_S=0$.
If $1\le s\le n-k-d+1$, we define $X_S = \psi\psi^*$, where $\psi$ is a vector indexed by sets $A\in\cX\cup\cY$ with entries
\begin{equation}
\label{eqn:psiA}
\psi\elem[A] = 
\begin{cases}
\alpha_s,& \text{if $A\in\cX$ and $A\cap S = \emptyset$;}\\
\beta_s,& \text{if $A\in\cY$ and $|A\cap S| = 1$;}\\
0,&\text{otherwise.}
\end{cases}
\end{equation}
Here $\alpha_s$ and $\beta_s$ are some positive real numbers satisfying
\begin{equation}
\label{eqn:XSElem}
 \alpha_s \beta_s =  \s[(n-k){\binom{n-k-1}{s-1}}]^{-1}.
\end{equation}
The values of $\alpha_s$ and $\beta_s$ depend on $d$ and will be chosen later in order to minimize the value of the objective function~\rf(eqn:advNewObjective).  Thus, ignoring repeated and zero entries, $X_S$ is essentially a $2\times2$ block matrix of the form
\[
X_S =
\begin{pmatrix}
\alpha_s^2 & \alpha_s \beta_s \\
\alpha_s\beta_s & \beta_s^2
\end{pmatrix}.
\]
The proof now follows from the two claims below.
\pfend

\begin{clm}
The matrices $X_S$ form a feasible solution to~\rf(eqn:advNewCondition) for the \eggt problem and any value of $d$.
\end{clm}

\pfstart
Fix $A\in\cX$, $B\in\cY$, and let $\ell=|B\setminus A|\ge d$.  Note that $X_S\elem[A,B] = 0$ if the condition on $S$ in the sum in~\rf(eqn:advNewCondition) is not satisfied.
Next,
\[
\sum_{S\subseteq[n]} X_S\elem[A,B] =
\frac{1}{n-k} \sum_{s=1}^{n-k-\ell+1} \frac{\ell\binom{n-k-\ell}{s-1}}{\binom{n-k-1}{s-1}} = \frac{\ell}{n-k}\; T(n-k-\ell, n-k-1)\;,
\]
where, for non-negative integers $a\le b$, we define
\[
T(a,b) = 1 + \frac ab + \frac{a(a-1)}{b(b-1)} +\cdots + \frac{a(a-1)(a-2)\cdots 1}{b(b-1)(b-2)\cdots (b-a+1)}\;.
\]
Thus, to show that $\sum_{S\subseteq[n]} X_S\elem[A,B]=1$ it remains to show that $T(a,b) = (b+1)/(b-a+1)$.  This is easy to check by induction on $a$: the case $a=0$ is trivial, and for the inductive step we have
\[
T(a,b) = 1 + \frac ab T(a-1,b-1) = 1+\frac ab\s[\frac{b}{b-a+1}] = \frac{b+1}{b-a+1}\;.\qedhere
\]
\pfend

%Thus, the off-diagonal elements $X_S\elem[A,B]$ do not depend on $d$.
%What depends on $d$ are the diagonal terms $X_S\elem[A,A]$ and $X_S\elem[B,B]$ that influence the objective value~\rf(eqn:advDualObjective), and we now show how to chose them.

\begin{clm}
\label{clm:complexity}
For each $d$, there exists a choice of $\alpha_s$ and $\beta_s$ satisfying~\rf(eqn:XSElem) such that the objective value~\rf(eqn:advNewObjective) is $O(\sqrt{1+k/d})$.
\end{clm}

\pfstart
Fix a positive integer $s\le n-k-d+1$.  For all $A\in\cX$ and $B\in\cY$, we have
\begin{equation}
\label{eqn:XSAA}
\sum_{S\subseteq [n]\colon |S|=s} X_S\elem[A,A] = \binom{n-k}{s} \alpha_s^2
\end{equation}
and
\begin{equation}
\label{eqn:XSBB}
\sum_{S\subseteq [n]\colon |S|=s} X_S\elem[B,B] = (k+d)\binom{n-k-d}{s-1} \beta_s^2.
\end{equation}

We take $\alpha_s$ and $\beta_s$ so that the values of~\rf(eqn:XSAA) and~\rf(eqn:XSBB) are equal.  In particular, they are equal to their geometric mean, which, by~\rf(eqn:XSElem), is
\begin{equation}
\label{eqn:XSMean}
\frac{\sqrt{\binom{n-k}{s}(k+d)\binom{n-k-d}{s-1}}}{(n-k)\binom{n-k-1}{s-1}}
\le \sqrt{\frac{k+d}{s(n-k)} {\s[1-\frac {s-1}{n-k-1}]}^{d-1}} ,
\end{equation}
where we used that $\binom{n-k}{s}=\frac{n-k}{s}\binom{n-k-1}{s-1}$ and 
\[
\frac{(k+d)\binom{n-k-d}{s-1}}{(n-k)\binom{n-k-1}{s-1}}
\le \frac{k+d}{n-k} \s[\frac{n-k-s}{n-k-1}]^{d-1}
= \frac{k+d}{n-k} {\s[1-\frac {s-1}{n-k-1}]}^{d-1} .
\]
Let us denote $m=n-k-1$.  Using~\rf(eqn:XSMean), we get that for all $A\in\cX\cup\cY$:
\begin{align*}
\sum_{S\subseteq[n]} X_S\elem[A,A] & = 
\sum_{S\subseteq [n]\colon |S|=1}X_S\elem[A,A] + \sum_{s=2}^{n-k-d+1}\sum_{S\subseteq [n]\colon |S|=s}X_S\elem[A,A] \\
&\le \sqrt{\frac{k+d}d} + \frac1m\sum_{s=2}^{m+1} \sqrt{\frac{k+d}{(s-1)/m} {\s[1-\frac {s-1}m]}^{d-1}} \\
&\le \sqrt{\frac{k+d}d} + \int_0^1 \sqrt{\frac{k+d}{p} (1-p)^{d-1}}\;\dd p \\
&= \sqrt{\frac{k+d}d} + \sqrt{k+d}\; \mathrm B\sA[1/2, (d+1)/2] = O\s[\sqrt{\frac{k+d}d}],
\end{align*}
where B stands for the beta function.  Here, we substituted $p=(s-1)/m$, used monotonicity of the function $p^{-1}(1-p)^{d-1}$, and applied a well-known asymptotic for the beta function.
\pfend

%\subsection{Extending the algorithm to \ggt: The role of irrelevant variables}\rnote{15/5: added this subsection heading in order to separate these technical observations that are only needed later from the main proof of the section}
%
%We will now elaborate a bit more on the role of irrelevant variables in the above algorithm. 

From~\rf(eqn:psiA) we have the following important observation.

\begin{obs}
\label{obs:qggtIrrelevant}
The feasible solution to the adversary bound~\rf(eqn:advNew) constructed in the proof of \rf(thm:combGroupTestbasic) has the following irrelevant variables in the sense of \rf(prp:irrelevant):
\itemstart
\item If the input $A$ is in $\cX$ (i.e., $|A|=k$), an input variable $S\subseteq[n]$ is irrelevant if $S\cap A\ne\emptyset$.
\item If the input $A$ is in $\cY$ (i.e., $|A|=k+d$), an input variable $S\subseteq[n]$ is irrelevant if $|S\cap A|\ne1$.
\itemend
\end{obs}

This observation will be used in \rf(sec:qalgs), when we compose the algorithm of this section with an influence tester to get a query-efficient junta tester.
%\rnote{30/6: say something about GGT here? and about what happens if $|A|<k$ or $>k+d$?}
\bnote{05.07: added:}
In \rf(sec:qggtalgorithm), we will show that the same algorithm also solves the \ggt problem.

\section{Quantum algorithm for junta testing}
\label{sec:qalgs}
\mycommand{klv}{\floor[\log(200k)]}
\mycommand{SubInf}{\mathrm{\underline{Inf}}}

The aim of this section is to prove the following theorem:
\begin{thm}
\label{thm:juntaMain}
There exists a bounded-error quantum tester that distinguishes $k$-juntas from functions that are $\eps$-far from any $k$-junta, with query complexity 
\[
O\left(\sqrt{k/\eps}\log k\right).
\] 
%Also, there exists a quantum tester that performs the above task in $\tO\sA[n\sqrt{k/\eps}]$ time and $\tO\sA[\sqrt{k/\eps}]$ queries, where $n$ is the number of variables in the tested function.
\end{thm}

%The query-efficient and the time-efficient algorithms in \rf(thm:juntaMain) are essentially the same, but in the query-efficient algorithm we apply \rf(thm:composition) in order to save some logarithmic factors.

Suppose $f\colon \01^n\rightarrow\{\pm 1\}$ depends on a set $J\subseteq[n]$ of $K$ variables. Thus, the promise is that either $f$ is a $k$-junta ($K\leq k$), or $f$ is $\eps$-far from any $k$-junta (we will call such $f$ a ``non-junta'' for simplicity). The goal of the tester is to distinguish these two cases.  
%This is achieved by \rf(alg:juntaTester).  In the remaining part of the section, we describe the algorithm in more detail, and prove its correctness and estimate its query complexity.

At the lowest level of our algorithm, there is the following subroutine.
\begin{lem}[Influence Tester]
\label{lem:infTester}
There exists an algorithm that, 
given a subset $V\subseteq[n]$, accepts with probability at least 0.9 if $\Inf_V(f)\ge\delta$ and rejects 
with certainty if $\Inf_V(f)=0$.  The algorithm uses $O(\sqrt{1/\delta})$ queries and $O(n/\sqrt{\delta})$ other elementary operations.
\end{lem}

\pfstart
We pick $x,y$ randomly as described in \rf(sec:prelim) below~\eqref{eqn:InfSdef}, and check if $f(x) \neq f(y)$. By applying amplitude amplification 
for $O(1/\sqrt{\delta})$ rounds to amplify the basis states where $f(x)\neq f(y)$, we obtain the lemma. 
\pfend

The idea is to run the \eggt algorithm of \rf(thm:combGroupTestbasic) with the Influence Tester of \rf(lem:infTester) as the input oracle.  The complication is that we do not know what value of $\delta$ we should specify: the Fourier weight of a non-junta can be either concentrated on few (though more than $k$) variables with large influence, or scattered over many variables with tiny influence, and these cases call for different values of $\delta$.
We identify roughly $\log k$ different types of non-juntas, and design a separate tester for each of them.  A junta will be accepted (with high probability) by all of these testers, whereas a non-junta will be rejected by at least one of them.  The description is given in \rf(alg:juntaTester), followed by the definition of the different types of non-juntas.

\begin{algorithm}[htb] 
\algcaption{Quantum Junta Tester\\ 
Input:& {Integer $k>0$, real $0<\eps<1$, and oracle access to $f\colon \cube\to\sbool$}\\
Output:& Accepts if $f$ is a $k$-junta; rejects if $f$ is $\eps$-far from any $k$-junta.
} \label{alg:juntaTester}
%\onote{12/21:is it really Grover or some kind of amplitude amplification?}\rnote{it's certainly not plain Grover; a modification of HMW can do this -- this should probably be made more precise} \bnote{Let us keep things simple.}
\enumstart
\item Accept if all of the following $\klv+2$ testers accept, reject if at least one of them rejects:
\negmedskip
\itemstart
\item Tester of the first kind with $\ell \in \{0,\dots,\klv\}$;
\item Tester of the second kind.
\itemend
\enumend

\sbrtncaption{Tester of the first kind\\
Input:& Integer $\ell\ge 0$.\\
Output:& Accepts if $f$ is a $k$-junta, rejects if $f$ is a non-junta of the first kind with this value of $\ell$.}
\label{sbrtn:firstkind}
\enumstart
\item Run the \eggt algorithm of \rf(thm:combGroupTestbasic) \bnote{03.07: changed} with parameters $k$ and $d=2^\ell$ and the following oracle:
\negmedskip
\itemstart
\item On input $S\subseteq [n]$, run Influence Tester on $V=S$ and $\delta = \eps/(2^{\ell+3}\log(400k))$.
\itemend
\negmedskip
\item Accept if the \eggt algorithm accepts, otherwise reject.
\enumend

\sbrtncaption{Tester of the second kind\\
Output:& Accepts if $f$ is $k$-junta, rejects if $f$ is a non-junta of the second kind.}
\label{sbrtn:secondkind}
\enumstart
\item Estimate the acceptance probability of the following subroutine with additive error\rnote{30/6: was: `precision'} $.05$:
\negmedskip
\itemstart
\item Generate $V\subseteq[n]$ by adding each $i$ to $V$ with probability $1/k$ independently at random.
\item Run Influence Tester with this choice of $V$ and $\delta = \eps/(4k)$.
\itemend\negmedskip
\item Accept if the estimated acceptance probability is $\le0.8$, otherwise reject.
\enumend

\end{algorithm}

Let us describe these types of non-juntas.  For notational convenience, assume the first~$K$ variables are the influential ones, ordered by influence (of course, the tester does not know this order): 
\begin{equation}
\label{eqn:order}
\Inf_1(f)\geq\Inf_2(f)\geq\cdots\geq\Inf_K(f)>0=\Inf_{K+1}(f)=\cdots=\Inf_n(f).
\end{equation}
\mycommand{InfL}{\Inf_{\{k+1,\dots,K\}}(f)}
\mycommand{InfA}{\Inf_{\{k+1,\dots,200k\}}(f)}
\mycommand{InfB}{\Inf_{\{200k+1,\dots,K\}}(f)}
Our tester does not know the number~$K-k$ of ``extra'' variables if $f$ is far from any $k$-junta. However, \rf(lem:weightonextravars) implies that 
\begin{equation}
\label{eqn:extraWeight}
\InfL \ge \eps.
\end{equation}
Our tests are tailored to the following cases:

\enumstart
\item $\sum_{j=k+1}^{200k} \Inf_j(f) \ge \eps/2$.
This case is additionally split into $\klv+1$ subcases:
\[
\abs|\sfig{j\in[n] \midB \Inf_j(f)\ge \frac{\eps}{2^{\ell+3}\log (400k)}}| \ge k+2^{\ell},\qquad\text{where $\ell\in\{0,\dots,\klv\}$;}
\]
Such an $f$ is \emph{of the first kind}, for this value of~$\ell$.
\item $\sum_{j=k+1}^{200k} \Inf_j(f) \le \eps/2$. Such an $f$ is \emph{of the second kind}.
\enumend

\rnote{30/6: added}
Note that $f$ may be a non-junta of the first kind for many different values of~$\ell$ simultaneously; an extreme example is if $f$ is the $n$-bit parity function.

\begin{lem}
\label{lem:cases}
Every non-junta $f$ satisfies at least one of the cases above.
\end{lem}

\pfstart
It is clear that any $f$ satisfies the first or the second case above, so the only thing we need to show is that the first case is fully covered by its $\klv+1$ subcases. Assume $f$ satisfies the first case.
Denote $\eps' = \eps/(8\log(400k))$ and consider the following intervals, which together partition the interval~$[0,1]$:
\[
A_\infty = \left[0,\; \frac{\eps'}{2^{\klv}}\right),\qquad
A_{\ell} = \left[\frac{\eps'}{2^{\ell}},\; \frac{\eps'}{2^{\ell-1}}\right),\qquad
A_0 = \left[\eps',\; 1\right],
\]
where $\ell$ runs from $\klv$ to 1.
Let 
\[
B_\ell = \sfigA{j\in \{k+1,\dots,200k\} \mid \Inf_j(f)\in A_\ell}.
\]
Each $j$ is included in exactly one of the $B_\ell$.  Let also $W_\ell = \sum_{j\in B_\ell} \Inf_j(f)$.  Thus, 
\(
\sum_\ell W_\ell \ge \eps/2,
\)
\rnote{30/6: added}
because we are in the first case.
Next,
\(
W_\infty < 200k\cdot \eps/(8\cdot 2^{\klv}) \le \eps/4
\).
Thus, there exists $\ell\in \{0,\dots,\klv\}$ such that $W_\ell\ge \eps/(4\log(400k))$.  
Then, either $\ell=0$ and $|B_\ell|\ge 1$, or
\[
|B_\ell| \ge W_\ell\bigfracR/\s[\frac{\eps}{2^{\ell+2}\log(400k)}]. \ge 2^\ell\;.
\]
Also, all $j\in B_\ell$ satisfy $\Inf_j(f)\ge \eps/(2^{\ell+3}\log(400k))$.  By~\rf(eqn:order), all $j\in[k]$ also satisfy this inequality.  This means that $f$ satisfies the first case with this value of $\ell$.
\pfend

\begin{lem}
\label{lem:Bjcase}
For each $\ell \in\{0,\dots,\klv\}$, \rf(sbrtn:firstkind) accepts if $f$ is a junta, and rejects if $f$ is a non-junta of the first kind for this value of $\ell$.
Its query complexity can be made $O(\sqrt{(k/\eps)\log k})$.
\end{lem}

\pfstart
The composition in \rf(sbrtn:firstkind) is understood here in the sense of \rf(defn:irrCompose) with the functions $F$ and $G$ defined as follows.

The partial function $F$ is the \eggt function
\rnote{30/6: should this be GGT instead of EGGT? what if $f$ depends on fewer than $k$ coordinates?}
\bnote{03.07:  I don't see no problem here.  The last two paragraphs of the proof use EGGT.  It is even cleaner than GGT.}
 from \rf(defn:exactGroupTesting).
Given a function $h\colon 2^{[n]}\to\bool$, $F(h)=0$ if $h=\Inter_A$ with $|A|=k$, and $F(h)=1$ if $h=\Inter_A$ with $|A|=k+d$, where $\Inter_A$ is defined in~\rf(eqn:fADefn).  In all other cases, the value $F(h)$ is not defined.

For each $S\subseteq[n]$, the partial function $G_S$ is as defined in \rf(lem:infTester).  Given a total function $f\colon \cube\to\bool$, define $G_S(f)=0$ if $\Inf_S(f)=0$ and $G_S(f)=1$ if $\Inf_S(f)\ge \delta$.  If $0<\Inf_S(f)<\delta$, the value $G_S(f)$ is not defined.

The function evaluated in \rf(sbrtn:firstkind) is the following restriction of the composed function $F\circ(G_\emptyset, G_{\{1\}}, G_{\{2\}},\dots, G_{[n]})$:
\[
f \mapsto F\sA[ G_\emptyset(f), G_{\{1\}}(f), G_{\{2\}}(f),\dots, G_{[n]}(f) ],
\]
where, as the arguments of~$F$, we have all possible $2^n$ functions~$G_S$.
The composition here is understood as in \rf(defn:irrCompose) with the irrelevant variables of $F$ given by \rf(obs:qggtIrrelevant).

The query complexity of the subroutine can be computed using \rf(cor:irrCompose).
The complexity of the algorithm for $F$ in \rf(thm:combGroupTestbasic) is $O(\sqrt{k/2^\ell})$, as $2^\ell = O(k)$.  The quantum query complexity of each $G_V$ is $O(\sqrt{(2^\ell/\eps)\log k})$ by \rf(lem:infTester).  Thus, the total query complexity of the subroutine is $O(\sqrt{(k/\eps)\log k})$.

Let us prove the correctness of the subroutine.
Assume $f$ is a non-junta of the first kind with this value of $\ell$.
By definition, there exists $A\subseteq [n]$ of size $k+d$ such that for all $j\in A$, $\Inf_j(f)\ge\delta$.  As the influence is monotone in $S$, $\Inf_S(f) \ge\delta$ for all $S$ that intersect $A$.
By \rf(obs:qggtIrrelevant), all input variables $S$ satisfying $S\cap A=\emptyset$ are irrelevant, hence the value of the composed function is~1 in this case.

On the other hand, if $f$ is a junta, there exists $A\subseteq [n]$ of size $k$ such that for all $S\subseteq[n]$ satisfying $S\cap A=\emptyset$, we have $\Inf_S(f) =0$. 
By \rf(obs:qggtIrrelevant), all input variables $S$ satisfying $S\cap A\ne \emptyset$ are irrelevant, hence the value of the composed function is~0 in this case.
\pfend

From the proof of \rf(lem:Bjcase), it is clear why we need a separate tester for the second case.  If $2^\ell$ becomes $\omega(k)$, the complexity of Influence Tester still grows as $\tO(2^\ell/\eps)$, whereas the \eggt algorithm cannot use fewer than $O(1)$ queries.
Our second tester (\rf(sbrtn:secondkind)) does not use the \eggt algorithm, and relies on more traditional means.

\begin{lem}
\label{lem:B0case}
\rf(sbrtn:secondkind) accepts if $f$ is a junta, and rejects if $f$ is a non-junta of the second kind.  Its query complexity is $O(\sqrt{k/\eps})$.
\end{lem}

\pfstart

The estimate of the query complexity of \rf(sbrtn:secondkind) is straightforward.  Let us prove its correctness. We will show that the inner procedure has acceptance probability $\leq 0.75$ if $f$ is a $k$-junta, and acceptance probability $\geq 0.85$ if $f$ satisfies the second case. 

If $f$ is a $k$-junta then the probability that the set $V$ does not intersect with the set $J$ of (at most~$k$) relevant variables is: 
$$
(1-1/k)^{|J|}\geq (1-1/k)^k\geq 1/4,
$$
assuming $k\geq 2$.
If $V$ and $J$ are disjoint, then the algorithm always rejects, hence, the acceptance probability is at most $0.75$.

Now suppose $f$ is a non-junta of the second kind.
For notational convenience, we still assume that the variables of $f$ are ordered by decreasing influence as in~\rf(eqn:order).
For $j\in[n]$, let us define
\[
\SubInf_j(f) = 
\begin{cases}
0,& \text{if }j\le 200k;\\
\sum_{S\colon S\cap\{200k+1,\dots,j\} = \{j\}} \hf(S)^2,& \text{otherwise.}
\end{cases}
\]
For $S\subseteq[n]$, define $\SubInf_S(f) = \sum_{j\in S} \SubInf_j(f)$.

This quantity satisfies two important properties.  First, $0\le \SubInf_S(f) \le \Inf_S(f)$ for all $S\subseteq[n]$.  And second, it is additive \bnote{was linear} is $S$, i.e.,
\(
\SubInf_{S\cup T}(f) = \SubInf_S(f) + \SubInf_T(f)
\)
for all disjoint $S$ and $T$.  Note that $\Inf_S(f)$ is only subadditive \bnote{was sublinear} in $S$.

Next, as $f$ satisfies the second case, $\Inf_j(f)\le \eps/(200k)$ for $j>200k$.  Hence, $\SubInf_j(f) \le \eps/(200k)$ for all $j\in[n]$.
Finally,
\[
\SubInf_{[n]}(f) = \InfB \ge \InfL - \sum_{j=k+1}^{200k} \Inf_j(f) \ge \frac\eps2\;.
\]
Consider the random variable $\SubInf_V(f)$ where $V$ is as in \rf(sbrtn:secondkind).
Its expectation is
\[
\mu=\Exp[\SubInf_V(f)]=\frac1k \SubInf_{[n]}(f)\ge \frac{\eps}{2k}\;,
\]
and its variance is
\[
\sigma^2=\Var[\SubInf_V(f)]\leq
\frac{1}{k}\sum_{j}\SubInf_j(f)^2 \le \frac1k \max_j \SubInf_j(f)\cdot \SubInf_{[n]}(f) \leq \frac{\eps}{200k}\mu\leq \frac{\mu^2}{100}\;.
\]
Then, Chebyshev's inequality implies 
\[
\Pr\skA[\SubInf_V(f)<\eps/4k]\leq 
\Pr\skA[|\SubInf_V(f)-\mu|\geq \mu/2] \leq
\Pr\skA[|\SubInf_V(f)-\mu|\geq 5\sigma]\leq \frac{1}{5^2}=0.04\;.
\]
Hence, with probability at least 0.96, we have $\Inf_V(f)\ge \SubInf_V(f) \ge \eps/4k$.
If this is indeed the case, the influence tester in Lemma~\ref{lem:infTester}
accepts with probability $\geq 0.9$.
Thus, the inner procedure accepts with probability at least $0.96\cdot0.9>0.85$ if $f$ satisfies the second case.
\pfend

From Lemmas~\ref{lem:Bjcase} and~\ref{lem:B0case}, it is easy to see that \rf(alg:juntaTester) is correct.
If $f$ is a junta, then all of the $O(\log k)$ subtesters accept.  If $f$ is a non-junta, then at least one of them rejects (and the output of the remaining ones is not defined).
Thus, our algorithm is of the \rnote{30/6: added}``robust conjunction'' from~\rf(eqn:irrelExample).  Hence, using \rf(exm:and) and \rf(cor:irrCompose), we get that the query complexity of \rf(alg:juntaTester) is
\[
O\sA[\sqrt{\log k}\cdot\sqrt{(k/\eps)\log k}] = O\sA[\sqrt{k/\eps}\log k].
\]
This concludes the proof of \rf(thm:juntaMain).
%The outer quantum search over the $O(\log k)$ subroutines invokes those subroutines $O(\sqrt{\log k})$ times in superposition.  Again, the output of some testers can be undefined, but due to \rf(rem:composition), we can apply \rf(thm:composition), and get the total query complexity of the algorithm to be $$.

\section{Efficient implementation}
\label{sec:qggtalgorithm}
The main aim of this section is to prove that the algorithm from \rf(thm:combGroupTestbasic) can be implemented time-efficiently. Here by ``time'' we mean the total number of gates the algorithm uses, both the query-gates and all elementary quantum gates (from some arbitrary fixed universal set of gates) used to implement the unitaries in between the queries.

Moreover, we will prove that our algorithm computes a function that has irrelevant variables as specified by \rf(obs:qggtIrrelevant).\rnote{30/6: what does it mean for an \emph{algorithm} to have irrelevant variables?}  
\bnote{03.07: I guess ``evaluation'' is synonymous with ``algorithm.''  So I don't know what you want to imply here.}\rnote{6/7: the point is that we never defined irrelevant variables for algorithms; it doesn't fit with Def 2.8.  We shouldn't be so casual about proper definitions, this is math not sociology!  I now added a few words}
For clarity, we will now explicitly describe the problem which arises from applying \rf(defn:irrEvaluate) to the \eggt problem of \rf(defn:exactGroupTesting).

\begin{defn}[\qggt]
\label{defn:quantumGroupTest}
In the \emph{quantum gap group testing (\qggt) problem} with parameters $k$ and $d$, one is given access to an oracle $O_f$ satisfying the following properties.
The oracle $O_f$ acts on two registers: the $n$-qubit input register $\reg I$, and an arbitrary internal working register $\reg W$.  The oracle is in the block-diagonal form $O_f = \bigoplus_{S\subseteq\cube} O_{f,S}$, where $O_{f,S}$ is a unitary operator on $\reg W$, that gets invoked in $O_f$ when the value of the register $\reg I$ is $S$.
We are promised that $O_f$ belongs to one of the following two families:
\begin{equation}
\label{eqn:cXQuantum}
\ctX = \sfig{O_f \midA
\exists A\in\cX\; \forall S\subseteq[n] : S\cap A = \emptyset \implies O_{f,S} \ket W|0> = \ket W|0>}
\end{equation}
and\footnote{One can also weaken the premise $S\cap B\ne\emptyset$ in~\rf(eqn:cYQuantum) to $|S\cap B| = 1$.  We chose this definition to make the \qggt problem more similar to the \ggt problem.}
\begin{equation}
\label{eqn:cYQuantum}
\ctY = \sfig{O_f \midA
\exists B\in\cY \; \forall S\subseteq[n] : S\cap B \ne \emptyset \implies O_{f,S} \ket W|0> = -\ket W|0>}.
\end{equation}
The task is to detect whether $O_f\in\ctX$ or $O_f\in\ctY$.
\end{defn}

\begin{thm}
\label{thm:combGroupTestfull}
There exists a quantum algorithm that solves the \qggt problem with parameters $k$ and $d$ in time $\tO \sA[n\sqrt{1 +k/d}]$ using $O \sA[\sqrt{1 + k/d}]$ queries.
\end{thm}

The time complexity of the algorithm is roughly $n$ times its query complexity; as mentioned in the introduction, this is probably the best one can hope for.

\bnote{03.07: Added the following corollary for clarity.}

Note that the \qggt problem incorporates the usual quantization of the \ggt problem from \rf(defn:GroupTesting), thus we have the following:
\begin{cor}
There exists a quantum algorithm that solves the \ggt problem with parameters $k$ and $d$ in time $\tO \sA[n\sqrt{1 +k/d}]$ using $O \sA[\sqrt{1 + k/d}]$ queries.
\end{cor}

However, the \qggt problem is more general than the \ggt problem, the difference being that $O_{f,S}$ may be an \emph{arbitrary} unitary in $\reg W$ when the premises in~\rf(eqn:cXQuantum) or~\rf(eqn:cYQuantum) do not hold.

With \rf(thm:combGroupTestfull) in hand, it is easy to show that \rf(alg:juntaTester) can be implemented time-efficiently as well, with a slight increase in the number of queries.

\begin{thm}
\label{thm:juntaTimeEfficient}
There exists a bounded-error quantum tester that distinguishes $k$-juntas from functions that are $\eps$-far from any $k$-junta in time $\tO\sA[n\sqrt{k/\eps}]$ using $\tO\sA[\sqrt{k/\eps}]$ queries.
\end{thm}

\pfstart
It is easy to see from the proof of \rf(lem:B0case) that the time complexity of \rf(sbrtn:secondkind) is $\tO\sA[n\sqrt{k/\eps}]$.
Unfortunately, it is hard to estimate the time complexity of the implementation of \rf(sbrtn:firstkind) in \rf(lem:Bjcase), 
because \rf(lem:Bjcase) invokes \rf(prp:irrelevant)(b) to analyze the composition of quantum algorithms (\rf(prp:irrelevant)(b) upper bounds the \emph{query} complexity of the composition but not its \emph{time} complexity). 
However, \rf(alg:juntaTester) can be implemented to have time complexity $\tO\sA[n\sqrt{k/\eps}]$ as follows. %We modify \rf(sbrtn:firstkind) as follows.  

We first reduce the error probability of each call to the Influence Tester of \rf(lem:infTester) to 
%\bnote{2/18: Used to be $1/k^2$.  Could not figure out where it had come from}\rnote{3/28: if the error prob per query is $\alpha$, the error of an algorithm involving $T$ queries is $O(T\sqrt{\alpha})$, so we want $\alpha\ll 1/T^2$. The current $\ll\eps/k$ is OK, keeping in mind that the `$\ll$' has to overcome a few log-factors.}
$\ll \eps/k$ by $O(\log \frac k\eps)$ repetitions, and run it backwards (after copying the answer) to set the workspace back to its initial state; then run the \qggt algorithm on this oracle as if it's errorless.  Standard techniques show that the resulting variant of \rf(sbrtn:firstkind) can be made to have error probability $\leq 1/3$, and we do not need to invoke \rf(prp:irrelevant)(b) anymore. 
The query complexity of \rf(sbrtn:firstkind) has now gone up by a factor $O(\log \frac k\eps)$, but its time complexity becomes $\tO\sA[n\sqrt{k/\eps}]$, because it is the time complexity of the \qggt algorithm, plus its oracle-query complexity multiplied by the time complexity of the amplified Influence Tester of~\rf(lem:infTester) that implements one oracle call.
%\footnote{Note that the time complexity of one call to Fourier sampling is $O(n)$, so the time complexity of the algorithm of \rf(rem:testInfluence) is $\tO\sA[n\sqrt{1/\delta}]$.}

The resulting variant of \rf(alg:juntaTester) can thus be implemented in time $\tO\sA[n\sqrt{k/\eps}]$.
\pfend

\subsection{Proof of \rf(thm:combGroupTestfull)}
\label{sec:proofOfQGGT}
It remains to prove \rf(thm:combGroupTestfull).
This is done by a (by now relatively standard) implementation of the dual adversary bound as in~\cite{reichardt:advTight}.
The analysis follows~\cite{lee:stateConversion}, with the simplification that we have Boolean input and output (see also~\cite[Section 3.4]{belovs:phd}).  
%Additionally, as promised at the end of \rf(sec:groupTesting), we show that this algorithm solves the more general \qggt problem.
Our main innovation here is an efficient implementation of a specific reflection in \rf(sec:lambda), which we do by means of a new and efficient quantum Fourier transform.

%\onote{was: Recall that the \qggt problem is defined in \rf(defn:quantumGroupTest).}
Recall the \qggt problem as defined in \rf(defn:quantumGroupTest).
Before we proceed with the algorithm, we have to make a minor modification to this definition.  
Due to technical reasons in \rf(clm:correct) below, we have to assume that $O_{f,S}$ not only satisfies~\rf(eqn:cXQuantum) or~\rf(eqn:cYQuantum), but is also a reflection.  This is without loss of generality.  
Indeed, assume $O_f$ satisfies the original constraints of \rf(defn:quantumGroupTest).  Add a new basis element $\ket |1>$ to $\reg W$, and assume that $O_{f,S}$ does not change it.  Let $V$ be a unitary on $\reg W$ that maps $\ket W|0>$ to $\tfrac 1{\sqrt{2}}\sA[\ket W|0> + \ket W|1>]$.  
Denote by $O'_f$ the following chain of operations: apply $O_fV$, reflect about $\tfrac 1{\sqrt{2}}\sA[\ket W|0> + \ket W|1>]$, and apply $V^{-1} O_f^{-1}$.  
It is clearly a reflection, and it has the same decomposition $O'_f = \bigoplus_{S\subseteq\cube} O'_{f,S}$ as $O_f$ does.  Moreover, it is straightforward to check that if $O_f$ is in $\ctX$ or $\ctY$, then $O_f'$ is also in $\ctX$ or $\ctY$, respectively.

Our algorithm only uses the input register $\reg I$, so we omit this subscript below.  The register $\reg W$ is not written, but assumed to be in the state $\ket W|0>$.
We also add a new basis state $\ket|0>$ to $\reg I$, and assume that $O_f\ket|0> = -\ket|0>$ for all $O_f$.

The query-efficient algorithm in \rf(thm:combGroupTestbasic) was obtained by constructing the matrices $X_S$ in~\rf(eqn:psiA) that depend on parameters $\alpha_s$ and $\beta_s$ satisfying~\rf(eqn:XSElem).
The objective value~\rf(eqn:advNewObjective) is $W=O(\sqrt{1+k/d})$ by \rf(clm:complexity). Also we define $\gamma = C_1\sqrt{W}$ for some constant $C_1$ to be determined later.

Let $\Lambda$ be the projector onto the span of the vectors
\begin{equation}
\label{eqn:psia}
\psi_A = \ket |0> + \gamma \sum_{s=1}^{n-k-d+1} \alpha_s
\sum_{S\subseteq[n] \colon |S|=s,\; S\cap A=\emptyset} \ket|S>
\end{equation}
over all $A\in\cX$, and $R_\Lambda = 2\Lambda - I$ be the corresponding reflection. (In \rf(sec:lambda) we show how to implement $R_\Lambda$ efficiently.) The \qggt problem is solved by \rf(alg:groupTesting), where $C$ is some constant to be defined later.

\begin{algorithm}[htb]
\algcaption{Quantum Algorithm for the \qggt problem\\
Input:& access to an oracle $O_f\in \ctX\cup \ctY$ as in~\rf(eqn:cXQuantum) or~\rf(eqn:cYQuantum).\\
Accepts:& iff $O_f\in\ctX$.}
\label{alg:groupTesting}
\enumstart
\item Prepare the state $\ket|0>$.
\item Perform phase estimation on the operator $U=O_fR_\Lambda$ with precision $\delta = 1/(CW)$.
\item Accept if and only if the phase-estimate is greater than $\delta$.
\enumend
\end{algorithm}

\begin{clm}
\label{clm:correct}
\rf(alg:groupTesting) is correct.
\end{clm}

\pfstart
Let us first assume that $O_f\in\ctY$.  Let $B\in\cY$ be a\rnote{30/6: was `the'} corresponding element from~\rf(eqn:cYQuantum), so $|B|=k+d$. Define the following vector
\begin{equation}\label{eqn:defu}
u = \gamma \ket|0> - \sum_{s=1}^{n-k-d+1} \beta_s
\sum_{S\subseteq[n] \colon |S|=s,\; |S\cap B|=1} \ket|S>.
\end{equation}
The squared norm of this vector is
\[
\|u\|^2 = \gamma^2 + \sum_{s=1}^{n-k-d+1}(k+d)\binom{n-k-d}{s-1} \beta_s^2=\gamma^2 + \sum_{S\subseteq[n]}X_S\elem[B,B]\leq C_1^2W + W,
\]
where the second equality uses~\rf(eqn:XSBB), and the last inequality uses that the objective value~\rf(eqn:advNewObjective) is~$W$.

We now show that $u$ is an eigenvector of $U=O_fR_\Lambda$ with eigenvalue~1 (so the eigenvalue's phase is~0). First, $O_f u = -u$, because $O_f\ket|S>=-\ket|S>$ for all $S$ occurring in \rf(eqn:defu) and we earlier already assumed $O_f\ket|0> = -\ket|0>$. Second, for all $A\in\cX$ we have
\[
\ip<\psi_A,u> = \gamma - \gamma \sum_{s=1}^{n-k-d+1}\sum_{S\subseteq[n]:|S|=s,S\cap A=\emptyset,|S\cap B|=1} \alpha_s\beta_s = \gamma - \gamma \sum_{S\colon A\cap S=\emptyset \;\mathrm{xor}\; B\cap S=\emptyset} X_S\elem[A,B] = 0,
\]
where we used~\rf(eqn:psiA) and~\rf(eqn:advNewCondition).
Hence, $\Lambda u=0$ and $R_\Lambda u = (2\Lambda-I)u = -u$. Therefore, $Uu=O_fR_\Lambda u=u$.  

Furthermore, the inner product of the normalized eigenvector $u/\|u\|$ and $\ket|0>$ is 
\[
\frac{\gamma}{\|u\|} \geq \frac{C_1\sqrt{W}}{\sqrt{C_1^2W + W}}=\frac{1}{\sqrt{1+1/C_1^2}},
\]
which can be made arbitrarily close to~1 by setting~$C_1$ to a sufficiently large constant. Since $\ket|0>$ is the starting state of \rf(alg:groupTesting), the algorithm will (with probability at least $2/3$ if we set $C_1$ appropriately) produce a phase estimate that is at most $\delta$, and correctly rejects $O_f\in\ctY$.

\medskip

Now assume $O_f\in\ctX$. Let $A\in\cX$ be the corresponding element from~\rf(eqn:cXQuantum), so $|A|=k$.  In this case, we will apply \rf(lem:effective) with $R_1 = -R_\Lambda=I-2\Lambda$, $R_2 = - O_f$ (hence $U=O_fR_\Lambda=R_2R_1$), $\Pi_1 = I - \Lambda$, $\Pi_2=(I-O_f)/2$, and $w = \psi_A$.  Indeed, since $\Lambda w=w$, $w$ lies in the kernel of $\Pi_1$, and we assume $O_f$ is a reflection, so the conditions of the lemma are satisfied.  We have $\Pi_2 w = \ket|0>$ because $O_f\ket|0> = -\ket|0>$, and $O_f\ket|S> = \ket|S>$ for all~$S$ in the support of $\psi_A$.  Also, 
\[
\|\psi_A \|^2 = 1+\gamma^2\sum_{s=1}^{n-k-d+1}\binom{n-k}{s} \alpha_s^2 = 1+\gamma^2W=1+C_1^2W^2.
\] 
Since also $W=\Omega(1)$, we have $\|w\| = O(W)$.
Therefore, using \rf(lem:effective), the algorithm's initial state $\ket|0>$ barely overlaps with eigenvectors of $U=R_2R_1$ whose phase is $(2\delta)$-close to~0:
\[
\|P_{2\delta}\ket|0>\|=\|P_{2\delta} \Pi_2 w \|\leq \delta\|w\|=O(1/C).
\]
Hence the probability that phase estimation erroneously yields an estimate that is $\delta$-close to~0 can be made less than~$1/3$ by choosing~$C$ a sufficiently large constant. Then \rf(alg:groupTesting) accepts all $O_f\in\ctX$ with probability at least~$2/3$. 
\pfend

As $R_\Lambda$ can be implemented without executing the input oracle, the query complexity of \rf(alg:groupTesting) is $O(W)=O(\sqrt{1+k/d})$ by \rf(thm:estimation).
To get the time complexity, the query complexity has to be multiplied by the cost of implementing $U=O_fR_\Lambda$.  
%By convention, the oracle query does not count towards the time complexity of the algorithm,  RdW I removed this, typically an oracle costs 1
In \rf(sec:lambda) we show that $R_\Lambda$ can be implemented in time $\tO(n)$. Thus, \rf(alg:groupTesting) can be implemented in time $\tO(n\sqrt{1+k/d})$.

\subsection{Efficient implementation of \texorpdfstring{$R_\Lambda$}{RLambda}}
\label{sec:lambda}

This section is devoted to the proof of the following lemma, which shows that the reflection $R_\Lambda = 2\Lambda - I$ can be implemented efficiently, in time  $\tO(n)$.
For simplicity we assume $n> 2k$.  This is without loss of generality, as we can extend the set $[n]$ with dummy elements.  Next, we identify $\ket|0>$ of Eq.~\rf(eqn:psia) with $\ket|\emptyset>$, and absorb $\gamma$ into $\alpha_s$.  To state the lemma, it is also more convenient to replace $A$ in Eq.~\rf(eqn:psia) by its complement, $T=[n]\setminus A$.

\begin{lem}
\label{lem:reflection}
Let $\alpha_0, \alpha_1, \ldots, \alpha_{n-k}$ be arbitrary complex numbers and let $\Lambda$ be the projector onto the span of the vectors
\[
\psi_T = 
\sum_{\ell=0}^{n-k} \alpha_\ell
\sum_{B\subseteq T \colon |B|=\ell} \ket|B>
\]
over all $T \subseteq [n]$ with $|T|=n-k$.
Then, the corresponding reflection $R_\Lambda = 2\Lambda - I$ can be implemented in time $\tO(n)$, up to an error in the operator norm that can be made smaller than any inverse polynomial in~$n$.
\end{lem}

%\onote{5/21: was ``For simplicity, we consider the case when all $\alpha_i$ are real, however, the same result holds when they are complex.''}
%It is easy to see that it suffices to prove the lemma for 
%real numbers $\alpha_0, \alpha_1, \ldots, \alpha_{n-k} \ge 0$, 
%and in the sequel we will assume that this is the case. 

\paragraph{Representation theory background.}
In order to prove \rf(lem:reflection), we will use the structural properties implied by the invariance of the vectors $\psi_T$ under permutations of~$[n]$. We need some basic results from the representation theory of the symmetric group. These results are only used in this section.
The reader may refer to a textbook on the topic such as~\cite{sagan:symmetricGroup}, or to the appendix of~\cite{belovs:learningSymmetricJuntas}, where we briefly formulate the required notions and results.

\mycommand{group}{\bS_n}
\mycommand{algebra}{\group}
Let $\group$ denote the symmetric group on $[n]$.  We consider (left) modules over the group algebra $\bC\group$.  We call them \emph{$\algebra$-modules}; they are also known as \emph{representations} of $\bS_n$.
There is a 1-1 correspondence between irreducible $\bS_n$-modules and partitions $(t_1, \ldots, t_k)$ of $n$ (where $t_1\geq t_2\geq \cdots\geq t_k$ and $t_1+t_2+\cdots+t_k=n$). 
Irreducible $\algebra$-modules are called \emph{Specht modules}.

A linear mapping $\theta\colon V\to W$ between two $\algebra$-modules is called an \emph{$\algebra$-homomorphism} if, for all $\pi\in\algebra$ and $v\in V$, we have $\theta(\pi v) = \pi(\theta(v))$.
The following result is basic for such homomorphisms:

\begin{lem}[Schur's Lemma]
Assume $\theta\colon V\to W$ is an $\algebra$-homomorphism between two irreducible $\algebra$-modules $V$ and $W$. Then, $\theta=0$ if $V$ and $W$ are not isomorphic.  Otherwise, $\theta$ is uniquely defined up to a scalar multiplier.
\end{lem}

Let $M$ denote the complex vector space with the set of subsets of $[n]$ as its orthonormal basis and with the group action $\pi A = \pi(A)$, where $\pi\in\group$, $A\subseteq[n]$, and $\pi(A)$ denotes the  image of the set $A$ under the transformation $\pi$.
We call $\{A\}_{A\subseteq[n]}$ the \emph{standard basis} of $M$.

The module $M$ naturally decomposes into a direct sum $M = \bigoplus_{\ell=0}^n M_\ell$, where $M_\ell$ is spanned by the subsets of cardinality $\ell$.%
\footnote{
In terms of~\cite{sagan:symmetricGroup}, $M_\ell$ is isomorphic to the permutation module corresponding to the partition $(n-\ell',\ell')$ of $n$, where $\ell'=\min\{\ell,n-\ell\}$.}
The following lemma describes the decomposition of $M_\ell$ into irreducible submodules $S_\ell(t)$ (for different values of $t$), which will be isomorphic to the Specht module $S(t)$ corresponding to the partition $(n-t,t)$ of~$n$.%  The Specht modules with different values of $t$ are not isomorphic.

In the formulation of the lemma and later we use $\otimes$ to denote disjoint union of subsets of~$[n]$, extended by linearity, so for example 
$(\{1\}-\{2\})\otimes (\{3\}-\{4\}) = \{1,3\} - \{1,4\} - \{2,3\} +\{2,4\}$.

\begin{lem}
\label{lem:specht}
The $\algebra$-module $M_\ell$ has the following decomposition into irreducible submodules:
\(
M_\ell = \bigoplus_{t=0}^{\ell'} S_\ell(t),
\)
where $\ell'=\min\{\ell,n-\ell\}$, and each $S_\ell(t)$ is isomorphic to $S(t)$.  The submodule $S_\ell(t)$ is spanned by the vectors
\begin{equation}
\label{eqn:vkDef}
v_\ell(t,a,b) = 
\sA[\{a_1\}-\{b_1\}]\otimes \cdots\otimes \sA[\{a_t\}-\{b_t\}] \otimes \sC[\sum_{A\subseteq [n]\setminus\{a_1,\dots,a_t,b_1,\dots,b_t\}\colon |A|=\ell-t} A ]
\end{equation}
defined by disjoint sequences $a=(a_1,\dots,a_t)$ and $b=(b_1,\dots,b_t)$ of pairwise distinct elements of~$[n]$.
The dimension of $S(t)$ is $\binom{n}{t}-\binom{n}{t-1}$.

There is a unique (up to a scalar) $\algebra$-isomorphism between $S_\ell(t)$ and $S_m(t)$. We can choose the scalar so that the isomorphism maps each vector $v_\ell(t,a,b)$ to the corresponding $v_m(t,a,b)$.
\end{lem}

The lemma follows from general theory~\cite[Sections 2.9 and 2.10]{sagan:symmetricGroup}.  
The Appendix of~\cite{belovs:learningSymmetricJuntas} contains a short proof, see also \rf(rem:iso) below.
\rf(fig:subspaces) depicts the different subspaces involved in the decomposition of~$M$.

\begin{figure}[htb]
%$\begin{array}{ccccl}
  %&             &        &             & \vdots\\
  %&             &        & \nearrow    & \\
  %&             &        &             & \\
  %&             & M_0    &             & S_\ell(0)\\
  %& \nearrow    & \vdots & \nearrow    & \vdots\\
%M & \rightarrow & M_\ell  & \rightarrow & S_\ell(t), \mbox{ with orthonormal basis }\{e_\ell(t,x)\}_x,~~S_\ell(t)\cong S(t))\\
  %& \searrow    & \vdots & \searrow    & \vdots\\
  %&             & M_n    &             & S_\ell(\ell'), \mbox{ where }\ell'=\min\{\ell,n-\ell\}\\
  %&             &        &             & \\
  %&             &        & \searrow    & \\
  %&             &        &             & \vdots
%\end{array}$
%\caption{Decomposition of $M$ into irreducible modules $S_\ell(t)$}\label{fig:subspaces}
\begin{center}
\includegraphics[width=0.5 \textwidth]{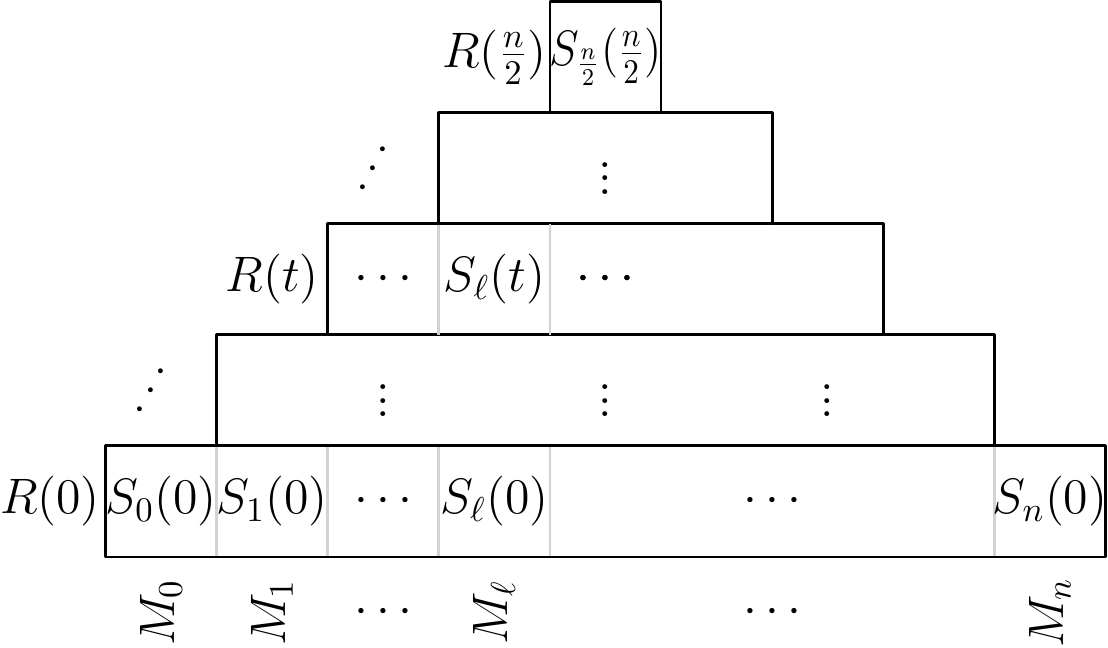}
\end{center}
\caption{Decomposition of $M$}\label{fig:subspaces}
\end{figure}

\mycommand{tv}{\tilde v}
Let $\tv_\ell(t,a,b)$ denote the normalized vector $v_\ell(t,a,b)/\|v_\ell(t,a,b)\|$, 
that is 
\begin{equation}
\label{eqn:normalize}
v_\ell(t,a,b) = \sqrt{2^t\binom{n-2t}{\ell-t}}\; \tv_\ell(t,a,b).
\end{equation}
Also, let $\vartheta_{t\to\ell}\colon S_t(t)\to S_\ell(t)$ denote the isomorphism from \rf(lem:specht) given by $\tv_t(t,a,b)\mapsto \tv_\ell(t,a,b)$.  This is a unitary transformation.

For each $t$, we choose an orthonormal basis $\{e_t(t,x)\}_x$ of $S_t(t)$. Also, let $e_\ell(t,x) = \vartheta_{t\to\ell}\; e_t(t,x)$, so that $\{e_\ell(t,x)\}_x$ is an orthonormal basis of $S_\ell(t)$.
The precise choice of the basis of $S_t(t)$ is irrelevant, but it is important that the bases of $S_\ell(t)$ for various $\ell$ are synchronized via the isomorphism $\vartheta_{t\to\ell}$.
The set $\{e_\ell(t,x)\}_{\ell,t,x}$ forms an orthonormal basis of $M$, which we call the \emph{Fourier basis}. Let $R(t)$ denote the submodule $\bigoplus_{\ell=t}^{n-t} S_\ell(t)$ of $M$.

\begin{clm}
\label{clm:basis}
In the Fourier basis, any $\algebra$-homomorphism from $M$ to itself is of the form 
$\bigoplus_{t=0}^{\floor[n/2]} A_t\otimes I_{S(t)}$, where $A_t\otimes I_{S(t)}$ acts on $R(t)$, $A_t$ is an $(n-2t+1)\times(n-2t+1)$ matrix, and $I_{S(t)}$ denotes the identity operator on $S(t)$.
\bnote{Removed $\reg L$ and $\reg B$ are not defined here: That is, in the Fourier basis, $A_t$ acts on the register $\reg L$, and $I_{S(t)}$ acts on $\reg B$.}
\end{clm}

\pfstart
By Schur's lemma, any $\algebra$-homomorphism maps each irreducible module to an isomorphic one.  Hence, each $R(t)$ is mapped to itself.  
Also, as $\vartheta_{t\to\ell}$ is the only isomorphism between $S_t(t)$ and $S_\ell(t)$, we see by Schur's lemma that a vector $e_\ell(t,x)$ is mapped into a linear combination of the vectors $\{e_m(t,x)\}_m$.  Thus, in $R(t)$, the homomorphism has the form $A_t\otimes I_{S(t)}$ for some~$A_t$.
\pfend

\paragraph{The quantum Fourier transform of $M$.}
Let the register $\reg A$ have $M$ as its vector space with the standard basis.
Let also $\reg T$ and $\reg L$ be $(n+1)$-qudits (i.e., registers of dimension $n+1$), and let $n$-qubit register $\reg B$ store indices $x$ of the Fourier basis elements $e_\ell(t,x)$.
The \emph{Fourier transform} of $M$ is the following unitary map, for which it is convenient to use ket notation:
\begin{equation}
\label{eqn:QFT}
F\colon \ket T|t>\ket L|\ell>\ket B|x> \mapsto \ket A|e_\ell(t,x)>.
\end{equation}
Note that while $F$ is a unitary from a $2^n$-dimensional space to a $2^n$-dimensional space, it is convenient to use more than~$n$ qubits to represent the basis states (namely $n+2\lceil\log(n+1)\rceil$ qubits).

The Fourier transform segregates copies of non-isomorphic Specht modules of $M$ by assigning them different values of~$t$ in the register~$\reg T$.  For each Specht module $S(t)$, the copies are labeled by~$\ell$, which is the size (as a subset of~$[n]$) of the standard basis elements of~$M$ used by the copy.  Finally, the basis elements of $S(t)$ are indexed by $x$, the precise choice of which is irrelevant for our application.%
\footnote{%
For readers familiar with the usual formula of the QFT over $\bS_n$, given by
\( 
\ket|g>  \mapsto \sum_{\rho} \sqrt{\frac{d_\rho}{|G|}} \ket|\rho> \sum_{i,j} \rho(g)_{ij}\ket|i>\ket|j>,
\)
(where $g\in\bS_n$, $\rho$ runs over all irreps of $\bS_n$, and $d_\rho$ is the dimension of $\rho$):
the register $\reg T$ in our algorithm is similar to $\rho$, $\reg L$ is similar to $i$, and $\reg B$ to $j$.  The formula, however, does not exactly apply in our case, because we are performing QFT of the module $M$, not the regular $\bS_n$-module, for which the latter formula applies. In other words, our QFT acts on the $2^n$-dimensional space spanned by the subsets of $[n]$, not the $n!$-dimensional space spanned by all permutations.}
The next theorem, which we prove in \rf(sec:QFT), shows our QFT can be implemented efficiently.

%For each $t$, we will choose an orthonormal basis $\{e_t(t,x)\}_{x\in[\dim S_t(t)]}$.
%Also, let $e_\ell(t,x) = \vartheta_{t\to\ell}\; e_t(t,x)$.
%The set $\{e_\ell(t,x)\}_{\ell,t,x}$ forms an orthonormal basis of $M$, which we call the \emph{Fourier basis}; this will be a basis in which \rf(clm:LambdaDecompose) holds. The following theorem shows there is a very efficient algorithm to implement the corresponding QFT (and hence its inverse as well). 

\begin{thm}
\label{thm:qft}
There exists a quantum algorithm with time complexity $\tO(n)$ that implements the map from~\rf(eqn:QFT) for all choices of $t$, $\ell$ and $x$ for which the last expression is defined, up to an error in the operator norm that can be made smaller than any inverse polynomial in~$n$.
\end{thm}

\paragraph{Decomposing $\Lambda$ in terms of representations.}
The main observation behind our implementation of $R_\Lambda$ is that $\Lambda$ is invariant under the permutation of elements, hence it is an $\algebra$-homomorphism from $M$ to itself, and thus subject to the decomposition of \rf(clm:basis).  In fact, it is not hard to obtain the matrices $A_t$ in this decomposition.

\mycommand{tw}{\tilde w}
\begin{clm}
\label{clm:LambdaDecompose}
The image of $\Lambda$ contains at most one copy of each $S(t)$ for $t=0,\dots,k$.
In the Fourier basis, $\Lambda$ has the following form: $\bigoplus_{t=0}^k (\tw_t\tw_t^*)\otimes I_{S(t)}$, where $\tw_t$ is the normalized version of the $(n-2t+1)$-dimensional vector $w_t$ that is given by
\begin{equation}
\label{eqn:wtell}
w_t\elem[\ell] = \alpha_\ell \binom{n-\ell-t}{k-t} \sqrt{\binom{n-2t}{\ell-t}}\;,
\end{equation}
for $\ell \in \{t,\dots, n-k\}$, and $w_t\elem[\ell]=0$ for $\ell \in \{n-k+1,\dots, n-t\}$. (If $w_t$ is the 0-vector, we assume that $\tw_t$ is the 0-vector as well.)
% and the direct sum is over all $t\in\{0, 1, \ldots, k\}$ for which $w_t\neq \overrightarrow{0}$.
\end{clm}

\pfstart
%\onote{5/21: old text:
%We first assume that $\alpha_{n-k}\ne 0$.  In this case, $\Lambda$ contains exactly one copy of each $S(t)$ with $0\le t\le k$:  The restrictions of the vectors $\psi_T$ to $M_{n-k}$ span the whole module $M_{n-k}$, and as we know from \rf(lem:specht), the latter does contain a copy of each $S(t)$. Since $S(t)$ has dimension $\binom{n}{t}-\binom{n}{t-1}$, the direct sum of these copies of $S(t)$ already has dimension $\sum_{t=0}^k{n\choose t}-{n\choose t-1}={n\choose k}$. The span $\Lambda$ of the vectors $\psi_T$ has dimension at most ${n\choose k}$, hence $\Lambda$ cannot contain more than one copy of any of these~$S(t)$.
%%
%Because $\Lambda$ is an $\algebra$-homomorphism, by \rf(clm:basis), $\Lambda$ has the form $\bigoplus_{t} A_t\otimes I_{S(t)}$ in the Fourier basis.  From the last paragraph, for each $t\le k$, $A_t$ is either a rank-1 projector (if $\Lambda$ contains $S(t)$), or the zero matrix.  Thus, $\Lambda$ actually has the form $\bigoplus_{t=0}^k (\tw_t\tw_t^*)\otimes I_{S(t)}$ for some vectors $\tw_t$.
%}
%
%\onote{5/21: new paragraph:}
We first assume that $\alpha_{n-k}\ne 0$.  
Because $\Lambda$ is an $\algebra$-homomorphism, by \rf(clm:basis), $\Lambda$ has the form $\bigoplus_{t} A_t\otimes I_{S(t)}$ in the Fourier basis.  
We claim that $\Lambda$ contains exactly one copy of each $S(t)$ with $0\le t\le k$, i.e., 
all $A_t$ are rank-1 projectors: indeed, the projection of $\Lambda$ on $M_{n-k}$ is surjective, 
and as we know from \rf(lem:specht), the latter does contain a copy of each $S(t)$. Since $S(t)$ has dimension $\binom{n}{t}-\binom{n}{t-1}$, the direct sum of these copies of $S(t)$ already has dimension $\sum_{t=0}^k \s[\binom nt-\binom n{t-1}]=\binom nk$. 
On the other hand, $\Lambda$ clearly has dimension at most $\binom nk$, hence $\Lambda$ cannot contain more than one copy of any of these~$S(t)$.
Thus, $\Lambda$ actually has the form $\bigoplus_{t=0}^k (\tw_t\tw_t^*)\otimes I_{S(t)}$ for some vectors $\tw_t$.

It remains to find the coefficients of $\tw_t$.  Take a vector $v_{n-k}(t,a,b)\in M_{n-k}$ for some sequences $a$ and $b$, and act on it with the linear transformation that maps a basis vector $T\in M_{n-k}$ to the vector $\psi_T\in M$.
The resulting vector $u$ is clearly in the image of $\Lambda$.  We claim that
\begin{equation}
\label{eqn:LambdaDecompose1}
u = \sum_{\ell=t}^{n-k} \alpha_\ell \binom{n-\ell-t}{k-t} v_\ell(t,a,b) = \sum_{\ell=t}^{n-k} \alpha_\ell \binom{n-\ell-t}{k-t} \sqrt{2^t\binom{n-2t}{\ell-t}} \tv_\ell(t,a,b).
\end{equation}
%\onote{5/21: it's unusual to have a proof environment without a claim environment attached; I can live with it}
%\bnote{I don't think that a proof environment is called for here.}
To prove the first equality of \rf(eqn:LambdaDecompose1), consider a basis element $A\in M_\ell$ for an $\ell \in\{t, \ldots, n-k\}$ and look at its coefficient in~$u$.  
If $A$ contains both $a_i$ and $b_i$ for some $i$, it appears in none of the $\psi_T$ of which $u$ consists.
If $A$ avoids both $a_i$ and $b_i$ for some $i$, then any coefficient it gets from $\psi_T$ is cancelled by the coefficient it gets from $\psi_{T\bigtriangleup \{a_i,b_i\}}$, where $\bigtriangleup$ stands for symmetric difference.
Finally, if $A$ uses exactly one of each $a_i,b_i$, it appears in $\psi_T$ for $\binom{n-\ell-t}{k-t}$ choices of $T$, with coefficient equal to its coefficient in $v_\ell(t,a,b)$ times $\alpha_\ell$ in each. This establishes the first equality of~(\ref{eqn:LambdaDecompose1}). The second equality in~\rf(eqn:LambdaDecompose1) follows immediately from~\rf(eqn:normalize).
\medskip

Now let $v=(v_x)$ be the vector of coefficients of the representation of $\tv_t(t,a,b)$ in the basis $\{e_t(t,x)\}_x$, i.e., $\tv_t(t,a,b)=\sum_x v_x e_t(t,x) = F\sA[ \ket T|t>\ket L|t>\ket B|v> ]$.
Then, by our choice of orthonormal basis, 
$F^{-1}(\tv_\ell(t,a,b))=\ket T|t>\ket L|\ell>\ket B|v>$
for $\ell \in \{t,\dots, n-k\}$, 
and from~\rf(eqn:LambdaDecompose1), we have
\[
F^{-1}\s[u/\sqrt{2^t}] = \sum_{\ell=t}^{n-k} \alpha_\ell \binom{n-\ell-t}{k-t} \sqrt{\binom{n-2t}{\ell-t}} F^{-1}(\tv_\ell(t,a,b))=\ket T|t>\ket L|w_t>\ket B|v>,
\]
where $w_t$ is defined in~\rf(eqn:wtell).
As $u$ is in the image of $\Lambda$ and $F^{-1}u$ is $u$ represented in the Fourier basis, $w_t$ must be proportional to $\tw_t$.

%The $\algebra$-homomorphisms between different copies of $S(t)$ map 
%$\tv_\ell(t,a,b)$ to $\tv_{\ell'}(t,a,b)$ with different $\ell'$ but the same $t, a, b$.
%Therefore, in the basis of \rf(clm:basis), the right-hand side\rnote{fix this} can be expressed as
%$w_t\otimes v$ for some $v$. Taking the span of such vectors for all choices of $a, b$ gives
%$w_t\otimes S(t)$. 
%he coefficients of $w_t$ are the coefficients of $u/\sqrt{2^t}$ in the basis $\{\tv_\ell(t,a,b)\}_\ell$.
%To get $\tilde w_t\elem[\ell]$, note that $\|v_\ell(t,a,b)\| = \sqrt{2^t{n-2t\choose \ell-t}}$, and also all the coefficients can be simultaneously divided by $\sqrt{2^t}$ (that is, $u$ scaled down by $\sqrt{2^t}$). \rnote{28/3: you mean as integers, inside the square-root?} \bnote{29/3: I don't get it.}

Now consider the case $\alpha_{n-k}=0$.  In this case, change $\alpha_{n-k}$ to an arbitrary non-zero value, and perform the above calculations for the resulting projector $\Lambda'$.
The result follows from the observation that the image of $\Lambda$ is a projection of the image of $\Lambda'$ onto $\bigoplus_{\ell<n-k} M_\ell$.
\pfend

\begin{rem}
\label{rem:iso}
Note that~\rf(eqn:LambdaDecompose1) essentially proves the second part of \rf(lem:specht).
Indeed, it shows that the transformation $v_{n-k}(t,a,b)\mapsto v_\ell(t,a,b)$ is linear (which is not obvious from~\rf(eqn:vkDef)).  It is clear that it is invariant under the action of $\bS_n$, hence, it is an isomorphism between the copies of $S(t)$ in $M_{n-k}$ and $M_\ell$.
\end{rem}

\paragraph{Implementing $R_\Lambda$.}
Having the efficient implementation of the QFT of \rf(thm:qft) (proved in the next section), it is straightforward to implement the translation $R_\Lambda$ up to polynomially small error.
First, we run the QFT of \rf(thm:qft) in reverse, and obtain the representation of $M$ in the Fourier basis.
In this basis, by \rf(clm:LambdaDecompose), our task boils down to the reflection about the vector $\tw_t$ in the register $\reg L$, where $t$ is the value of the register $\reg T$. We show how to implement this reflection below. The register $\reg B$ is not used in this reflection. After that, we run the QFT of \rf(thm:qft).

\paragraph{Reflection about $\tw_t$.}
Reflection about $\tw_t$ is done using standard techniques.
We describe the transformation for a fixed $t$, but a key point is that they can all be done in superposition (conditioned on~$t$), so that the total time complexity of this step is the same as for a fixed value of $t$.

We may assume $\tw_t$ is a non-zero vector. To reflect about~$\tw_t$, it suffices to give an $\tO(n)$-time unitary~$U_t$ that maps $\ket L|t>$ to $\tw_t$: the reflection about $\tw_t$ is then implemented by first applying $U_t^{-1}$, then reflecting about $\ket L|t>$ (which is easy), and then applying $U_t$. 
To implement this map $U_t$, we start in the state $\ket L|t>$, and run through all the values of $\ell$ from $t$ to $n-k$.  In the course of this process, we gradually generate states of the form
\[
\tw_t\elem[t]\; \ket L|t> + \dots + \tw_t\elem[\ell-1]\; \ket L|\ell-1> + \sqrt{R}\; \ket L|\ell>,
\]
where $R = 1-\sum_{i=t}^{\ell-1} \absA|\tw_t\elem[i]|^2$ \bnote{added $|.|$} is the remaining weight to use.
At each step, we perform the transformation
\begin{equation}
\label{eqn:QFTstep}
\sqrt R\; \ket L|\ell> \mapsto \tw_t\elem[\ell]\; \ket L |\ell> + \sqrt{R'} \ket L|\ell+1>,
\end{equation}
where $R' = 1-\sum_{i=t}^{\ell} \absA|\tw_t\elem[i]|^2$.
To help with this, we keep two registers $\reg C$ and $\reg R$.
In $\reg C$ at step $\ell\ge t$, we store the number
\[
c_\ell = \frac{\binom{n-\ell-t}{k-t} \sqrt{\binom{n-2t}{\ell-t}}}{\|w_t\|}
\]
so that $\tw_t\elem[\ell] = \alpha_\ell c_\ell$ by~\rf(eqn:wtell).
In the register $\reg R$, we store $R$.  Both are stored as floating point real numbers with $\Theta(\log n)$ bits of precision.
\bnote{Added floating point}
With these values, and knowing $\alpha_\ell$, 
the transformation~\rf(eqn:QFTstep) is a routine polylogarithmic computation.

The registers $\reg C$ and $\reg R$ are manipulated as follows.  
We start by computing $\|w_t\|$ and then $c_t$.
We then upload the precomputed value of $c_t$ into $\reg C$, set $\ell=t$ and set $\reg R$ to 1.
After each step (\ref{eqn:QFTstep}) we use
\[
\frac{c_{\ell+1}}{c_\ell} = \frac{n-\ell-k}{\sqrt{(n-\ell-t)(\ell+1-t)}} ,
\]
to replace $c_\ell$ with $c_{\ell+1}$ in the register $\reg C$, in polylogarithmic time.
We also update $\reg R$ (in the straightforward way).
At the end, when $\ell=n-k$, we uncompute the values of $\reg C$ and $\reg R$.

It suffices to keep track of $c_\ell$ up to $A \log n$ relative bits of precision, for an appropriate $A$.  This can be achieved as $c_\ell$ is only multiplied by a constant at each step, and is never added or subtracted.
Then, $R$ is accurate within an absolute error of $O(1/n^{A-1})$ (to perform the entire reflection, 
we subtract $O(n)$ values $\absA|\tw_t\elem[i]|^2 = (\alpha_\ell c_\ell)^2$ from $R$, with absolute error $O(1/n^A)$ in each of those values).
This means that the transformation (\ref{eqn:QFTstep}) produces amplitude for $\ket L |\ell>$ that is within
$O(1/n^{A-1})$ of the correct amplitude $\tw_t\elem[\ell]$. 
Since this is achieved for every $\ell$ and there are $O(n)$ values of $\ell$, the state that is produced by
our algorithm is within Euclidean distance $O(1/n^{A-3/2})$ of the correct one.
The error analysis for the computation of $\|w_t\|$ at the beginning of the algorithm is similar.

Thus, each step takes time $\polylog(n)$, and there are $O(n)$ steps in total.
Computing $\|w_t\|$ and $c_t$ at the beginning 
and removal of the final values of $\reg C$ and $\reg R$ at the end also takes time $\tO(n)$.  
As both applications of the QFT take time $\tO(n)$ as well, the total running time of our implementation of the reflection $R_\Lambda$ is $\tO(n)$.

\subsection{Quantum Fourier Transform of the module \texorpdfstring{$M$}{M}}
\label{sec:QFT}
\mycommand{down}{\!\!\downarrow}
\bnote{05.07: Added the following}
In this section, we prove \rf(thm:qft).  Our algorithm is essentially the same as the quantum Schur-Weyl transform in~\cite{bch:prl04,bch:soda07}.  
However, the algorithm in~\cite{bch:prl04,bch:soda07} is defined for another group (a product of a general linear group and a symmetric group), and we have to verify that the algorithm really performs the transformation in~\rf(eqn:QFT). Therefore we will give a complete description of the algorithm here.

We will define our algorithm inductively in~$n$, so we add a superscript $n$ to the notations of the previous section in order to specify the value of $n$, e.g., $M^n, S^n(t), v^n_\ell(t,a,b)$ and so on.
%Next, let $\tilde v^n_\ell(t,a,b)$ denote the normalized vector $v^n_\ell(t,a,b)/\|v^n_\ell(t,a,b)\|$.

\paragraph{Basis for modules $S_t^n(t)$.}
We now define the orthonormal basis $\{e^n_t(t,x)\}_x$ of $S_t^n(t)$, for all $n, t\in \mathbb{N}$ 
such that $n\geq 2t$. %(If $n<2t$, the module $S_t^n(t)$ is not defined because then $(n-t, t)$ is not a partition of $n$ with elements of the partition sorted in decreasing order.) 
As in~\cite{bch:prl04,bch:soda07,KawanoS13}, our choice is a variant of the Gelfand-Tsetlin basis.
The index $x$ is encoded as a binary string of length $n$
(not all binary strings correspond to basis elements though), and the definition of $e^n_t(t,x)$ is as follows.

\mycommand{subalgebra}{\bS_{n-1}}
First, we define $e^n_0(0,0^n)$ as the basis element $\emptyset$ of $S^n_0(0)$.
For $t>0$, we define $e^n_t(t,x)$ in terms of $e^{n-1}_t(t,x)$.
For that, we first need to understand the interplay between $\algebra$- and $\subalgebra$-modules.

Since $\bS_{n-1}$ is a subgroup of $\bS_n$, we can look at the action of $\bS_{n-1}$ on $S^n(t)$. Let $S^n(t)\down_{\bS_{n-1}}$ denote $S^n(t)$, considered as an $\bS_{n-1}$-module.
The branching rule~\cite[Section~2.8]{sagan:symmetricGroup} says that $S^n(t)\down_{\bS_{n-1}}$ is $\bS_{n-1}$-isomorphic to  the direct sum $S^{n-1}(t) \oplus S^{n-1}(t-1)$.  (If $n=2t$, the $S^{n-1}(t)$ term is absent in the last expression, as $n-1<2t$.)

Let $\tau_0$ be the $\subalgebra$-isomorphism between $S_t^{n-1}(t)$ and the copy of $S^{n-1}(t)$ in $S^n_t(t)\down_{\bS_{n-1}}$, and $\tau_1$ be the $\subalgebra$-isomorphism between $S_{t-1}^{n-1}(t-1)$ and the copy of $S^{n-1}(t-1)$ in $S^n_t(t)\down_{\bS_{n-1}}$.  We normalize $\tau_0$ and $\tau_1$ so that they both become unitary operators.
We define $e^n_t(t,x0) = \tau_0 e^{n-1}_t(t,x)$ and $e^n_{t}(t,x1) = \tau_1 e^n_{t-1}(t-1,x)$, where $x$ runs through the basis elements of $S^{n-1}_t(t)$ and $S^{n-1}_{t-1}(t-1)$, respectively.
Thus, $\{e^n_t(t,x)\}_x$ is an orthonormal basis of $S_t^n(t)$, because 
$\tau_0$ and $\tau_1$ are unitary, and their images are orthogonal as vector spaces.

Let us describe the isomorphisms $\tau_0$ and $\tau_1$ explicitly.
If $n>2t$, the isomorphism $\tau_0$ is given by $\tv^{n-1}_t(t,a,b)\mapsto \tv^{n}_t(t,a,b)$.
It is more convenient to define $\tau_1$ as the $\subalgebra$-isomorphism between $S^{n-1}_{t}(t)$ and the copy of $S^{n-1}(t)$ in $S^n_{t+1}(t+1)\down_{\bS_{n-1}}$.
A non-normalized version of the isomorphism is given by
\begin{equation}
\label{eqn:QFTisomorphism}
v^{n-1}_{t}(t,a,b)\mapsto 
\sum_{i\in [n-1]\setminus\{a_1,\dots,a_{t},b_1,\dots,b_{t}\}} \sA[\{a_1\}-\{b_1\}]\otimes \cdots\otimes \sA[\{a_{t}\}-\{b_{t}\}] \otimes \sA[\{i\}-\{n\}] .
\end{equation}
Indeed, this is a linear combination of linear transformations $v^{n-1}_t(t,a,b)\mapsto v^{n-1}_{t+1}(t,a,b)$ and $v^{n-1}_t(t,a,b) \mapsto v^{n-1}_t(t,a,b)\otimes\{n\}$.  It is also clearly invariant under the action of $\bS_{n-1}$.  $\tau_1$ is the normalized version of this mapping, i.e., it maps $\tv^{n-1}_t(t,a,b)$ into the normalized vector on the right-hand side of~\rf(eqn:QFTisomorphism).

\paragraph{Identity between basis elements.}
Recall that $e_\ell(t,x) = \vartheta_{t\to\ell}\; e_t(t,x)$, where $\vartheta_{t\to\ell}$ is defined after~\rf(eqn:normalize).
\begin{clm}
\label{clm:QFTe}
We have
\begin{equation}
\label{eqn:QFT1}
e^n_\ell(t,x0) =
\sqrt{\frac{n-\ell-t}{n-2t}}\; e^{n-1}_\ell(t,x) + \sqrt{\frac{\ell-t}{n-2t}}\; e^{n-1}_{\ell-1}(t,x)\otimes \{n\}
\end{equation}
and
\begin{equation}
\label{eqn:QFT2}
e^n_\ell(t+1,x1) =
\sqrt{\frac{\ell-t}{n-2t}}\; e_\ell^{n-1} (t,x) - \sqrt{\frac{n-\ell-t}{n-2t}}\; e_{\ell-1}^{n-1} (t, x)\otimes\{n\},
\end{equation}
whenever the vectors on the left-hand side are defined.
\bnote{03.07: Commented in}\rnote{6/7: added the $\ell=0$ case}
(If $\ell=0$, $t=\ell$ or $n=t+\ell$, the right-hand sides of~\rf(eqn:QFT1) and~\rf(eqn:QFT2) contain only one term.)
Here we identify elements $S\subseteq[n-1]$ of the standard basis of $M^{n-1}$ with the corresponding basis elements of $M^n$.
\end{clm}

\pfstart
To prove~\rf(eqn:QFT1), it suffices to show that for any two disjoint sequences $a=(a_1,\dots,a_t)$ and $b=(b_1,\dots,b_t)$ of distinct elements of $[n-1]$, we have
\[
\tilde v^n_\ell(t,a,b) = 
\sqrt{\frac{n-\ell-t}{n-2t}}\; \tilde v^{n-1}_\ell(t,a,b) + \sqrt{\frac{\ell-t}{n-2t}}\; \tilde v^{n-1}_{\ell-1}(t,a,b)\otimes \{n\},
\]
as~\rf(eqn:QFT1) then follows by linearity.
\rnote{15/5: The last part goes too fast: it's not obvious how the normalization of the $v_\ell$ and the linearity of (27) interact.  Add a calculation to show that arguing about the $\tilde v_\ell$ vectors really suffices?} \bnote{No, this is obvious.  I changed how we define the basis elements to make this more transparent.}
The latter equality follows from~\rf(eqn:vkDef), as a uniformly random $(\ell-t)$-subset $A$ of $[n]\setminus \{a_1,\dots,a_t,b_1,\dots,b_t\}$ has probability $(\ell-t)/(n-2t)$ of including the element~$n$.

To prove~\rf(eqn:QFT2), let us apply $\vartheta^n_{t+1\to \ell}$ to the right-hand side of~\rf(eqn:QFTisomorphism).  The resulting vector is
\[
\sA[\{a_1\}-\{b_1\}]\otimes \cdots\otimes \sA[\{a_{t}\}-\{b_{t}\}]
\otimes\sD[ (\ell-t)\!\!\!\!\sum_{\substack{A\subseteq[n]\setminus \{a_1,\dots,a_t,b_1,\dots,b_t\}\\ |A|=\ell-t,\; n\notin A}}\!\!\!\! A 
\;-\; (n-\ell-t)\!\!\!\! \sum_{\substack{A\subseteq[n]\setminus \{a_1,\dots,a_t,b_1,\dots,b_t\}\\ |A|=\ell-t,\; n\in A}}\!\!\!\! A 
],
\]
because each $A\not\ni n$ appears once from each term $i\in A$ of the sum in~\rf(eqn:QFTisomorphism), and each $A\ni n$ appears once from each term $i\notin A$.  By an easy calculation, the normalized version of this vector is
\[
\sqrt{\frac{\ell-t}{n-2t}} \tilde v_\ell^{n-1} (t, a,b) - \sqrt{\frac{n-\ell-t}{n-2t}} \tilde v_{\ell-1}^{n-1} (t, a,b)\otimes\{n\},
\]
and~\rf(eqn:QFT2) again follows by linearity.
\pfend

\paragraph{Quantum Fourier transform.}
With \rf(clm:QFTe) in hand, it is easy to describe the algorithm of \rf(thm:qft).
Both $\reg A$ and $\reg B$ are represented as arrays of qubits $\reg A_1,\dots,\reg A_n$ and $\reg B_1,\dots,\reg B_n$, respectively.  A standard basis element $S\in M^n$ is represented as $\ket A|S>$, where $\reg A_i$ contains 1 iff $i\in S$.
A string $x=(x_i)\in\cube$ is represented as $\ket B|x>$, where $\reg B_i$ stores $x_i$.
The transformation is described in \rf(alg:QFT).

\begin{algorithm}[htb]
\algcaption{Quantum Fourier Transform of the module $M^n$\\
Given:& Positive integer $n$, and registers $\reg A$, $\reg B$, $\reg T$ and $\reg L$ as described above.\\
Action:& {Transformation $\ket T|t> \ket L|\ell>\ket B|x> \mapsto \ket A|e_\ell(t,x)>$}.}
\label{alg:QFT}
\enumstart
\item While $n>1${\bf :}
\negmedskip
	\enumstart
	\item Let $\reg T \gets \reg T - \reg B_n$
	\item Perform the unitary qubit transformation
	\[ 	
	\ket |0> \mapsto \sqrt{\frac{n-\ell-t}{n-2t}} \ket |0> + \sqrt{\frac{\ell-t}{n-2t}} \ket|1>,
	\qquad \ket |1> \mapsto \sqrt{\frac{\ell-t}{n-2t}} \ket|0> - \sqrt{\frac{n-\ell-t}{n-2t}} \ket|1>
	\]
	on $\reg B_n$, conditioned on the content $\ket T|t>\ket L|\ell>$ of the registers $\reg T$ and $\reg L$ (and on $2t<n$ and $t\le \ell\le n-t$.)
	\item Swap $\reg B_n$ and $\reg A_n$, let $\reg L \gets \reg L - \reg A_n$  and $n\gets n-1$.
	\enumend
\item Implement the necessary transformation in the case $n=1$.
\enumend
\end{algorithm}

The correctness of the algorithm can be proven by induction on $n$.  The base case $n=1$ is trivial.
Next, for a fixed value of $n$, the first iteration of the loop in Step~1 performs the transformation
\begin{align*}
{\ket T|t> \ket L|\ell>\ket B_n |0> } &\mapsto 
{\sqrt{\frac{n-\ell-t}{n-2t}} \ket T|t> \ket L|\ell> \ket A_n |0>}
+ {\sqrt{\frac{\ell-t}{n-2t}} \ket T|t> \ket L|\ell-1> \ket A_n |1> }
\\
{ \ket T|t+1> \ket L|\ell>\ket B_n |1>} &\mapsto 
{\sqrt{\frac{\ell-t}{n-2t}} \ket T|t> \ket L|\ell> \ket A_n |0>  }
- {\sqrt{\frac{n-\ell-t}{n-2t}} \ket T|t> \ket L|\ell-1> \ket A_n |1> }
\end{align*}
By the inductive assumption, the remaining iterations of the loop and Step 2 correctly perform the Quantum Fourier Transform of $M^{n-1}$, hence, the whole algorithm performs the Quantum Fourier Transform of $M^n$ due to \rf(clm:QFTe).

Using standard techniques, each iteration of the loop can be performed with sufficiently high precision in time $\polylog(n)$. 
%\rnote{31/3: added the following sentence}
Specifically, the one-qubit transformation for part~(b) can be implemented (conditioned on $t$ and $\ell$) with sufficiently high precision using $\polylog(n)$ gates from any universal set of elementary gates.
The error (in the operator norm) of this algorithm can be made smaller than any inverse polynomial.
The second step takes time $O(1)$.  
The total time complexity of the algorithm is thus $\tO(n)$. 

\section{Quantum lower bound for junta testing}
\label{sec:lower}
Let us assume that $\eps = \Omega(1)$.
Tight classical lower bounds on junta testing~\cite{chockler:testingJuntasLower, blais:testingLowerViaCommunication} are based on distinguishing a $k$-junta from a function that depends on $k+O(1)$ variables.  As noted by \atici and Servedio~\cite{atici:testingJuntas}, this approach is doomed in the quantum setting because these two cases can be distinguished in $O(\log k)$ quantum queries as follows.
For a function that depends on only $k+O(1)$ variables but is far from any $k$-junta, it follows from \rf(lem:weightonextravars) that at least one of the $O(1)$ ``extra'' variables has $\Omega(1)$ influence. Hence there exists a subset $S\subseteq[n]$ of $k+1$ variables each having influence $\Omega(1)$. Each of those $k+1$ variables will occur in a Fourier Sample with constant probability, so the probability that a fixed variable from $S$ is not seen in $t$ Fourier Samples is exponentially small in $t$.  By the union bound, after $t=O(\log k)$ Fourier Samples, with high probability all $k+1$ variables of $S$ will have been seen and we can conclude the function is not a $k$-junta. 

Instead of this, \atici and Servedio presented a different approach based on distinguishing a $k$-junta from a function that depends on $k+\Omega(k)$ variables.  Using this technique, they proved an $\Omega(\sqrt{k})$ lower bound for a special class of \emph{non-adaptive} quantum algorithms.  

In this section, we give an explicit description of the \atici-Servedio construction, and use it to prove a quantum lower bound for the junta testing problem.  Consider the following problem.

\begin{defn}[Testing the image size]
\label{defn:rangeSize}
An \emph{image size tester}, given oracle access to a function $g\colon [m]\to[n]$, is required to distinguish \bnote{removed: with bounded error} whether the image of $g$ is of size at most $\ell$, or $g$ is $\eps$-far away from any such function. 
\end{defn}

It turns out that a junta tester can be used to solve this problem.  The connection is through the following ancillary function.

\begin{defn}[Addressing function]
Assume that $m$ is a power of two, and $g\colon [m]\to[n]$ is a function.  We define the corresponding \emph{addressing function} $f\colon \bool^{n + \log m}\to \sbool$ as follows.
Interpret the input string $x$ of $f$ as a concatenation $yz$ with $y\in\cube$ and $z\in\{0,1\}^{\log m} = [m]$.  Then, $f(x) = (-1)^{y_{g(z)}}$.
The variables in $y$ are called \emph{addressed variables}, and the variables in $z$ are called the \emph{address variables}.
\end{defn}

It is easy to see that a quantum query to $f$ can be simulated by two quantum queries to $g$:  one to compute $g(z)$, and one to uncompute it.

\begin{lem}
\label{lem:lowerGeneral}
For a function $g\colon[m]\to[n]$ with $m$ a power of $2$, let $f\colon \bool^{n + \log m}\to \sbool$ be the corresponding addressing function.
Let $\ell \ge 1$ be an integer and define $k=\ell+\log m$.  
If the size of the image of $g$ does not exceed~$\ell$, then $f$ is a $k$-junta.
Conversely, if $g$ is $\eps$-far from any function with an image of size at most $\ell$,
then $f$ is $\eps'/2$-far from any $k$-junta where $\eps'=\eps - (\log m)/k$. 
\end{lem}

\pfstart
The first statement is obvious. 
So assume $g$ is $\eps$-far from any function with an image of size at most $\ell$. 
We claim that $g$ is also $\eps'$-far from any function 
with an image of size at most $k$. Indeed, if $h$ is a function with image 
of size at most $k$, we can reduce its image to be of size at most $\ell$ 
by modifying it on at most a $(\log m)/k$ fraction of inputs corresponding to the ``least popular'' 
outputs. 

In order to show that $f$ is $\eps'/2$-far from any $k$-junta,
take an arbitrary $k$-subset $W\subseteq[n+\log m]$,
and any Boolean function $h$ depending only on the variables in $W$.
We want to show that $f$ is $\eps'/2$-far from $h$. 
Indeed, by the previous claim, at least $\eps'$ fraction of the inputs to $g$ 
map to indices outside $W \cap [n]$. For any such $z \in [m]$, and any $y \in \cube$,
exactly one $x\in\{yz, y^{\oplus g(z)}z\}$ satisfies $f(x)\ne h(x)$ (where $x^{\oplus j}$ stands for $x$ with the $j$th bit flipped).  
Hence, the distance between $f$ and $h$ is at least $\eps'/2$.
\pfend

Let us now state some corollaries of this result.
First, we get an upper bound on the quantum query complexity of testing the support size.

\begin{cor}
If $\log m = o(\ell)$, the image size can be tested in $O\sA[\sqrt{\ell/\eps}\log\ell]$ quantum queries.
\end{cor}

More importantly, however, we get a lower bound on the quantum query complexity of junta testing.  This is based on the following well-known special case of the image size testing problem.

\begin{defn}[Collision Problem~\cite{brassard:collision}]
Let $m$ be an even integer, and $n\in\bN$.  In the \emph{collision problem}, one is given oracle access to a function $g\colon [m]\to[n]$, that is either 1-to-1 or 2-to-1.  The task is decide which is the case.
\end{defn}
Brassard \etal~\cite{brassard:collision} constructed a quantum $O(m^{1/3})$-query algorithm for the collision problem.  Later, Aaronson and Shi~\cite{shi:collisionLower} proved a matching lower bound:

\begin{thm}
\label{thm:collisionLower}
The bounded-error quantum query complexity of the collision problem is $\Omega(m^{1/3})$.
\end{thm}

If $g\colon[m]\to[n]$ is 2-to-1, then its image size is $m/2$ and the corresponding addressing function $f$ depends on only $m/2+\log m$ variables.  
On the other hand, if $g$ is 1-to-1, then its image size is $m$ and $f$ depends on $m+\log m$ variables.  
Moreover, it follows from \rf(lem:lowerGeneral) that $f$ is $1/5$-far from any 
$(m/2+\log m)$-junta if $m$ is large enough. 
Combined with \rf(thm:collisionLower), we get
\begin{thm}
Every quantum tester that distinguishes $k$-juntas from functions that are $1/5$-far from any $k$-junta with bounded error, needs to make $\Omega(k^{1/3})$ queries to the function.
\end{thm}

%The following function is important in their construction, and we will use it as well:
%
%
%
%
%\onote{5/20 removing: 
%The last observation has little to do with 
%the collision function, and can be stated in the following general form.}
%
%
%
%\onote{5/20: old statement:
%Let $\eps>0$ be fixed.
%Assume that there exists a quantum algorithm that $\eps$-tests $k$-juntas in $O(k^c)$ quantum queries for some fixed $c>0$.  Then, for any fixed $\eps'>2\eps$, there exists a quantum image size tester (as defined in \rf(defn:rangeSize)) that $\eps'$-tests whether the image size of a function $g\colon [m]\to[n]$ is at most $\ell$ using $O(\ell^c)$ queries, provided that $m$ is a power of two and $\log m = o(\ell)$.
%}

%\section{Fourier dimensionality}
%\label{sec:other}
%\input{other}

%\section{Quantum Algorithm for Sparsity Testing}
%\input{sparsity}

%\section{Alternative Support Size Tester}
%\label{sec:alternativeSupport}
%\input{sparsity}

\section{Conclusion and open problems}

In this paper we presented quantum algorithms for testing several well known properties of Boolean functions.
Our main result is a quantum algorithm for the $k$-junta testing problem with query complexity $O(\sqrt{k/\eps}\log k)$, and a time-efficient implementation of this based on a new near-linear time implementation of a shallow version of the quantum Fourier transform over the symmetric group.
The query complexity of our tester is almost quadratically better than the best previous quantum tester and also almost quadratically better than the best-possible classical tester. 
%We also presented quantum algorithms for testing Fourier dimensionality, Fourier sparsity and support size.

The topics for future work include:
\begin{enumerate}
\item
{\bf Better lower bound for junta testing.}
The main open question is: what is the actual quantum query complexity of this problem?  

We believe that the true answer is around $\sqrt{k/\eps}$ but it is quite challenging to improve our current lower bound of $\Omega(k^{1/3})$.
Nevertheless, we think that \rf(lem:lowerGeneral) may give a lower bound of $\Omega(k^{1/2-\delta})$ for any $\delta>0$.
In particular, we think that it should be possible to combine the lower bound construction by Raskhodnikova \etal~\cite{raskhodnikova:approximatingSupportSize} with two recent developments in quantum lower bounds: Zhandry's new machinery for the polynomial method~\cite{zhandry:howToConstruct}, which he applied to the collision and the set equality problems~\cite{zhandry:setEquality}, and Belovs's and Rosmanis's tight adversary lower bounds for the same functions~\cite{belovs:setEquality}.

\item
{\bf Better upper bound.}
Regarding the upper bound, we wonder if the $\log k$ factor can be removed.  
This question is essentially equivalent to finding a solution to the adversary bound for the \ggt problem that works \emph{for all values of $d$ simultaneously}.  
By this, we mean a feasible solution to~\rf(eqn:advNew) such that 
\[
\sum_{S\subseteq[n]} X_S\elem[f,f] = 
\begin{cases}
O(\sqrt{k}),& \text{if $f$ is an $\Inter_A$ function with $|A|=k$;}\\
O\s[\frac{\sqrt{k}}d],& \text{if $f$ is an $\Inter_A$ function with $|A|=k+d>k$.}
\end{cases}
\]
Note that our current solution does not satisfy this property because we use different rescaling for each value of $d$.  A different approach may be needed to obtain this property.
\item
{\bf Other applications of \qggt and our QFT.} \bnote{05.07: Updated here.}
Several of our algorithms are based on a quantum algorithm for a group testing problem, \qggt, which we find quite interesting in its own right, as it shows a quartic quantum-over-classical speedup. We think there might be more applications for \qggt waiting to be found.

%Another interesting tool that we developed in this paper is 
%the efficient quantum Fourier transform over part of the symmetric group (\rf(sec:lambda)).
%Does it have other applications?
\end{enumerate}

\subsection*{Acknowledgements.}
We thank Ashley Montanaro for getting some of us interested in quantum junta testers in the first place, and for initial discussions; Eric Blais for some discussions about monotonicity testing, and in particular for his suggestion to study the classical complexity of the GGT problem; Mark Zhandry for answering a question about~\cite{zhandry:setEquality}; Rocco Servedio for sending us a copy of~\cite{stw:adaptivityhelps}; Alexander Russell for helpful discussions about the quantum Fourier transform over the symmetric group; \rnote{6/7: added}Aram Harrow for answering a question about the Schur-Weyl transform; and Jeroen Zuiddam for some helpful comments. \bnote{05.07: Added}  We thank anonymous referees for many helpful suggestions and bringing our attention to references~\cite{bch:prl04, garciasoriano:phd}.

\bibliographystyle{habbrvM}
\bibliography{bib}

\begin{thebibliography}{10}

\bibitem{shi:collisionLower}
S.~Aaronson and Y.~Shi.
\newblock Quantum lower bounds for the collision and the element distinctness
  problems.
\newblock {\em Journal of the ACM},
  51(4):\myhref{http://dx.doi.org/10.1145/1008731.1008735}{595--605}, 2004.

\bibitem{ambainis:adv}
A.~Ambainis.
\newblock Quantum lower bounds by quantum arguments.
\newblock {\em Journal of Computer and System Sciences},
  64(4):\myhref{http://dx.doi.org/10.1006/jcss.2002.1826}{750--767}, 2002.
\newblock Earlier: \myhref{http://dx.doi.org/10.1145/335305.335394}{{\em
  STOC'00}},  \myhref{http://arxiv.org/abs/quant-ph/0002066}{{\ttfamily
  arXiv:quant-ph/0002066}}.

\bibitem{abblss:separations}
A.~Ambainis, K.~Balodis, A.~Belovs, T.~Lee, M.~Santha, and J.~Smotrovs.
\newblock Separations in query complexity based on pointer functions.
\newblock  \myhref{http://arxiv.org/abs/1506.04719}{{\ttfamily
  arXiv:1506.04719}}, 2015.

\bibitem{atici:testingJuntas}
A.~At{\i}c{\i} and R.~A. Servedio.
\newblock Quantum algorithms for learning and testing juntas.
\newblock {\em Quantum Information Processing},
  6(5):\myhref{http://dx.doi.org/10.1007/s11128-007-0061-6}{323--348}, 2007.
\newblock  \myhref{http://arxiv.org/abs/0707.3479}{{\ttfamily
  arXiv:0707.3479}}.

\bibitem{bch:prl04}
D.~Bacon, I.~Chuang, and A.~Harrow.
\newblock Efficient quantum circuits for {S}chur and {C}lebsch-{G}ordan
  transforms.
\newblock {\em Physical Review Letters},
  97(17):\myhref{http://dx.doi.org/10.1103/PhysRevLett.97.170502}{170502},
  2006.
\newblock  \myhref{http://arxiv.org/abs/quant-ph/0407082}{{\ttfamily
  arXiv:quant-ph/0407082}}.

\bibitem{bch:soda07}
D.~Bacon, I.~Chuang, and A.~Harrow.
\newblock The quantum {S}chur and {C}lebsch-{G}ordan transforms: {I}. efficient
  qudit circuits.
\newblock In {\em Proc.\ of 18th ACM-SIAM SODA}, pages 1235--1244, 2007.
\newblock  \myhref{http://arxiv.org/abs/quant-ph/0601001}{{\ttfamily
  arXiv:quant-ph/0601001}}.

\bibitem{Beals97}
R.~Beals.
\newblock Quantum computation of {F}ourier transforms over symmetric groups.
\newblock In {\em Proc.\ of 29th ACM STOC}, pages
  \myhref{http://dx.doi.org/10.1145/258533.258548}{48--53}, 1997.

\bibitem{belovs:learning}
A.~Belovs.
\newblock Span programs for functions with constant-sized 1-certificates.
\newblock In {\em Proc.\ of 44th ACM STOC}, pages
  \myhref{http://dx.doi.org/10.1145/2213977.2213985}{77--84}, 2012.
\newblock  \myhref{http://arxiv.org/abs/1105.4024}{{\ttfamily
  arXiv:1105.4024}}.

\bibitem{belovs:phd}
A.~Belovs.
\newblock {\em Applications of the Adversary Method in Quantum Query
  Algorithms}.
\newblock PhD thesis, University of Latvia, 2013.
\newblock  \myhref{http://arxiv.org/abs/1402.3858}{{\ttfamily
  arXiv:1402.3858}}.

\bibitem{belovs:learningSymmetricJuntas}
A.~Belovs.
\newblock Quantum algorithms for learning symmetric juntas via the adversary
  bound.
\newblock {\em Computational Complexity},
  24(2):\myhref{http://dx.doi.org/10.1007/s00037-015-0099-2}{255--293}, 2015.
\newblock Earlier: \myhref{http://dx.doi.org/10.1109/CCC.2014.11}{{\em
  CCC'14}},  \myhref{http://arxiv.org/abs/1311.6777}{{\ttfamily
  arXiv:1311.6777}}.

\bibitem{belovs:variations}
A.~Belovs.
\newblock Variations on quantum adversaries.
\newblock  \myhref{http://arxiv.org/abs/1504.06943}{{\ttfamily
  arXiv:1504.06943}}, 2015.

\bibitem{belovs&blais:monotonicity}
A.~Belovs and E.~Blais.
\newblock Quantum algorithm for monotonicity testing on the hypercube.
\newblock  \myhref{http://arxiv.org/abs/1503.02868}{{\ttfamily
  arXiv:1503.02868}}, 2015.

\bibitem{belovs:learningClaws}
A.~Belovs and B.~W. Reichardt.
\newblock Span programs and quantum algorithms for $st$-connectivity and claw
  detection.
\newblock In {\em Proc.\ of 20th ESA}, volume 7501 of {\em LNCS}, pages
  \myhref{http://dx.doi.org/10.1007/978-3-642-33090-2_18}{193--204}, 2012.
\newblock  \myhref{http://arxiv.org/abs/1203.2603}{{\ttfamily
  arXiv:1203.2603}}.

\bibitem{belovs:setEquality}
A.~Belovs and A.~Rosmanis.
\newblock Adversary lower bounds for the collision and the set equality
  problems.
\newblock  \myhref{http://arxiv.org/abs/1310.5185}{{\ttfamily
  arXiv:1310.5185}}, 2013.

\bibitem{bendavid:supergrover}
S.~Ben-David.
\newblock A super-{G}rover separation between randomized and quantum query
  complexities.
\newblock  \myhref{http://arxiv.org/abs/1506.081106}{{\ttfamily
  arXiv:1506.081106}}, 2015.

\bibitem{blais:testingJuntas}
E.~Blais.
\newblock Testing juntas nearly optimally.
\newblock In {\em Proc.\ of 41st ACM STOC}, pages
  \myhref{http://dx.doi.org/10.1145/1536414.1536437}{151--158}, 2009.

\bibitem{blais:testingJuntasSurvey}
E.~Blais.
\newblock Testing juntas: a brief survey.
\newblock In Goldreich \cite{goldreich:prop}, pages
  \myhref{http://dx.doi.org/10.1007/978-3-642-16367-8_4}{32--40}.

\bibitem{blais:testingLowerViaCommunication}
E.~Blais, J.~Brody, and K.~Matulef.
\newblock Property testing lower bounds via communication complexity.
\newblock {\em Computational Complexity},
  21(2):\myhref{http://dx.doi.org/10.1007/s00037-012-0040-x}{311--358}, 2012.
\newblock Earlier: \myhref{http://dx.doi.org/10.1109/CCC.2011.31}{{\em
  CCC'11}},  \myhref{http://eccc.hpi-web.de/report/2011/045}{{\ttfamily
  ECCC:2011/045}}.

\bibitem{boyer:groverTight}
M.~Boyer, G.~Brassard, P.~H{\o}yer, and A.~Tapp.
\newblock Tight bounds on quantum searching.
\newblock {\em Fortschritte der Physik}, 46(4-5):493--505, 1998.
\newblock  \myhref{http://arxiv.org/abs/quant-ph/9605034}{{\ttfamily
  arXiv:quant-ph/9605034}}.

\bibitem{brassard:amplification}
G.~Brassard, P.~H{\o}yer, M.~Mosca, and A.~Tapp.
\newblock Quantum amplitude amplification and estimation.
\newblock In {\em Quantum Computation and Quantum Information: A Millennium
  Volume}, volume 305 of {\em AMS Contemporary Mathematics Series}, pages
  53--74, 2002.
\newblock  \myhref{http://arxiv.org/abs/quant-ph/0005055}{{\ttfamily
  arXiv:quant-ph/0005055}}.

\bibitem{brassard:collision}
G.~Brassard, P.~H{\o}yer, and A.~Tapp.
\newblock Quantum cryptanalysis of hash and claw-free functions.
\newblock In {\em Proc.\ of 3rd LATIN}, volume 1380 of {\em LNCS}, pages
  \myhref{http://dx.doi.org/10.1007/BFb0054319}{163--169}. Springer, 1998.
\newblock  \myhref{http://arxiv.org/abs/quant-ph/9705002}{{\ttfamily
  arXiv:quant-ph/9705002}}.

\bibitem{buhrman:querySurvey}
H.~Buhrman and R.~de~Wolf.
\newblock Complexity measures and decision tree complexity: a survey.
\newblock {\em Theoretical Computer Science},
  288:\myhref{http://dx.doi.org/10.1016/S0304-3975(01)00144-X}{21--43}, 2002.

\bibitem{ChakrabartiR10}
A.~Chakrabarti and O.~Regev.
\newblock An optimal lower bound on the communication complexity of gap
  {H}amming distance.
\newblock {\em SIAM Journal on Computing},
  41(5):\myhref{http://dx.doi.org/10.1137/120861072}{1299--1317}, 2012.
\newblock Earlier: \myhref{http://dx.doi.org/10.1145/1993636.1993644}{{\em
  STOC'11}},  \myhref{http://arxiv.org/abs/1009.3460}{{\ttfamily
  arXiv:1009.3460}}.

\bibitem{Cheng11}
Y.~Cheng.
\newblock An efficient randomized group testing procedure to determine the
  number of defectives.
\newblock {\em Operations Research Letters},
  39(5):\myhref{http://dx.doi.org/10.1016/j.orl.2011.07.001}{352--354}, 2011.

\bibitem{ChengX14}
Y.~Cheng and Y.~Xu.
\newblock An efficient {FPRAS} type group testing procedure to approximate the
  number of defectives.
\newblock {\em Journal of Combinatorial Optimization},
  27(2):\myhref{http://dx.doi.org/10.1007/s10878-012-9516-5}{302--314}, 2014.

\bibitem{chockler:testingJuntasLower}
H.~Chockler and D.~Gutfreund.
\newblock A lower bound for testing juntas.
\newblock {\em Information Processing Letters},
  90(6):\myhref{http://dx.doi.org/10.1016/j.ipl.2004.01.023}{301--305}, 2004.

\bibitem{cleve:phaseEstimation}
R.~Cleve, A.~Ekert, C.~Macchiavello, and M.~Mosca.
\newblock Quantum algorithms revisited.
\newblock {\em Proceedings of the Royal Society of London A: Mathematical,
  Physical and Engineering Sciences},
  454(1969):\myhref{http://dx.doi.org/10.1098/rspa.1998.0164}{339--354}, 1998.
\newblock  \myhref{http://arxiv.org/abs/quant-ph/9708016}{{\ttfamily
  arXiv:quant-ph/9708016}}.

\bibitem{vanDam:exponentialCongruences}
W.~van Dam and I.~E. Shparlinski.
\newblock Classical and quantum algorithms for exponential congruences.
\newblock In {\em Proc.\ of 3rd TQC}, volume 5106 of {\em LNCS}, pages
  \myhref{http://dx.doi.org/10.1007/978-3-540-89304-2_1}{1--10}. Springer,
  2008.
\newblock  \myhref{http://arxiv.org/abs/0804.1109}{{\ttfamily
  arXiv:0804.1109}}.

\bibitem{DamaschkeM10}
P.~Damaschke and A.~S. Muhammad.
\newblock Bounds for nonadaptive group tests to estimate the amount of
  defectives.
\newblock In {\em Proc.\ of 4th COCOA}, volume 6509 of {\em LNCS}, pages
  \myhref{http://dx.doi.org/10.1007/978-3-642-17461-2_10}{117--130}. Springer,
  2010.

\bibitem{dorfman:grouptesting}
R.~Dorfman.
\newblock The detection of defective members of large populations.
\newblock {\em The Annals of Mathematical Statistics}, 14(4):436--440, 1943.

\bibitem{du:combinatorialGroupTesting}
D.~Z. Du and F.~Hwang.
\newblock {\em Combinatorial group testing and its applications}, volume~3 of
  {\em Series on Applied Mathematics}.
\newblock World Scientific, 1993.

\bibitem{fischer:testingJuntas}
E.~Fischer, G.~Kindler, D.~Ron, S.~Safra, and A.~Samorodnitsky.
\newblock Testing juntas.
\newblock {\em Journal of Computer and System Sciences},
  68(4):\myhref{http://dx.doi.org/10.1016/j.jcss.2003.11.004}{753--787}, 2004.
\newblock Earlier: \myhref{http://dx.doi.org/10.1109/SFCS.2002.1181887}{{\em
  FOCS'02}}.

\bibitem{garciasoriano:phd}
D.~{Garc\'{i}a-Soriano}.
\newblock {\em Query-Efficient Computation in Property Testing and Learning
  Theory}.
\newblock PhD thesis, CWI and University of Amsterdam, 2012.

\bibitem{goldreich:prop}
O.~Goldreich, editor.
\newblock {\em Property Testing: Current Research and Surveys}, volume 6390 of
  {\em LNCS}.
\newblock Springer, 2010.

\bibitem{grover:search}
L.~K. Grover.
\newblock A fast quantum mechanical algorithm for database search.
\newblock In {\em Proc.\ of 28th ACM STOC}, pages
  \myhref{http://dx.doi.org/10.1145/237814.237866}{212--219}, 1996.
\newblock  \myhref{http://arxiv.org/abs/quant-ph/9605043}{{\ttfamily
  arXiv:quant-ph/9605043}}.

\bibitem{hoyer:advNegative}
P.~H{\o}yer, T.~Lee, and R.~{\v Spalek}.
\newblock Negative weights make adversaries stronger.
\newblock In {\em Proc.\ of 39th ACM STOC}, pages
  \myhref{http://dx.doi.org/10.1145/1250790.1250867}{526--535}, 2007.
\newblock  \myhref{http://arxiv.org/abs/quant-ph/0611054}{{\ttfamily
  arXiv:quant-ph/0611054}}.

\bibitem{iwama:quantumCounterfeit}
K.~Iwama, H.~Nishimura, R.~Raymond, and J.~Teruyama.
\newblock Quantum counterfeit coin problems.
\newblock {\em Theoretical Computer Science},
  456:\myhref{http://dx.doi.org/10.1016/j.tcs.2012.05.039}{51--64}, 2012.
\newblock Earlier: \myhref{http://dx.doi.org/10.1007/978-3-642-17517-6_10}{{\em
  ISAAC'10}},  \myhref{http://arxiv.org/abs/1009.0416}{{\ttfamily
  arXiv:1009.0416}}.

\bibitem{kaas:MedianBinomial}
R.~Kaas and J.~Buhrman.
\newblock Mean, median and mode in binomial distributions.
\newblock {\em Statistica Neerlandica},
  34(1):\myhref{http://dx.doi.org/10.1111/j.1467-9574.1980.tb00681.x}{13--18},
  1980.

\bibitem{KawanoS13}
Y.~Kawano and H.~Sekigawa.
\newblock Quantum {F}ourier transform over symmetric groups.
\newblock In {\em Proc.\ of 38th ISAAC}, pages
  \myhref{http://dx.doi.org/10.1145/2465506.2465940}{227--234}, 2013.

\bibitem{KawanoS14}
Y.~Kawano and H.~Sekigawa.
\newblock Quantum {F}ourier transform over symmetric groups: {I}mproved result.
\newblock {\em ACM Communications in Computer Algebra},
  48(3):\myhref{http://dx.doi.org/10.1145/2733693.2733708}{127--129}, 2014.

\bibitem{kitaev:phaseEstimation}
A.~Kitaev.
\newblock Quantum measurements and the {A}belian stabilizer problem.
\newblock  \myhref{http://arxiv.org/abs/quant-ph/9511026}{{\ttfamily
  arXiv:quant-ph/9511026}}, 1995.

\bibitem{kushilevitz&nisan:cc}
E.~Kushilevitz and N.~Nisan.
\newblock {\em Communication Complexity}.
\newblock Cambridge University Press, 1997.

\bibitem{lee:stateConversion}
T.~Lee, R.~Mittal, B.~W. Reichardt, R.~{\v Spalek}, and M.~Szegedy.
\newblock Quantum query complexity of state conversion.
\newblock In {\em Proc.\ of 52nd IEEE FOCS}, pages
  \myhref{http://dx.doi.org/10.1109/FOCS.2011.75}{344--353}, 2011.
\newblock  \myhref{http://arxiv.org/abs/1011.3020}{{\ttfamily
  arXiv:1011.3020}}.

\bibitem{montanaro:quantumProperyTest}
A.~Montanaro and R.~de~Wolf.
\newblock A survey of quantum property testing.
\newblock {\em Theory of Computing}, 2015.
\newblock  \myhref{http://arxiv.org/abs/1310.2035}{{\ttfamily
  arXiv:1310.2035}}.
\newblock To appear.

\bibitem{mrr:genericqft}
C.~Moore, D.~Rockmore, and A.~Russell.
\newblock Generic quantum {F}ourier transforms.
\newblock {\em ACM Trans. Algorithms},
  2(4):\myhref{http://dx.doi.org/10.1145/1198513.1198525}{707--723}, 2006.
\newblock Earlier: {\em SODA'04},
  \myhref{http://arxiv.org/abs/quant-ph/0304064}{{\ttfamily
  arXiv:quant-ph/0304064}}.

\bibitem{raskhodnikova:approximatingSupportSize}
S.~Raskhodnikova, D.~Ron, A.~Shpilka, and A.~Smith.
\newblock Strong lower bounds for approximating distribution support size and
  the distinct elements problem.
\newblock {\em SIAM Journal on Computing},
  39(3):\myhref{http://dx.doi.org/10.1137/070701649}{813--842}, 2009.
\newblock Earlier: \myhref{http://dx.doi.org/10.1109/FOCS.2007.47}{{\em
  FOCS'07}}.

\bibitem{reichardt:spanPrograms}
B.~W. Reichardt.
\newblock Span programs and quantum query complexity: The general adversary
  bound is nearly tight for every boolean function.
\newblock In {\em Proc.\ of 50th IEEE FOCS}, pages
  \myhref{http://dx.doi.org/10.1109/FOCS.2009.55}{544--551}, 2009.
\newblock  \myhref{http://arxiv.org/abs/0904.2759}{{\ttfamily
  arXiv:0904.2759}}.

\bibitem{reichardt:advTight}
B.~W. Reichardt.
\newblock Reflections for quantum query algorithms.
\newblock In {\em Proc.\ of 22nd ACM-SIAM SODA}, pages
  \myhref{http://dx.doi.org/10.1137/1.9781611973082.44}{560--569}, 2011.
\newblock  \myhref{http://arxiv.org/abs/1005.1601}{{\ttfamily
  arXiv:1005.1601}}.

\bibitem{reichardt:formulae}
B.~W. Reichardt and R.~{\v Spalek}.
\newblock Span-program-based quantum algorithm for evaluating formulas.
\newblock {\em Theory of Computing},
  8:\myhref{http://dx.doi.org/10.4086/toc.2012.v008a013}{291--319}, 2012.
\newblock Earlier: \myhref{http://dx.doi.org/10.1145/1374376.1374394}{{\em
  STOC'08}},  \myhref{http://arxiv.org/abs/0710.2630}{{\ttfamily
  arXiv:0710.2630}}.

\bibitem{sagan:symmetricGroup}
B.~E. Sagan.
\newblock {\em The symmetric group: representations, combinatorial algorithms,
  and symmetric functions}, volume 203 of {\em Graduate Texts in Mathematics}.
\newblock Springer, 2001.

\bibitem{stw:adaptivityhelps}
R.~A. Servedio, L.-Y. Tang, and J.~Wright.
\newblock Adaptivity helps for testing juntas.
\newblock In {\em Proc.\ of 30th IEEE CCC}, volume~33 of {\em LIPIcs}, pages
  \myhref{http://dx.doi.org/10.4230/LIPIcs.CCC.2015.264}{264--279}, 2015.

\bibitem{Sherstov12}
A.~A. Sherstov.
\newblock The communication complexity of gap {H}amming distance.
\newblock {\em Theory of Computing},
  8(8):\myhref{http://dx.doi.org/10.4086/toc.2012.v008a008}{197--208}, 2012.
\newblock  \myhref{http://eccc.hpi-web.de/report/2011/063}{{\ttfamily
  ECCC:2011/063}}.

\bibitem{slud:inequality}
E.~V. Slud.
\newblock Distribution inequalities for the binomial law.
\newblock {\em The Annals of Probability},
  5(3):\myhref{http://dx.doi.org/doi:10.1214/aop/1176995801}{404--412}, 1977.

\bibitem{zhandry:howToConstruct}
M.~Zhandry.
\newblock How to construct quantum random functions.
\newblock In {\em Proc.\ of 53rd IEEE FOCS}, pages
  \myhref{http://dx.doi.org/10.1109/FOCS.2012.37}{679--687}, 2012.
\newblock  \myhref{http://eprint.iacr.org/2012/182}{{\ttfamily
  ePrint:2012/182}}.

\bibitem{zhandry:setEquality}
M.~Zhandry.
\newblock A note on the quantum collision and set equality problems.
\newblock {\em Quantum Information \& Computation}, 15(7\&8):557--567, 2015.
\newblock  \myhref{http://arxiv.org/abs/1312.1027}{{\ttfamily
  arXiv:1312.1027}}.

\end{thebibliography}

\appendix

\section{Proof of \rf(prp:irrelevant)(b)}
\label{app:irrelevant}

The proof of this result is the same as the proof of \rf(thm:composition).  We use the following easy properties of semi-definite matrices.

\begin{prp}
\label{prp:sdm}
Let $A$ and $B$ be positive semi-definite matrices of the same size.  Then, the following matrices are also positive semi-definite:
\itemstart
\item[(a)] obtained from $A$ by simultaneously duplicating one of its rows and the same column;
\item[(b)] the Hadamard (entry-wise) product $A\circ B$.
\itemend
\end{prp}

First, trivially,
\[
\begin{pmatrix}
1&0&0&1\\
0&1&1&0\\
0&1&1&0\\
1&0&0&1\\
\end{pmatrix} \succeq 0.
\]
For $j\in[n]$, use \rf(prp:sdm)(a) to duplicate rows and columns of this matrix, so that its rows correspond to the following four groups of inputs: $z\in\fin$, $z_j=0$; $z\in\fin$, $z_j=1$; $z\in\fip$, $z_j=0$; and $z\in\fip$, $z_j=1$.  We can take the Hadamard product of this matrix and $X_j$, so we may assume without loss of generality that, for all $x,y$ satisfying $F(x)\ne F(y)$, instead of~\rf(eqn:advOrigCondition) we have the following stronger condition:
\begin{equation}
\label{eqn:newCondition}
X_j\elem[x,y]=0\quad\text{if}\quad x_j = y_j\;,\qquad\text{and}\qquad \sum_{j\in [n]} X_j\elem[x,y]=1.
\end{equation}

Let $X^{(j)}_i$ be an optimal solution to~\rf(eqn:advOrig) for $G_j$.  Then, we define a feasible solution $\tX_{ji}$ to~\rf(eqn:advOrig) for the composed function $F\circ(G_1,\dots,G_n)$ in the following way.  

Let 
\(
x=(x_{11},\dots,x_{1m_1},\;\dots\;,x_{n1},\dots,x_{nm_n})
\)
be an input in the domain of $F\circ(G_1,\dots,G_n)$.  Thus, there exists $z$ in the domain of $f$ such that $z_j = G_j(x_{j1},\dots,x_{jm_j})$ for all relevant $j$.  Fix this choice of $z$ to $x$.  Assume that $y=(y_{11},\dots,y_{1m_1},\;\dots\;,y_{n1},\dots,y_{nm_n})$ and $t$ are defined analogously to $x$ and $z$, respectively. Then, define
\[
\tX_{ji}\elem[x, y] = X_j\elem[z,t]X^{(j)}_i\elem[(x_{j1},\dots,x_{jm_j}),(y_{j1},\dots,y_{jm_j})],
\]
where we extend $X^{(j)}_i$ with zeroes for all inputs outside the domain of $G_j$.

First, $\tX_{ji}\succeq 0$ due to \rf(prp:sdm).  Next, if $F(z)\ne F(t)$, then
\[
\sum_{ji\colon x_{ji}\ne y_{ji}} \tX_{ji}\elem[x, y]=
\sum_{j\in [n]} X_j\elem[z,t] \sum_{i\colon x_{ji}\ne y_{ji}} X_i^{(j)}\elem[(x_{j1},\dots,x_{jm_j}),(y_{j1},\dots,y_{jm_j})] = \sum_{j\in [n]} X_j\elem[z,t] = 1.
\]
Here we used that if $X_j\elem[z,t]\ne 0$, then $j$ is relevant for both $z$ and $t$, and $G_j(x_{j1},\ldots,x_{jm_j}) \ne G_j(y_{j1},\ldots,y_{jm_j})$ because of~\rf(eqn:newCondition).
Finally, we have
\[
\sum_{ji} \tX_{ji}\elem[x,x] = \sum_{j\in [n]} X_j\elem[z,z] \sum_{i} X_i^{(j)}\elem[(x_{j1},\dots,x_{jm_j}),(x_{j1},\dots,x_{jm_j})] \le Q\; \max_{j\in[n]} \Adv(G_j).
\]

\end{document}